\documentclass[
    a4paper, 
    fontsize=10pt, 
    twoside=true, 
	numbers=noenddot, 
]{kaobook}

\usepackage{tikz}

\usepackage{blindtext} 

\definecolor{codebg}{HTML}{E9E9E9}
\usepackage{siunitx} 
\usepackage{braket} 

\renewcommand{\marginlayout}{%
	\newgeometry{
		top=27.4mm,				
		bottom=27.4mm,			
		inner=24.8mm,			
		textwidth=117mm,		
		marginparsep=6.2mm,		
		marginparwidth=41.4mm,	
	}%
}

\usetikzlibrary{arrows}
\usetikzlibrary{fit}
\usetikzlibrary{positioning}

\usepackage{tikzit}

\tikzstyle{Z dot}=[inner sep=0mm, minimum size=2mm, shape=circle, draw=black, fill={rgb,255: red,221; green,255; blue,221}, tikzit category=zx]
\tikzstyle{Z phase dot}=[minimum size=1.2em, font={\footnotesize\boldmath}, shape=rectangle, rounded corners=0.5em, inner sep=0.2em, outer sep=-0.2em, scale=0.8, draw=black, fill={rgb,255: red,221; green,255; blue,221}, tikzit shape=circle, tikzit draw=blue, tikzit category=zx]
\tikzstyle{Z box}=[Z phase dot, rounded corners=0, fill={rgb,255: red,221; green,255; blue,221}, tikzit shape=rectangle, tikzit draw=blue, tikzit category=zx]
\tikzstyle{X dot}=[Z dot, shape=circle, draw=black, fill={rgb,255: red,232; green,165; blue,165}, tikzit category=zx]
\tikzstyle{X phase dot}=[Z phase dot, tikzit shape=circle, tikzit draw=blue, fill={rgb,255: red,232; green,165; blue,165}, font={\footnotesize\boldmath}, tikzit category=zx]
\tikzstyle{X box}=[Z phase dot, rounded corners=0, fill={rgb,255: red,232; green,165; blue,165}, tikzit shape=rectangle, tikzit draw=blue, tikzit category=zx]
\tikzstyle{H}=[fill=yellow, draw=black, shape=rectangle, inner sep=0.6mm, minimum height=1.5mm, minimum width=1.5mm]
\tikzstyle{hadamard}=[fill=yellow, draw=black, shape=rectangle, inner sep=0.6mm, minimum height=1.5mm, minimum width=1.5mm, tikzit category=zx]
\tikzstyle{u_triang}=[fill=yellow, draw=black, shape=isosceles triangle, shape border rotate=90, isosceles triangle stretches=true, inner sep=0.8pt, minimum width=0.25cm, minimum height=2mm]
\tikzstyle{d_triang}=[fill=yellow, draw=black, shape=isosceles triangle, shape border rotate=-90, isosceles triangle stretches=true, inner sep=0.8pt, minimum width=0.25cm, minimum height=2mm]
\tikzstyle{r_triang}=[fill=yellow, draw=black, shape=isosceles triangle, shape border rotate=0, isosceles triangle stretches=true, inner sep=0.8pt, minimum width=0.25cm, minimum height=2mm]
\tikzstyle{l_triang}=[fill=yellow, draw=black, shape=isosceles triangle, shape border rotate=180, isosceles triangle stretches=true, inner sep=0.8pt, minimum width=0.25cm, minimum height=2mm]
\tikzstyle{h_star}=[fill=yellow, draw=black, shape=star, star points=6, star point ratio=1.74, inner sep=0.8pt, minimum height=2.75mm, rotate=90]
\tikzstyle{v_star}=[fill=yellow, draw=black, shape=star, star points=6, star point ratio=1.74, inner sep=0.8pt, minimum height=2.75mm]
\tikzstyle{u_black_triang}=[fill=black, draw=black, scale=1, shape=isosceles triangle, shape border rotate=90, isosceles triangle stretches=true, inner sep=1pt, minimum width=0.4cm, minimum height=3mm]
\tikzstyle{d_black_triang}=[fill=black, draw=black, scale=1, shape=isosceles triangle, shape border rotate=-90, isosceles triangle stretches=true, inner sep=1pt, minimum width=0.4cm, minimum height=3mm]
\tikzstyle{r_black_triang}=[fill=black, draw=black, scale=1, shape=isosceles triangle, shape border rotate=0, isosceles triangle stretches=true, inner sep=1pt, minimum width=0.4cm, minimum height=3mm]
\tikzstyle{l_black_triang}=[fill=black, draw=black, scale=1, shape=isosceles triangle, shape border rotate=180, isosceles triangle stretches=true, inner sep=1pt, minimum width=0.4cm, minimum height=3mm]
\tikzstyle{paulibox}=[fill={rgb,255: red,221; green,221; blue,255}, draw=black, shape=rectangle, inner sep=0.6mm, minimum height=5mm, minimum width=5mm, font={\footnotesize}, text height=1.5ex, text depth=0.25ex, tikzit category=zx]
\tikzstyle{gate}=[shape=rectangle, fill=white, draw=black, minimum height=4mm, minimum width=4mm, text centered, inner sep=0mm, line width=0.64pt, font={\small}, tikzit category=circuit]
\tikzstyle{black dot}=[fill=black, draw=black, shape=circle, inner sep=0pt, minimum width=1.2mm, tikzit category=circuit]
\tikzstyle{cnot targ}=[fill=white, draw=white, shape=circle, tikzit category=circuit, label={center:$\oplus$}, inner sep=0pt, minimum width=2.1mm, tikzit fill={rgb,255: red,102; green,204; blue,255}, tikzit draw=black]
\tikzstyle{cross}=[path picture={ 
    \draw[black] (path picture bounding box.south west) -- (path picture bounding box.north east);
    \draw[black] (path picture bounding box.south east) -- (path picture bounding box.north west);
}, shape=circle, minimum size=2mm, inner sep=0pt, outer sep=0pt, tikzit category=circuit]
\tikzstyle{q_measurement}=[draw, rectangle, shape=rectangle, fill=white, inner sep=0.8mm, minimum width=6.4mm, minimum height=4.8mm, line width=0.64pt, scale=0.8, path picture={ 
        \draw[black] 
            ([shift={(.1,.12)}]path picture bounding box.south west) 
            to[bend left=50] 
            ([shift={(-.1,.12)}]path picture bounding box.south east); 
        \draw[black,-latex] 
            ([shift={(0,.08)}]path picture bounding box.south) 
            -- ([shift={(.24,-.072)}]path picture bounding box.north); 
    }, tikzit category=circuit]
\tikzstyle{green_dot}=[fill={rgb,255: red,61; green,161; blue,0}, draw={rgb,255: red,61; green,161; blue,0}, shape=circle, tikzit fill={rgb,255: red,61; green,161; blue,0}, tikzit draw={rgb,255: red,61; green,161; blue,0}, inner sep=0.5mm]

\tikzstyle{hadamard edge}=[-, dashed, dash pattern=on 2pt off 0.5pt, thick, draw={rgb,255: red,68; green,136; blue,255}]
\tikzstyle{star edge}=[-, dashed, dash pattern=on 2pt off 0.5pt, thick, draw={rgb,255: red,255; green,136; blue,68}]
\tikzstyle{box edge}=[-, dashed, dash pattern=on 2pt off 0.5pt, thick, draw={rgb,255: red,203; green,192; blue,225}]
\tikzstyle{brace edge}=[-, tikzit draw=blue, decorate, decoration={brace,amplitude=1mm,raise=-1mm}]
\tikzstyle{diredge}=[->]
\tikzstyle{double edge}=[-, double, shorten <=-1mm, shorten >=-1mm, double distance=2pt]
\tikzstyle{boldedge}=[-, line width=1.6pt, shorten <=-0.17mm, shorten >=-0.17mm]
\tikzstyle{pink_highlight}=[-, fill={rgb,255: red,243; green,234; blue,255}, draw={rgb,255: red,188; green,181; blue,197}, tikzit fill={rgb,255: red,243; green,234; blue,255}, tikzit draw={rgb,255: red,188; green,181; blue,197}]
\tikzstyle{green_highlight}=[-, fill={rgb,255: red,228; green,255; blue,225}, draw={rgb,255: red,159; green,178; blue,157}, tikzit fill={rgb,255: red,228; green,255; blue,225}, tikzit draw={rgb,255: red,159; green,178; blue,157}]
\tikzstyle{blue_highlight}=[-, fill={rgb,255: red,225; green,250; blue,255}, tikzit fill={rgb,255: red,225; green,250; blue,255}, draw={rgb,255: red,195; green,217; blue,221}, tikzit draw={rgb,255: red,195; green,217; blue,221}]
\tikzstyle{yellow_highlight}=[-, fill={rgb,255: red,253; green,255; blue,228}, tikzit fill={rgb,255: red,253; green,255; blue,228}, draw={rgb,255: red,220; green,221; blue,198}, tikzit draw={rgb,255: red,220; green,221; blue,198}]
\tikzstyle{red_highlight}=[-, fill={rgb,255: red,255; green,211; blue,211}, tikzit fill={rgb,255: red,255; green,211; blue,211}, draw={rgb,255: red,179; green,3; blue,6}, tikzit draw={rgb,255: red,179; green,3; blue,6}]
\tikzstyle{blue_line}=[-, fill=none, draw={rgb,255: red,17; green,0; blue,255}, tikzit draw={rgb,255: red,17; green,0; blue,255}]
\tikzstyle{red_line}=[-, draw={rgb,255: red,255; green,0; blue,12}, tikzit draw={rgb,255: red,255; green,0; blue,12}]

\input{my_tikzit_style.tikzdefs}

\usepackage{import}
\usepackage{xifthen}
\usepackage{pdfpages}
\usepackage{transparent}

\ifxetexorluatex
	\usepackage{polyglossia}
	\setmainlanguage{english}
\else
	\usepackage[english]{babel} 
\fi
\usepackage[english=british]{csquotes}	

\usepackage{blindtext}

\usepackage{kaobiblio}
\usepackage[framed=true]{kaotheorems}

\usepackage{kaorefs}

\graphicspath{{examples/documentation/images/}{images/}} 

\usepackage{orcidlink}

\begin{document}


\pagelayout{wide} 
\begin{titlepage}
	\begin{center}
		\vspace*{1cm}
 
		\fontsize{20pt}{20pt}\selectfont
		\textbf{Graphical Stabilizer Decompositions for Multi-Control Toffoli Gate Dense Quantum Circuits}
 
		\vspace{0.5cm}
		\LARGE
		Maturitätsarbeit
			 
		\vspace{1.5cm}
 
		\textbf{Yves Vollmeier \orcidlink{0009-0003-5164-6848}}

		\vspace{4.75cm}
	  
		\begin{figure}[ht]
			\centering
    \def\svgwidth{\textwidth}
    \import{./figures/}{Title_Page/title_page.pdf_tex}

		\end{figure}
		
		\vspace{4.75cm}

		\large
		Supervised by Dr. Riccardo Ferrario\\
		Mathematisch Naturwissenschaftliches Gymnasium Rämibühl\\
		January 6, 2025
			 
	\end{center}
\end{titlepage}


\chapter*{Abstract}
\addcontentsline{toc}{chapter}{Abstract} 

In this thesis, we study concepts in quantum computing using graphical languages, specifically using the ZX-calculus. The core of the research revolves around
(graphical) stabilizer decompositions. The first major focus is on the decomposition of non-stabilizer states created from star edges. We discuss previous results
and then present novel decompositions that yield a theoretical improvement. The second major focus is on weighting algorithms, applied to the special class of
multi-control Toffoli gate dense quantum circuits. The representation of the corresponding gates is based on star edges. The applicability of known methods, such
as CNOT-grouping, traditionally used for other classes, is examined in the context of this specific class. We then present a novel weighting algorithm that
attempts to determine the best vertex to decompose. A refined version is implemented to simulate a known class of quantum querying algorithms, which is used to
search for causal configurations of multiloop Feynman diagrams. For this case, as well as for a generalized benchmark consisting of randomly generated quantum circuits,
we demonstrate occasional improvements in the final number of terms against traditional methods. These results are discussed by considering different simplification
strategies. This thesis also provides a brief but broad outline of the important preliminaries.

\begin{flushright}
	\textit{Yves Vollmeier}
\end{flushright}

\index{abstract}


\begingroup 

\setstretch{1.0} 
\setlength{\textheight}{230\hscale} 

\etocstandarddisplaystyle 
\etocstandardlines 

\tableofcontents 

\endgroup


\mainmatter 
\setchapterstyle{kao} 

\setchapterpreamble[u]{\margintoc}
\chapter{Introduction}
\labch{intro}

\begin{quote}
    \small \textit{Yeah, but those things don’t mean anything to me. [...] and I see in the Physics Review these idiot diagrams I cooked up [...]} \\
    \small --- \textup{Richard Feynman}
\end{quote}

\section{Motivation}
\labsec{motivation}

In 1966, Richard Feynman was asked about the awards he had received for his work. This provoked his annoyance, as he is often referred to as a humble
physicist. We can only suspect that this annoyance was the reason for why he referred to \textit{Feynman diagrams} as "idiot diagrams"
\cite{physicsRichardFeynmanSession2015}. The reality stands in stark contrast: Feynman diagrams are said to have revolutionized nearly every
aspect of theoretical physics \cite{kaiserPhysicsFeynmansDiagrams2005}. It was therefore a pleasant surprise when we found a graphical language that was
reminiscent of these diagrams, and did not seem to require advanced knowledge about quantum field theory. We found out about this graphical language, which
is referred to as the \textit{ZX-calculus}, while learning about quantum mechanics and quantum computation and looking through research papers.

The initial goal set out for this \textbf{Maturitätsarbeit} was not only to learn about the ZX-calculus as a tool, but also to gain insight into the research
field as a whole, as it was only introduced in 2008 \cite{coeckeInteractingQuantumObservables2008}. This learning should
be accompanied by active experimentation with the newly learned material.

The journey proved to be very fruitful and diverse, learning about many aspects of the ZX-calculus and related topics. We also got to experience scientific
research in general. In order to create a common theme throughout this thesis and in order to remain brief, it was decided to only focus on
so-called \textit{stabilizer decompositions}. This was the topic we found particularly interesting and where we were able to find novel insights.

Stabilizer decompositions have first been introduced in 2012 \cite{garciaEfficientInnerproductAlgorithm2013} and have later become an important research direction
within the ZX-calculus \cite{kissingerClassicalSimulationQuantum2022}.

\begin{figure}[!ht]
    \def\svgwidth{\textwidth}
    \import{./figures/}{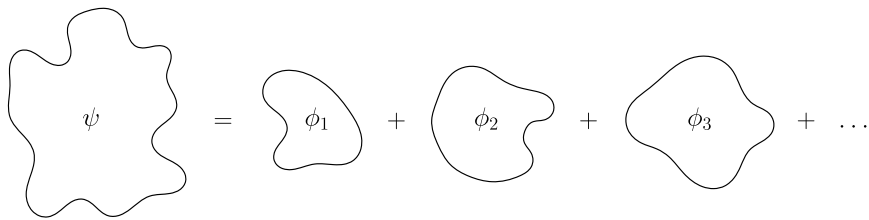}

	\caption[Simplistic description of stabilizer decompositions]{Simplified illustration of stabilizer decompositions.}
	\labfig{simplistic_description_stab_decomps}
\end{figure}

The basic idea of stabilizer decompositions is illustrated in \reffig{simplistic_description_stab_decomps}. We have a mathematical object $\psi$ that
corresponds to a structure used in quantum computation, represented using the ZX-calculus. More precisely, it is a so-called \textit{non-stabilizer state}, and
has certain properties that make it difficult to simulate. The idea is now to \textit{decompose} it into a sum 
of \textit{stabilizer states} $\phi_1, \phi_2, \phi_3, \dots$, whose properties make their simulation easier.

\section{Structure of This Thesis}

This thesis consists of three parts:
\begin{itemize}
	\item \myrefpart{part_i}
	\item \myrefpart{part_ii}
	\item \myrefpart{part_iii}
\end{itemize}

The first part aims to provide important preliminaries for the rest of this thesis, although notably, these will not be enough to get an in-depth understanding
of all the concepts used throughout this thesis. Interested reader are therefore encouraged to consult the referenced sources for further study. Conversely,
experienced readers may skip the first part and go directly to the second and third part.

Part I consists of \refch{quantum_computation}, which introduces the basic ideas of \textit{quantum mechanics}. This knowledge is then linked to
\textit{quantum computation}. Lastly, this chapter contains a brief overview of \textit{computational complexity theory}. In \refch{graphical_languages},
we give a general motivation for the use of \textit{graphical languages}, as well as a basic introduction to the \textit{ZX-calculus} and related calculi. Readers specifically interested in the applications of the ZX-calculus can find an extensive list of previous and current research topics compiled in \refsec{applications}.

Part II focuses on \textit{state decompositions}. We start by reviewing previous work from the literature. Before showing our novel results, we also outline a few of our
previous attempts.

Part III starts with \refch{dynamic_decompositions}, where we very briefly describe decompositions of \textit{matrices} instead of \textit{vectors}, that are
furthermore \textit{dynamic} in their size. Hence, we will refer to them as \textit{dynamic decompositions}. These are then also compared to state decompositions
in \refsec{comparison_with_state_decomps}.

In \refch{weighting_algorithms}, we build on these ideas, together with ideas from the literature, in order to create a \textit{weighting algorithm} meant to
decompose multi-control Toffoli gate dense quantum circuits with certain constraints as efficiently as possible. After very briefly explaining the
\textit{loop-tree duality formalism}, we demonstrate how our algorithm can be used to simulate a \textit{quantum search algorithm}, which is used to find
\textit{causal configurations} of \textit{multiloop Feynman diagrams}. Finally, the results are compared with the literature and discussed using further benchmarks.

The conclusion of this work features a summary of the research results and an outlook on future work.

\pagelayout{wide} 
\addpart{Part I: Preliminaries}\labpart{part_i}
\pagelayout{margin} 

\setchapterpreamble[u]{\margintoc}
\chapter{Quantum Computation}
\labch{quantum_computation}

\begin{quote}
    \small \textit{God does not play dice} \\
    \small --- \textup{Albert Einstein}
\end{quote}

\begin{quote}
    \small \textit{Einstein, stop telling God what do!} \\
    \small --- \textup{Niels Bohr}
\end{quote}

\begin{marginfigure}[*2]
    \includegraphics{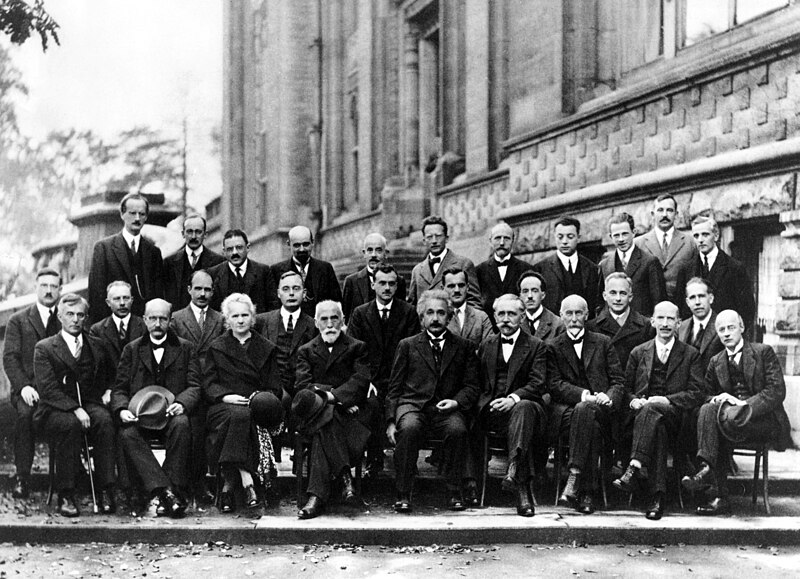}
    \caption[The Solvay Conference]{The Solvay Conference \cite{Solvay_conference_1927}.}
    \labfig{margin_solvay}
\end{marginfigure}

In 1927, perhaps one of the most iconic pictures in physics was taken (\cf \reffig{margin_solvay}). Seventeen of the 29 attendees
of the \textit{fifth Solvay Conference on Physics} went down in history by winning a Nobel Prize. But not only that, many of them were
integral to the development of \textbf{Quantum Mechanics} \cite{SolvayConference2024}. It is arguably one of the most successful, yet
also one of the most mysterious scientific theories humanity has ever come up with \cite{nielsenQuantumComputationQuantum2012}.

In 1900, Max Planck postulated that electromagnetic energy is quantized, i.e. comes in discrete bundles of
energy, which he called \textit{quanta}\sidenote[][*3]{Hence the name \textit{Quantum Mechanics}.}.
The energy of a \textit{single} quanta could now be described by the equation
\begin{equation*}
    E=h\nu
\end{equation*}
where $h=\SI{6.62610d-34}{\joule\second}$ is a fit parameter called \textit{Planck's constant}, and $\nu$ is the frequency of the electromagnetic wave.
In 1905, based on Planck's work, Albert Einstein proposed the idea that light itself is made up of discrete quanta of energy\sidenote[][*-1]{This ultimately led to the notion of \textit{photons}.}.
It is worth noting that these efforts originated from the \textit{ultraviolet catastrophe}, an inconsistency between experimental data and \textit{classical} theoretical predictions,
which could be resolved using Planck's new ideas \cite{UltravioletCatastrophe2024}.

During the next two decades, there was an active back and forth between new discoveries\sidenote[][*-1]{Discovery of the atomic nucleus, compton scattering, electron diffraction, \dots} and new theories\sidenote[][*2]{Postulation of wave-particle duality, electron spin, description of wave mechanics, matrix mechanics, \dots}. This culmination
of work ultimately led to what we now see as the origins of quantum mechanics \cite{182BriefSummary2017}.

Although this "first" period between 1900 and the 1920s provided profound insights into the natural world, it was only in the 1970s and 1980s that a shift
in perspective allowed a greater understanding of quantum mechanics itself. This shift in perspective was led by pioneers such as Paul Benioff,
Yuri Manin and Richard Feynman, who were now starting to think about \textit{designed} systems, as opposed to \textit{natural} systems. They realized
that classical means of computation may not be optimal to study quantum mechanical systems and were thus asking questions like "What would a better
suited machine look like?", "What is the space-time complexity of a certain quantum operation?", "How else can we exploit these properties?" and
many more. This was the emergence of a new, highly interdisciplinary field of research, combining questions from physics, computer science,
information theory and even engineering. It is what we now call \textit{Quantum Information Theory} and \textit{Quantum Computation} \cite{nielsenQuantumComputationQuantum2012}.

With these foundations in mind, \refsec{postulates_of_qm} will focus on the concepts from the "first period", whereas \refsec{qubits_qgates_and_qcircs} and \refsec{complexity_and_simulation}
will focus on the concepts from the "second period". These three sections outline the information necessary for the rest of the thesis. They also aim to make the reader familiar with
some key concepts from this fascinating field of research.

\section{Postulates of Quantum Mechanics}
\labsec{postulates_of_qm}

Postulates are of great importance in physics\sidenote{A famous example is that the speed of light $c$ is constant. It is a key requirement to make
special relativity work, whose predictions have been confirmed many times.}. They are requirements to make a certain theory work, even though
they are not provable. Still, one may consider the validity of the resulting predictions \cite{Postulate2024}.

Before starting with the \textit{postulates of quantum mechanics}, it is important to realize that the mathematical tools to describe
different quantum systems might vary. For example, in quantum computing, it suffices to use \textit{finite-dimensional Hilbert spaces}, whereas
the description of a hydrogen atom requires \textit{infinite-dimensional Hilbert spaces}. For the purpose of this discussion,
we will refrain from exploring the details of the latter\sidenote{The definition of infinite-dimensional Hilbert spaces requires additional
statements about limits and convergence.}, only presenting a few basic use cases \cite{KissingerWetering2024Book}.
\begin{definition}
    \labdef{hilbert_space}
    A \textit{finite-dimensional Hilbert space} $\mathcal{H}$ is a vector space over the complex numbers together with an \textit{inner product}, that is, a mapping
    \begin{equation*}
        \braket{\cdot|\cdot}:\mathcal{H}\times\mathcal{H}\to\mathbb{C}
    \end{equation*}
    such that the following conditions are satisfied:
    \begin{enumerate}
        \item \textit{Linearity in the second argument}: $\forall \psi,\phi_1,\phi_2\in\mathcal{H}$, $\forall \lambda_1,\lambda_2\in\mathbb{C}:$
              \begin{equation*}
                  \braket{\psi|\lambda_1\phi_1+\lambda_2\phi_2}=\lambda_1\braket{\psi|\phi_1}+\lambda_2\braket{\psi|\phi_2}.
              \end{equation*}
        \item \textit{Conjugate-symmetry}: $\forall \psi,\phi\in\mathcal{H}: \overline{\braket{\phi|\psi}}=\braket{\psi|\phi}$.
        \item \textit{Positive-definiteness}: $\forall \phi\in\mathcal{H}: \braket{\psi|\psi}\in\mathbb{R}, \ \braket{\psi|\psi}>0$.
    \end{enumerate}
\end{definition}
\begin{remark}
    By combining the first two conditions from \refdef{hilbert_space}, we find that the inner product is \textit{conjugate-linear in the
    first argument}, that means, $\forall \psi,\phi_1,\phi_2\in\mathcal{H}$, $\forall \lambda_1,\lambda_2\in\mathbb{C}:$
    \begin{equation*}
        \braket{\lambda_1\psi_1+\lambda_2\psi_2|\phi}=\overline{\lambda_1}\braket{\psi_1|\phi}+\overline{\lambda_2}\braket{\psi_2|\phi}.
    \end{equation*}
\end{remark}
The notation used in \refdef{hilbert_space} is linked to the \textit{Dirac notation}. Since it will commonly be used throughout this thesis,
the following table should serve as a cheat sheet:
\begin{kaobox}[frametitle={Dirac Notation Cheat Sheet \cite{nielsenQuantumComputationQuantum2012}}]
    \begin{equation*}
        \begin{array}{c|l}
        \hline z^* & \text { Complex conjugate of the complex number } z . \\
        \ket{\psi} & \text { Vector. Also known as a ket. } \\
        \bra{\psi} & \text { Vector dual to }\ket{\psi} \text {. Also known as a bra. } \\
        \braket{\varphi|\psi} & \text { Inner product between the vectors }\ket{\varphi} \text { and }\ket{\psi} . \\
        \ket{\varphi} \otimes\ket{\psi} & \text { Tensor product \sidenote{The tensor product \textit{in our context} is sometimes called \textit{Kronecker product}.}
        of }\ket{\varphi} \text { and }\ket{\psi} . \\
        \ket{\varphi}\ket{\psi} & \text { Abbreviated notation for tensor product of }\ket{\varphi} \text { and }\ket{\psi} . \\
        A^* & \text { Complex conjugate of the matrix } A \text{.} \\
        A^T & \text { Transpose of the matrix } A \text{.} \\
        A^{\dagger} & \text { Hermitian conjugate or adjoint of the matrix } A \text {, that}\\ & \text{ is, } A^{\dagger}=\left(A^T\right)^* . \\
        \braket{\varphi|A|\psi} & \text { Inner product between }\ket{\varphi} \text { and } A\ket{\psi} . \\
        \hline \hline
        \end{array}
    \end{equation*}
\end{kaobox}

\newpage

\begin{definition}
    A matrix $U$ is said to be \textit{unitary} if $U^\dag U=I$, where $I$ is the identity matrix.
\end{definition}

The distinction between finite and infinite Hilbert spaces is typically accompanied by the distinction between \textit{state vectors} and \textit{wave functions}.
In particular, given an $n$-dimensional Hilbert space $\mathcal{H}$ equipped with an orthonormal\sidenote{\textit{Orthonormal} means that the vectors are
mutually orthogonal and are normalized.} basis of kets $\{ \ket{\psi_i} \}$, we can express any element $\ket{\psi}$ 
of this Hilbert space as a linear combination (superposition) of the basis elements
\begin{equation}\labeq{hilb_lin_comb}
    \ket{\psi}=\sum_{i=1}^n c_i \ket{\psi_i}
\end{equation} 
with $c_i\in\mathbb{C}$. One refers to $\ket{\psi}$ as \textit{state vector}. In the case of an infinite-dimensional Hilbert space, however, this sum becomes an integral
\begin{equation}
    \ket{\psi}=\int dx \ \Psi(x)\ket{x}
\end{equation}
where $\Psi(x)$ is referred to as \textit{wave function}, essentially playing the same role as $c_i$ in \refeq{hilb_lin_comb}. Note that here we are using
the position basis $\{ \ket{x} \}$, hence the more precise term \textit{position-space wave function}\sidenote{We could also use the momentum basis, which would give
us the \textit{momentum-space wave function} \cite{WaveFunction2024b}.} \cite{rojoQuantumMechanics2LectureNotes2021}. To be precise, $\ket{x}$ is
an \textit{improper vector}, whereas the $\ket{\psi_i}$ are \textit{proper vectors}. The latter can be normalized to unity. The former, however, can only be normalized
to the \textit{Dirac delta function} \cite{WaveFunction2024b}. \textit{Heuristically}\sidenote{In fact, there does not exist any function with these properties, which
is why it can be seen as a "trick" used by physicists. However, there are rigorous definitions involving measure theory \cite{DiracDeltaFunction2024}.},
it is defined as
\begin{equation*}
    \delta(x)= \begin{cases}0 & \text{if } x \neq 0 \\ \infty & \text{if } x=0\end{cases}
\end{equation*}
such that
\begin{equation*}
    \int_{-\infty}^{\infty} \delta(x)=1 .
\end{equation*}
As already mentioned, we have
\begin{equation*}
    \braket{x^\prime|x}=\delta(x^\prime-x)
\end{equation*}
and thus\sidenote{Again, even though this is the conventional way (\cite{WaveFunction2024b}) to do it, these steps are not perfectly rigorous.}
\begin{equation*}
    \braket{x^\prime|\psi}=\int dx \ \Psi(x)\braket{x^\prime|x}=\Psi(x^\prime).
\end{equation*}
This is usually written more concisely as
\begin{equation}\labeq{wave_function_as_inner_product_eq}
    \braket{x|\psi}=\Psi(x).
\end{equation}

What will follow now are the four postulates of quantum mechanics as described in \cite{nielsenQuantumComputationQuantum2012}. Since there are many different
ways to formulate them, this resource was chosen as it focuses on quantum computation and thus uses the \textit{state vector} representation.

\newpage

\begin{kaobox}[frametitle={Postulate 1}]
    Any isolated physical system is described by a \textit{state space}, mathematically speaking a \textit{Hilbert Space} $\mathcal{H}$, and the state of the system is
    completely described by its \textit{state vector} $\ket{\psi}$, which we require to be of unit-length and therefore satisfy $\braket{\psi|\psi}=1$. 
\end{kaobox}

\begin{kaobox}[frametitle={Postulate 2}]
    The \textit{evolution} of the state in a closed quantum system is described by a \textit{unitary transformation}. More precisely, the
    state $\ket{\psi_1}$ at time $t_1$ is related to the state $\ket{\psi_2}$ at time $t_2$ by a unitary operator $U$ which depends only on the
    times $t_1$ and $t_2$:
    \begin{equation}\labeq{unitary_time_ev_eq}
        \ket{\psi_2}=U\ket{\psi_1}.
    \end{equation}

    The case where only two specific times are considered can be recovered from the differential equation
    \begin{equation}\labeq{simple_schroedinger_eq}
        i\hbar \frac{d\ket{\psi(t)}}{dt}=\hat{H}(t)\ket{\psi(t)}
    \end{equation}
    where $\hbar=\frac{h}{2\pi}$ is the reduced Planck's constant and $\hat{H}$ is the \textit{Hamiltonian operator} of the system\sidenote[][*-2.5]{In essence, operators are
    functions over a space of physical states onto another state of states. One of them is the Hamiltonian operator, corresponding to the total energy of the
    system \cite{OperatorPhysics2024}.}.
    This equation is called the \textit{time-dependent Schrödinger equation}.
\end{kaobox}

\begin{kaobox}[frametitle={Postulate 3}]
    A quantum measurement is described by a collection ${M_m}$ of \textit{measurement operators} acting on the state space that is being measured,
    where $m$ is a potential measurement outcome of the experiment\sidenote{Note that there are different types of measurements, such as general measurements, projective measurements
    or POVMs \cite{nielsenQuantumComputationQuantum2012}.}. The state right after a quantum measurement is given by
    \begin{equation*}
        \ket{\psi}\xrightarrow[]{\text{measurement}} \ket{\psi^\prime}=\frac{M_m\ket{\psi}}{\sqrt{p(m)}} = \frac{M_m\ket{\psi}}{\sqrt{\braket{\psi|M_m^\dag M_m|\psi}}}
    \end{equation*}
    where $p(m)$ is the probability of getting the measurement outcome $m$. The measurement operators satisfy the \textit{completeness equation}
    \begin{equation*}
        \sum_m M_m^\dag M_m = I
    \end{equation*}
    from which one can see that the total probability is indeed one $\sum_m p(m)=1$.
\end{kaobox}

\begin{kaobox}[frametitle={Postulate 4}]
    The state space of a composite physical system is formed by taking the tensor product of the state spaces from the individual subsystems. The state for the
    composite system is $\ket{\psi_1}\otimes\ket{\psi_2}\otimes\ket{\psi_3}\otimes \cdots \otimes\ket{\psi_n}$ if we number the systems $1$ through $n$
    and $\ket{\psi_i}$ is the state for system $i$.
\end{kaobox}

\newpage

As already mentioned, most calculations, apart from those in quantum computation, involve the wave function representation. Using \refeq{wave_function_as_inner_product_eq},
we can rewrite \refeq{simple_schroedinger_eq} to get the \textit{position-space Schrödinger equation}
\begin{equation}
    i\hbar \frac{\partial \Psi(x,t)}{\partial t}=H(t)\Psi(x,t).
\end{equation}
By making the Hamiltonian time-independent, $H(t)\leadsto E$, we get
\begin{equation}\labeq{const_ham_schroedinger_eq}
    i\hbar \frac{\partial \Psi(x,t)}{\partial t}=E\Psi(x,t).
\end{equation}
We can now solve for $\Psi(x,t)$:
\begin{align*}
    \frac{1}{\Psi(x, t)}\frac{\partial \Psi(x,t)}{\partial t}&=-\frac{iE}{\hbar} \\
    \int \frac{1}{\Psi(x, t)}\partial \Psi(x,t)&=\int -\frac{iE}{\hbar}\partial t \\
    \ln{\Psi(x,t)}&=-\frac{iEt}{\hbar}+\tilde{c} \\
    \Psi(x,t)&=c e^{-iEt / \hbar}.
\end{align*}
Finally, since $\Psi(x, 0)=c$, we get
\begin{equation}\labeq{const_hamiltonian_wave_function_sol_eq}
    \Psi(x, t)=e^{-iEt / \hbar}\Psi(x, 0)
\end{equation}
or, rewritten using \refeq{wave_function_as_inner_product_eq},
\begin{equation}
    \ket{\psi(t)}=e^{-iEt / \hbar}\ket{\psi(0)} .
\end{equation}
This is the solution\sidenote{Notice that this resembles \refeq{unitary_time_ev_eq}.} for the \textit{time-dependent Schrödinger equation with constant Hamiltonian}.
Plugging \refeq{const_hamiltonian_wave_function_sol_eq} back into \refeq{const_ham_schroedinger_eq}, we get
\begin{align*}
    i\hbar\frac{\partial}{\partial t}(e^{-iEt / \hbar}\Psi(x,0))&=H e^{-iEt / \hbar}\Psi(x,0) \\
    i\hbar e^{-iEt / \hbar}\left(-\frac{iE}{\hbar}\right)\Psi(x,0)&=H e^{-iEt / \hbar}\Psi(x,0) \\
    E\Psi(x)&=H\Psi(x) \stepcounter{equation}\tag{\theequation}\labeq{soft_const_hamiltonian_wave_function_sol_eq}
\end{align*}
The Hamiltonian $H$ of such systems can be written as
\begin{equation}\labeq{hamiltonian_laplacian_eq}
    H=-\frac{\hbar^2}{2m}\nabla^2+V(x,y,z),
\end{equation}
where $\nabla^2=\frac{\partial^2}{\partial x^2}+\frac{\partial^2}{\partial y^2}+\frac{\partial^2}{\partial z^2}$ is the Laplacian in Cartesian coordinates.
Using \refeq{wave_function_as_inner_product_eq}, \refeq{soft_const_hamiltonian_wave_function_sol_eq} becomes the \textit{time-independent Schrödinger equation}
\begin{equation}\labeq{time_independent_eq}
    \hat{H}\ket{\psi}=E\ket{\psi}
\end{equation}
and we can see that we have set up an \textit{eigenvalue problem}.

We can now use this theory to study our first and last actual physical system. This is usually referred to as the "particle in a box" problem \cite{SchrodingerEquation2024}.

\newpage

\begin{example}
    \labexample{particle_in_box_walk_through}
    \begin{marginfigure}[*1]
        \includegraphics{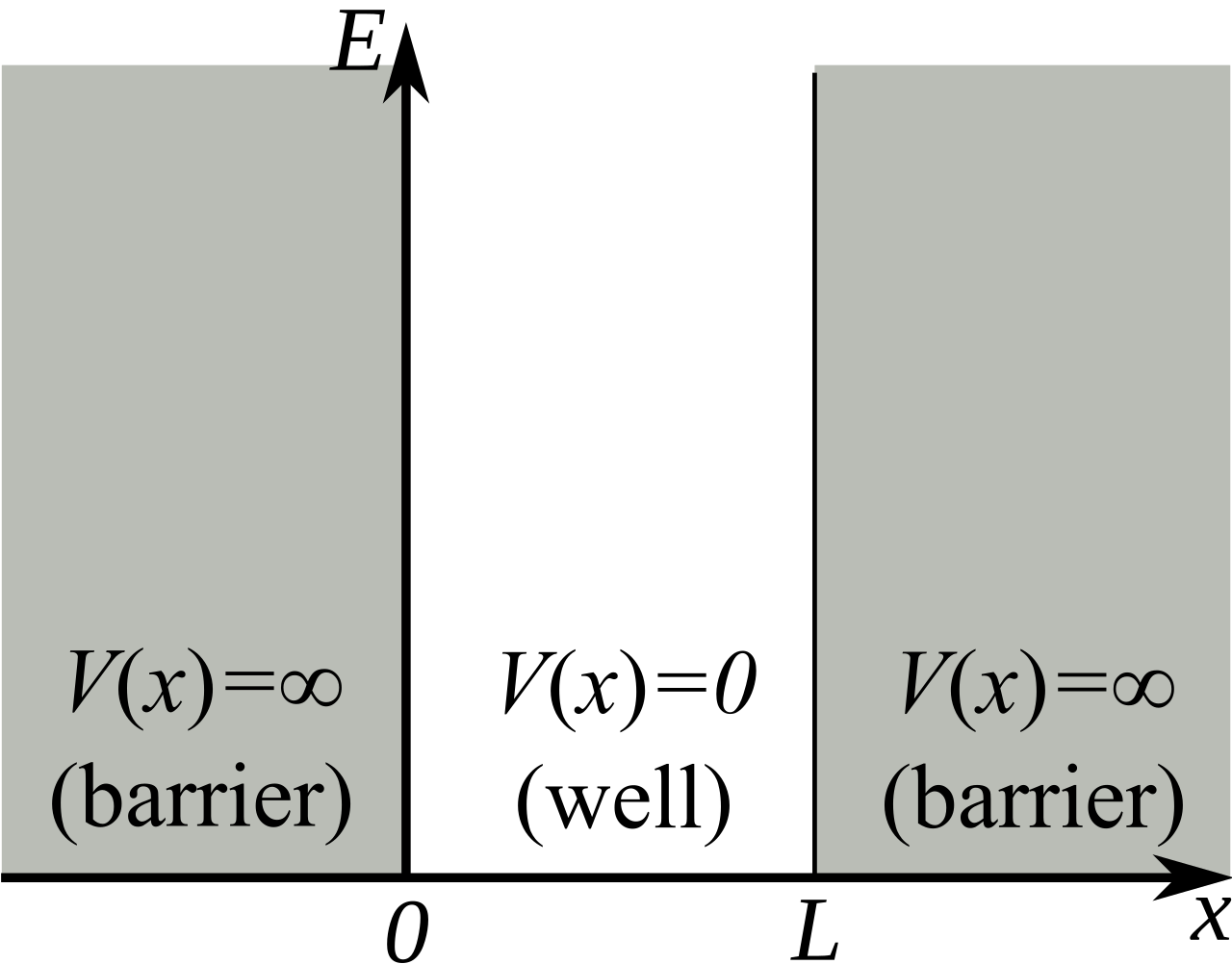}
        \caption[Particle in a box]{Visual representation of the "Particle in a box" problem \cite{ParticleBox2024}.}
        \labfig{margin_particle_in_box}
    \end{marginfigure}
    Consider a particle in the system represented in \reffig{margin_particle_in_box}. The system is one-dimensional, and the following
    applies to the potential of the particle $V(x)$:
    \begin{align}\labeq{potential_energy_particle_in_box_eq}
        V(x)=
        \begin{cases}
            0 & \text{if } 0<x<L \\
            \infty & \text{if } x\leq 0, x\geq L .
        \end{cases}
    \end{align}
    Since a particle cannot have infinite potential energy, it must stay in the well given by the box. This knowledge allows us to
    combine \refeq{soft_const_hamiltonian_wave_function_sol_eq}, \refeq{hamiltonian_laplacian_eq} and \refeq{potential_energy_particle_in_box_eq} into

    \begin{equation*}
        -\frac{\hbar^2}{2m}\frac{d^2\Psi(x))}{dx^2}=E\Psi(x)).
    \end{equation*}
    We can rearrange this to get
    \begin{equation}\labeq{schroedinger_evaluated}
        \frac{d^2\Psi(x)}{dx^2}=-\frac{2mE}{\hbar^2}\Psi(x)).
    \end{equation}
    The roots of the characteristic polynomial $\lambda^2+\frac{2mE}{\hbar^2}$ are $\lambda_{1,2}=\pm ik$, with $k\coloneqq\sqrt{\frac{2mE}{\hbar^2}}$. Thus, the general solution of \refeq{schroedinger_evaluated} is
    \begin{equation}\labeq{particle_in_box_ansatz}
        \Psi(x)=c_1\cos(kx)+c_2\sin(kx).
    \end{equation}
    We can now consider our boundary conditions
    \begin{align*}
        \Psi(x\leq 0)=0 &\implies \Psi(0)=0 \\
        \Psi(x\geq L)=0 &\implies \Psi(L)=0.
    \end{align*}
    By plugging these values into \refeq{particle_in_box_ansatz}, we obtain
    \begin{equation*}
        c_1=0
    \end{equation*}
    and
    \begin{equation*}
        c_2\sin(kL)=0\Leftrightarrow kL=n\pi, \ n\in\mathbb{Z}.
    \end{equation*}
    With this knowledge, we can find $c_2$ and $E$. To find $c_2$, we need to use
    the \textit{normalization condition}

    \begin{equation}
        \int_{-\infty}^{\infty} |\Psi(x)|^2 dx=1,
    \end{equation}
    which in our case becomes
    \begin{equation*}
        \int_0^L (c_2)^2 \sin^2(kx)dx=1.
    \end{equation*}
    Evaluating the definite integral yields
    \begin{equation*}
        -\frac{(c_2)^2 \sin(2kL)}{4k}+\frac{(c_2)^2 L}{2}=1 \Leftrightarrow c_2=\sqrt{\frac{2}{L}}.
    \end{equation*}
    Thus, we get our final equation
    \begin{equation}
        \Psi_n(x)=\sqrt{\frac{2}{L}}\sin\left(\frac{n\pi}{L}x\right).
    \end{equation}
    To find the energy eigenvalues, we simply need to compare $k=\sqrt{\frac{2mE}{\hbar^2}}$ with $kL=n\pi$ to get
    \begin{equation}
        E_n=\frac{n^2\pi^2\hbar^2}{2mL^2}=\frac{h^2 n^2}{8mL^2},
    \end{equation}
    where $n=1,2,3,\dots$ . This is because $\Psi_0(x)=0$ cannot be normalized and furthermore because any $n<0$ yields the same quantum state
    as its positive equivalent. These results are visualized in \reffig{margin_particle_in_a_box_waves}.
    \begin{marginfigure}[*-10]
        \includegraphics{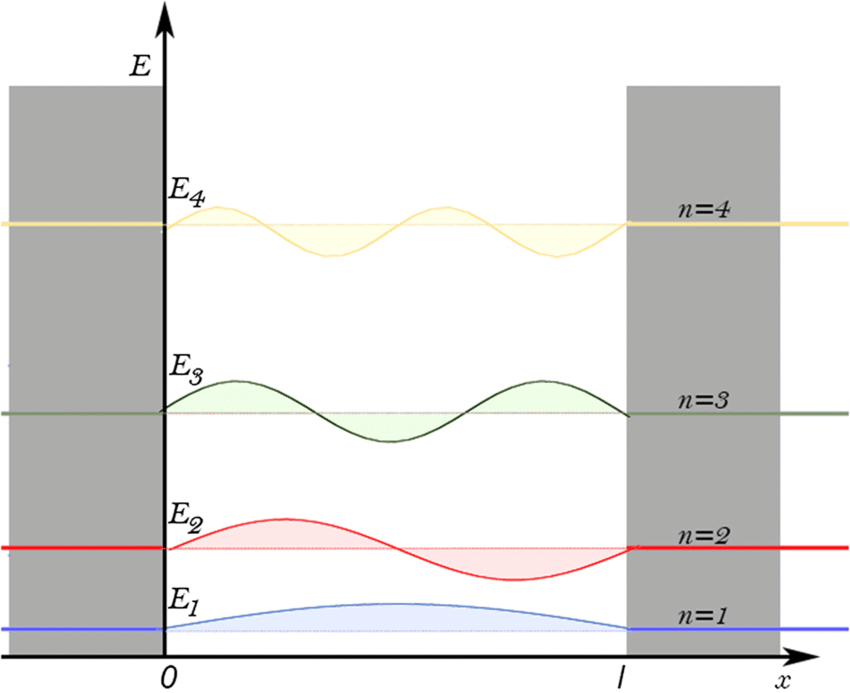}
        \caption[Particle in a box waves]{Visual representation of the different standing waves and corresponding quantized energy levels \cite{GraphicalRepresentationParticle}.}
        \labfig{margin_particle_in_a_box_waves}
    \end{marginfigure}
\end{example}

\section{Qubits, Quantum Gates and Quantum Circuits}
\labsec{qubits_qgates_and_qcircs}

In 1982, the renowned physicist Richard Feynman published a paper about \textit{quantum computers}.

\begin{quote}
    \small \textit{Can physics be simulated by a universal computer? [...] the physical world is quantum mechanical, and therefore the proper problem is the simulation of quantum physics [...] Can you do it with a new kind of computer--a quantum computer?} \\
    \small --- \textup{Richard Feynman \cite{feynmanSimulatingPhysicsComputers1982}}
\end{quote}

The problem mentioned here is quite simple to understand. Imagine you want to study the molecular structure of drugs. For that purpose, one may consider
the electron spin\sidenote{In reality, of course, the task of studying such systems is much more intricate.}. One electron may just be spin-up or spin-down, so there are two states to be considered.
Looking at two electrons, we have four possible states. With ten electrons, there are already 1024 states. Fifty electrons are enough to have more
than 1 quadrillion\sidenote{A one followed by fifteen zeros!}. However, this is only the case when working with a \textit{classical computer} that utilizes
\textit{bits} of information. We will later see that \textit{qubits} --- something inherently quantum mechanical --- could resolve this problem.

To explain the basic ideas behind quantum computers, we will consider an analogy to classical computers. The reader is therefore assumed to have sufficient
knowledge of classical computers.

\begin{figure}[!ht]
    \def\svgwidth{\textwidth}
    \import{./figures/}{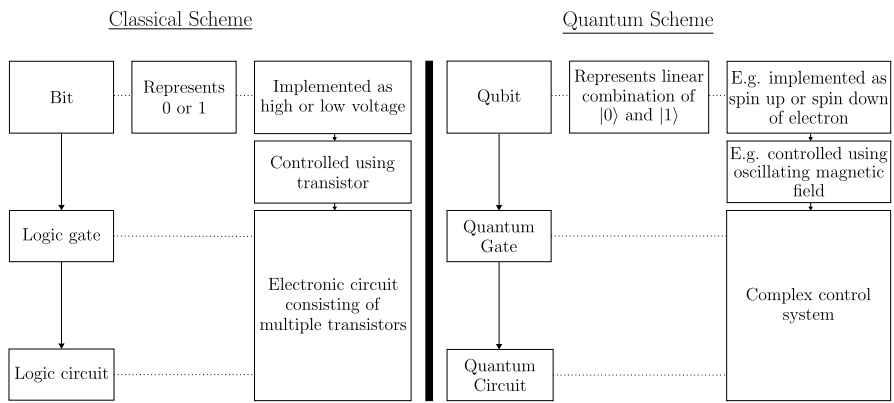}

	\caption[Classical Quantum Analogy]{An analogy between classical computers and quantum computers.}
	\labfig{normal_classical_quantum_computer_analogy}
\end{figure}

The left sides of the two respective schemes in \reffig{normal_classical_quantum_computer_analogy}, classical and quantum, represent the mathematical and information theoretic concepts, whereas the right sides
represent their physical implementation. Notice that it is a simplified description of how those systems actually work. For example, many types of quantum
computers also have a classical part, which handles tasks like error-correction\sidenote{Qubits are prone to \textit{quantum decoherence}, a process
which results in the loss of information due to loose energetic couplings. This is why one very often sees big cryostats cooling qubits, so that
there is less thermal noise and thus less couplings.} \cite{QuantumDecoherence2024}.

\begin{definition}
    A \textit{qubit} is an abstract representation of a single bit of information, represented as a state vector $\ket{\psi}$ in a two-dimensional Hilbert space
    \begin{equation}
        \ket{\psi}=\alpha\ket{0}+\beta\ket{1}, \ \alpha,\beta\in\mathbb{C},
    \end{equation}
    where the vectors $\ket{0}\coloneqq\binom{1}{0}$ and $\ket{1}\coloneqq\binom{0}{1}$ form an orthonormal basis and are thus called \textit{computational basis states}.
\end{definition}

We can now use the third postulate of quantum mechanics to study how a measurement on a qubit may look like. For reasons mentioned later on, the type of measurement we are going
to use is called projective measurement. Let us define our measurement operators:
\begin{align*}
    M_0&=\ket{0}\bra{0}=
    \left(\begin{array}{ll}
        1 & 0 \\
        0 & 0
        \end{array}\right) \\
    M_1&=\ket{1}\bra{1}=
    \left(\begin{array}{ll}
        0 & 0 \\
        0 & 1
        \end{array}\right)
\end{align*}
We can thus compute the measurement probabilities
\begin{align}
    \labeq{prob_qubit_0_eq}
    p(0)&=\braket{\psi|{M_0^\dag M_0}|\psi}=\alpha^* \alpha\braket{0|0}=|\alpha|^2 \\
    \labeq{prob_qubit_1_eq}
    p(1)&=\braket{\psi|M_1^\dag M_1|\psi}=\beta^* \beta\braket{1|1}=|\beta|^2
\end{align}
and the state after the measurement
\begin{align}
    \frac{M_0\ket{\psi}}{\sqrt{p(0)}}=\frac{\alpha}{|\alpha|}\ket{0} \\
    \frac{M_1\ket{\psi}}{\sqrt{p(1)}}=\frac{\beta}{|\beta|}\ket{1}.
\end{align}

From \refeq{prob_qubit_0_eq} and \refeq{prob_qubit_1_eq} we can deduce the \textit{Born rule}
\begin{equation}\labeq{born_rule_eq}
    |\alpha|^2+|\beta|^2=1.
\end{equation}

The absolute values of $\alpha,\beta\in\mathbb{C}$ can be interpreted as real coordinates of points on the unit circle:
\begin{align*}
    |\alpha|=\cos(\omega), \quad |\beta|=\sin(\omega)
\end{align*}
for $\omega\in[0,2\pi)$. Therefore, $\alpha$ and $\beta$ can be written down in polar coordinates
\begin{align*}
    \alpha=e^{i\phi_1}\cos(\omega), \quad \beta=e^{i\phi_2}\sin(\omega)
\end{align*}
for $\phi_1,\phi_2\in[0,2\pi)$. Notice that \refeq{born_rule_eq} is preserved, since $|e^{i\phi_1}|=|e^{i\phi_2}|=1$. We now have
\begin{align*}
    \ket{\psi}&=\alpha\ket{0}+\beta\ket{1} \\
    &=e^{i\phi_1}\cos(\omega)\ket{0}+e^{i\phi_2}\sin(\omega)\ket{1} \\
    &=e^{i\phi_1}(\cos(\omega)\ket{0}+e^{i(\phi_2-\phi_1)}\sin(\omega)\ket{1}),
\end{align*}
where $e^{i\phi_1}$ is the \textit{global phase}, which can be omitted since it influences the probabilities of $\ket{0}$ and $\ket{1}$ the same way, and thus cannot
be distinguished by measurements.
For reasons that will become clear in a bit, we substitute $\omega\coloneqq\frac{\theta}{2}$ and $\varphi\coloneqq\phi_2-\phi_1$. We thus obtain our final equation
\begin{equation}\labeq{bloch_sphere_qubit_def_eq}
    \ket{\psi}=\cos\left(\frac{\theta}{2}\right)\ket{0}+e^{i\varphi}\sin\left(\frac{\theta}{2}\right)\ket{1},
\end{equation}
for $\theta\in [0,\pi], \varphi\in [0,2\pi)$. We call $\varphi$ the \textit{phase of the qubit}. This result can be visualized
using the \textit{Bloch sphere} (\cf \reffig{margin_bloch_sphere}).

One can now see that plugging in $\theta=\pi$ produces a vector pointing straight down to $\ket{1}$ according to the visualization and
$\ket{1}$ according to our formula too. This is why we substituted $\omega=\frac{\theta}{2}$.

Measuring the qubit state $\ket{\psi}$ \textit{projects} it onto the Z-axis (either to $\ket{0}$ or $\ket{1}$). This is why the procedure is called a projective measurement.

\begin{marginfigure}[*-12]
    \includegraphics{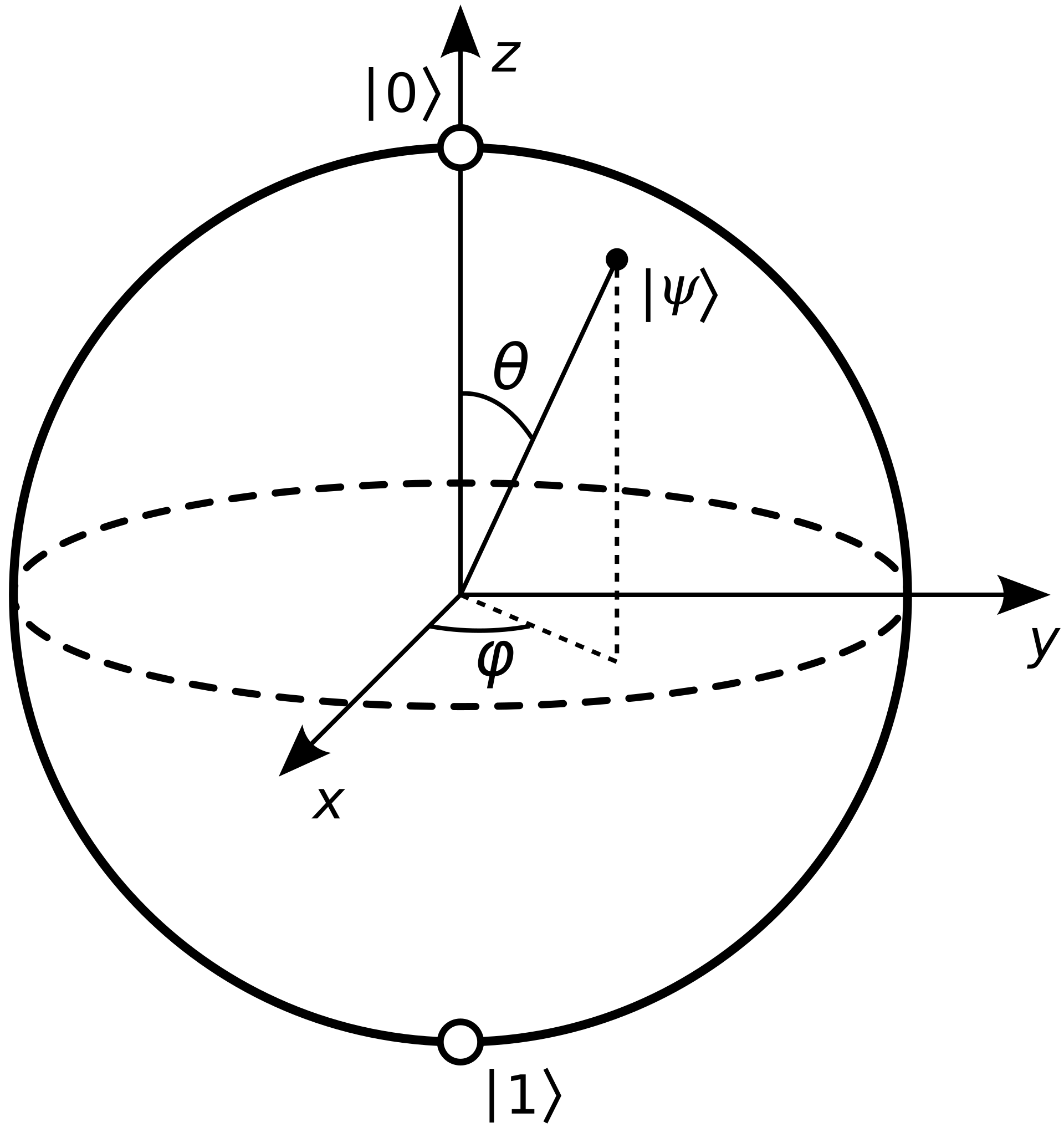}
    \caption[Bloch Sphere]{Bloch sphere \cite{BlochSphere2024}.}
    \labfig{margin_bloch_sphere}
\end{marginfigure}

We will demonstrate how Postulate 4 applies by considering a two-qubit systems:
\begin{equation*}
    \ket{\psi}=\alpha_{00}\ket{00}+\alpha_{01}\ket{01}+\alpha_{10}\ket{10}+\alpha_{11}\ket{11}.
\end{equation*}
Intuitively, there must be a normalization condition
\begin{equation}
    \sum_{x\in\{0,1\}^2}|\alpha_x|^2=1,
\end{equation}
where the notation $\{0,1\}^2$ means "the set of strings of length two with each letter being either zero or one" \cite{nielsenQuantumComputationQuantum2012}.
When measuring the first qubit alone, this will yield 0 with probability $p(m)$ and leave the post-measurement state
\begin{equation*}
    \ket{\psi^\prime}=\frac{\alpha_{00}\ket{00}+\alpha_{01}\ket{01}}{\sqrt{p(m)}}=\frac{\alpha_{00}\ket{00}+\alpha_{01}\ket{01}}{\sqrt{|\alpha_{00}|^2+|\alpha_{01}|^2}}.
\end{equation*}

Something very interesting happens when our two-qubit system is a so-called \textit{Bell state} or \textit{EPR Pair}, which is given by
\begin{equation}
    \frac{\ket{00}+\ket{11}}{\sqrt{2}}.
\end{equation}
Measuring the first qubit yields $0$ with probability $0.5$, leaving the post-measurement state $\ket{00}$. The probability for $1$ is $0.5$ and
the post-measurement state is $\ket{11}$. However, notice what happens when measuring the second qubit of the post-measurement state: One will get the same value as in the first
measurement. We thus say they are \textit{correlated}. This phenomenon  is called \textit{quantum entanglement} \cite{nielsenQuantumComputationQuantum2012}.

In order to manipulate the state of our qubit(s), we can use \textit{quantum gates}. There are \textit{single-qubit gates} and \textit{n-qubit gates}.
Both are \textit{unitary transformations}.

One of the simplest single-qubit gates is the \textit{NOT gate}, \textit{Pauli-X gate} or simply \textit{X gate}, whose matrix
representation\sidenote{If not mentioned otherwise, we always represent them in the computational basis.} is
\begin{equation}
    X=\left(\begin{array}{ll}
        0 & 1 \\
        1 & 0
        \end{array}\right).
\end{equation}
\begin{marginfigure}[*2]
    \begin{equation*}
        \input{1_Quantum_Computation/X_gate.tikz}
    \end{equation*}
    \caption[NOT gate]{Quantum circuit notation for NOT gate \cite{nielsenQuantumComputationQuantum2012}.}
    \labfig{margin_not_gate}
\end{marginfigure}
The corresponding \textit{quantum circuit notation} is depicted in \reffig{margin_not_gate}.

Like the logical NOT gate, it "flips" the state:
\begin{equation*}
    X\ket{0}=\ket{1} \quad \text{ and } \quad X\ket{1}=\ket{0}.
\end{equation*}
Another very important single-qubit gate is the \textit{Hadamard gate}
\begin{equation}
    H=\frac{1}{\sqrt{2}}\left(\begin{array}{ll}
        1 & 1 \\
        1 & -1
        \end{array}\right),
\end{equation}
\begin{marginfigure}[*2]
    \centering
    \begin{equation*}
        \input{1_Quantum_Computation/Hadamard_gate.tikz}
    \end{equation*}
    \caption[Hadamard gate]{Quantum circuit notation for Hadamard gate \cite{nielsenQuantumComputationQuantum2012}.}
    \labfig{margin_hadamard_gate}
\end{marginfigure}
which is depicted in \reffig{margin_hadamard_gate}. One can easily see how useful it is to prepare states in a superposition:
\begin{equation*}
    H\ket{0}=\frac{\ket{0}+\ket{1}}{\sqrt{2}}\eqqcolon\ket{+} \quad \text{ and } \quad H\ket{1}=\frac{\ket{0}-\ket{1}}{\sqrt{2}}\eqqcolon\ket{-}.
\end{equation*}

In fact, it is possible to visualize all single-qubit gate transformations in the Bloch sphere according to \refthm{sing_qubit_rot_decomp_thm}.
\begin{theorem}\labthm{sing_qubit_rot_decomp_thm}
    Any $2\times2$ unitary matrix can be decomposed as
    \begin{equation}
        U=e^{i \alpha}\left[\begin{array}{cc}
            e^{-i \beta / 2} & 0 \\
            0 & e^{i \beta / 2}
            \end{array}\right]\left[\begin{array}{cc}
            \cos(\frac{\gamma}{2}) & -\sin(\frac{\gamma}{2}) \\
            \sin(\frac{\gamma}{2}) & \cos(\frac{\gamma}{2})
            \end{array}\right]\left[\begin{array}{cc}
            e^{-i \delta / 2} & 0 \\
            0 & e^{i \delta / 2}
            \end{array}\right],
    \end{equation}
    where $\alpha$, $\beta$, $\gamma$ and $\delta$ are real-valued. The second matrix is an ordinary rotation, and the other two matrices are
    rotations in different planes.
    \cite{nielsenQuantumComputationQuantum2012}
\end{theorem}
This should make sense intuitively, since any state\sidenote{To be more precise, any \textit{pure} state. There are also \textit{mixed} states, but they
are of no importance to us \cite{nielsenQuantumComputationQuantum2012}.} can be represented by a vector pointing onto the surface of the Bloch sphere, and
because any allowed transformation again creates such a vector, we know that they can be transformed into each other using rotations according to basic geometry.
\begin{marginfigure}
    \begin{equation*}
        \input{1_Quantum_Computation/CNOT_gate.tikz}
    \end{equation*}
    \caption[CNOT gate]{Quantum circuit notation for CNOT gate, where the filled black circle represents the control bit and the bottom one represents
    the target bit \cite{nielsenQuantumComputationQuantum2012}.}
    \labfig{margin_cnot_gate}
\end{marginfigure}
One of the simplest n-qubit gates is the \textit{CNOT gate} or \textit{controlled NOT gate} (\cf \reffig{margin_cnot_gate}). As the name suggests, there
is a \textit{control bit} and a \textit{target bit}. If the state for the control bit is $\ket{0}$, nothing happens.
If it is $\ket{1}$, then a NOT gate is applied on the target bit. The corresponding matrix representation is
\begin{equation}
    \ket{0}\bra{0}\otimes I+\ket{1}\bra{1}\otimes X=\left[\begin{array}{llll}
        1 & 0 & 0 & 0 \\
        0 & 1 & 0 & 0 \\
        0 & 0 & 0 & 1 \\
        0 & 0 & 1 & 0
        \end{array}\right].
\end{equation}
Notice that we could also let the control bit be below the target bit, which would result in
\begin{equation*}
    \left[\begin{array}{llll}
        1 & 0 & 0 & 0 \\
        0 & 0 & 0 & 1 \\
        0 & 0 & 1 & 0 \\
        0 & 1 & 0 & 0
        \end{array}\right].
\end{equation*}
As we will see later, there could also be multiple qubits in-between the control and the target bit.

Another very important gate is the \textit{Toffoli gate} or \textit{CCNOT gate} (\cf \reffig{margin_toffoli_gate}). It is a subcase of the \textit{multi-control Toffoli gate}, where the
number of control bits was set to two. The matrix representation is
\begin{marginfigure}
    \centering
    \begin{equation*}
        \input{1_Quantum_Computation/Toffoli_gate.tikz}
    \end{equation*}
    \caption[Toffoli gate]{Quantum circuit notation for Toffoli gate \cite{nielsenQuantumComputationQuantum2012}.}
    \labfig{margin_toffoli_gate}
\end{marginfigure}
\begin{equation}
    \left[\begin{array}{llllllll}
        1 & 0 & 0 & 0 & 0 & 0 & 0 & 0 \\
        0 & 1 & 0 & 0 & 0 & 0 & 0 & 0 \\
        0 & 0 & 1 & 0 & 0 & 0 & 0 & 0 \\
        0 & 0 & 0 & 1 & 0 & 0 & 0 & 0 \\
        0 & 0 & 0 & 0 & 1 & 0 & 0 & 0 \\
        0 & 0 & 0 & 0 & 0 & 1 & 0 & 0 \\
        0 & 0 & 0 & 0 & 0 & 0 & 0 & 1 \\
        0 & 0 & 0 & 0 & 0 & 0 & 1 & 0
        \end{array}\right].
\end{equation}

For the sake of completeness, we will now briefly list the rest of the most important quantum gates:
\begin{kaobox}[frametitle={Common quantum gates}]
    \begin{equation*}
        \begin{array}{ccc}
            \text{Pauli-Y} & \begin{tikzpicture}
	\begin{pgfonlayer}{nodelayer}
		\node [style=gate] (0) at (0, 0) {$Y$};
		\node [style=none] (1) at (-1, 0) {};
		\node [style=none] (2) at (1, 0) {};
	\end{pgfonlayer}
	\begin{pgfonlayer}{edgelayer}
		\draw (1.center) to (0);
		\draw (0) to (2.center);
	\end{pgfonlayer}
\end{tikzpicture}
 & 
            \left[\begin{array}{cc}
                0 & -i \\
                i & 0
            \end{array}\right] \\
            \noalign{\vskip 1.3mm}
            \text{Pauli-Z} & \begin{tikzpicture}
	\begin{pgfonlayer}{nodelayer}
		\node [style=gate] (0) at (0, 0) {$Z$};
		\node [style=none] (1) at (-1, 0) {};
		\node [style=none] (2) at (1, 0) {};
	\end{pgfonlayer}
	\begin{pgfonlayer}{edgelayer}
		\draw (1.center) to (0);
		\draw (0) to (2.center);
	\end{pgfonlayer}
\end{tikzpicture}
 &
            \left[\begin{array}{cc}
                1 & 0 \\
                0 & -1
            \end{array}\right] \\
            \noalign{\vskip 1.3mm}
            \text{S-gate / phase} & \begin{tikzpicture}
	\begin{pgfonlayer}{nodelayer}
		\node [style=gate] (0) at (0, 0) {$S$};
		\node [style=none] (1) at (-1, 0) {};
		\node [style=none] (2) at (1, 0) {};
	\end{pgfonlayer}
	\begin{pgfonlayer}{edgelayer}
		\draw (1.center) to (0);
		\draw [in=180, out=0] (0) to (2.center);
	\end{pgfonlayer}
\end{tikzpicture}
 &
            \left[\begin{array}{ll}
                1 & 0 \\
                0 & i
            \end{array}\right] \\
            \noalign{\vskip 1.3mm}
            \text{T-gate / }\frac{\pi}{8} & \begin{tikzpicture}
	\begin{pgfonlayer}{nodelayer}
		\node [style=gate] (0) at (0, 0) {$T$};
		\node [style=none] (1) at (-1, 0) {};
		\node [style=none] (2) at (1, 0) {};
	\end{pgfonlayer}
	\begin{pgfonlayer}{edgelayer}
		\draw (1.center) to (0);
		\draw (0) to (2.center);
	\end{pgfonlayer}
\end{tikzpicture}
 &
            \left[\begin{array}{cc}
                1 & 0 \\
                0 & e^{i \pi / 4}
            \end{array}\right] \\
            \noalign{\vskip 1.3mm}
            \text{SWAP} & \input{1_Quantum_Computation/SWAP_gate.tikz} &
            \left[\begin{array}{llll}
                1 & 0 & 0 & 0 \\
                0 & 0 & 1 & 0 \\
                0 & 1 & 0 & 0 \\
                0 & 0 & 0 & 1
            \end{array}\right] \\
            \noalign{\vskip 1.3mm}
            \text{Controlled-Z} & \input{1_Quantum_Computation/controlled_Z_gate.tikz} &
            \left[\begin{array}{cccc}
                1 & 0 & 0 & 0 \\
                0 & 1 & 0 & 0 \\
                0 & 0 & 1 & 0 \\
                0 & 0 & 0 & -1
            \end{array}\right] \\
            \noalign{\vskip 1.3mm}
            \text{Controlled-S} & \input{1_Quantum_Computation/controlled_S_gate.tikz} &
            \left[\begin{array}{llll}
                1 & 0 & 0 & 0 \\
                0 & 1 & 0 & 0 \\
                0 & 0 & 1 & 0 \\
                0 & 0 & 0 & i
            \end{array}\right] \\
            \noalign{\vskip 1.3mm}
            \shortstack{Fredkin / \\Controlled-SWAP} & \input{1_Quantum_Computation/Fredkin_gate.tikz} &
            \left[\begin{array}{llllllll}
                1 & 0 & 0 & 0 & 0 & 0 & 0 & 0 \\
                0 & 1 & 0 & 0 & 0 & 0 & 0 & 0 \\
                0 & 0 & 1 & 0 & 0 & 0 & 0 & 0 \\
                0 & 0 & 0 & 1 & 0 & 0 & 0 & 0 \\
                0 & 0 & 0 & 0 & 1 & 0 & 0 & 0 \\
                0 & 0 & 0 & 0 & 0 & 0 & 1 & 0 \\
                0 & 0 & 0 & 0 & 0 & 1 & 0 & 0 \\
                0 & 0 & 0 & 0 & 0 & 0 & 0 & 1
            \end{array}\right] \\
        \end{array}
    \end{equation*}
\end{kaobox}

As it is the case with classical computers, more sophisticated computations require ensembles of gates. An example of such a \textit{quantum circuit} can be seen in \reffig{circuit_example}.
On the left most side, one can see the initial state, usually\sidenote{
    We use the following notation:
    \begin{equation*}
    M_1 \otimes M_2 \otimes \ldots \otimes M_n \eqqcolon \bigotimes_{k=1}^n M_k
    \end{equation*} and
    \begin{equation*}
        \bigotimes_{k=1}^n M \eqqcolon M^{\otimes n}
    \end{equation*}
} $\ket{0}^{\otimes n}$. The time evolution is from left to right and horizontal lines, so-called \textit{wires}, correspond to one of the $n$ qubits.
Groups of qubits, sometimes denoted by indexed variables such as $q_i$, are termed \textit{quantum registers} (\cf \refsec{application_multiloop_feynman_diagrams}).

\begin{figure}[!ht]
    \begin{equation*}
        \input{1_Quantum_Computation/circuit_example.tikz}
    \end{equation*}
	\caption[Quantum Circuit Example]{An example of a whole quantum circuit.}
	\labfig{circuit_example}
\end{figure}

\begin{marginfigure}[*3]
    \begin{equation*}
        \input{1_Quantum_Computation/measurement_gate.tikz}
    \end{equation*}
    \caption[Measurement Gate]{Measurement symbol in a quantum circuit.}
    \labfig{measurement_symbol}
\end{marginfigure}

Horizontal composition of gates represents the matrix product of the corresponding matrices, so a Hadamard gate on the left of a NOT gate transforms the input state $\ket{\psi}$ according
to $XH\ket{\psi}$, in other words the Hadamard gate is applied first. Vertical composition of gates represents the tensor product of the corresponding matrices, so a Hadamard gate
positioned above a NOT gate transforms the input state $\ket{\psi\otimes\phi}$ according to $H\ket{\psi}\otimes X\ket{\phi}$. Notice that we order $\psi$ and $\phi$ from left to right,
corresponding to top to bottom in the circuit. Finally, on the right-most side of the circuit, there might be symbols like in \reffig{measurement_symbol}. These represent quantum measurements, which result in classical bits, and are hence connected to classical wires, denoted as \textit{double wires}.

\section{Complexity and Classical Simulation}
\labsec{complexity_and_simulation}

In 1994, Peter Shor published his work on an algorithm that is today known as \textit{Shor's algorithm}. It is a \textit{quantum algorithm}, meaning the computation is done on
a quantum computer. The goal is to factor an integer $N$. The general, most-efficient classical algorithm, the \textit{general number field sieve}, has a sub-exponential
time-complexity\sidenote{This means it may grow faster than any polynomial, but is still significantly smaller than an exponential \cite{TimeComplexity2024}.}. Shor's algorithm
on the other hand is polynomial in $\log(N)$. The key part here is that we do not know whether there exists a classical algorithm that could match or beat Shor's algorithm.
There is strong evidence that the quantum algorithm is actually better, however, we do not know for certain \cite{ShorsAlgorithm2024}. Even more so, we \textit{do not know for certain}
whether quantum computers actually pose any advantage \cite{nielsenQuantumComputationQuantum2012}.

The field concerned with the difficulty of classical and quantum computational problems is called \textit{computational complexity theory}. The most basic concept is
the \textit{complexity class}, a set of computational problems that require similar computational resources to be solved. \textbf{P} and \textbf{NP} are the most
important classes. \textbf{P} is the set of all problems that can be solved quickly on a classical computer. \textbf{NP} is the set of all problems whose solutions
can be quickly checked on a classical computer. As an example, consider the problem of factoring an integer $N$. As previously discussed, we think that there exists no algorithm
capable of solving the problem quickly on a classical computer, and we thus think it is not in \textbf{P}. However, given a potential solution consisting of $p_1,\dots,p_n$, we
can simply multiply them to see if the result equals $N$. Thus, it definitely is in \textbf{NP}. It is clear that \textbf{P} is a subset of \textbf{NP}. What remains
one of the greatest unsolved problems in theoretical computer science is whether these two classes are different, in other words
\begin{equation}
    \mathbf{P} \stackrel{?}{\neq} \mathbf{N P}.
\end{equation}
Although we know that all problems in \textbf{P} can be solved efficiently using a quantum computer, we again do not know whether that is true for \textbf{NP} as
well \cite{nielsenQuantumComputationQuantum2012}.

While such grand questions remain unsolved, smaller yet still very important statements have been proven. One of them is
the \textbf{Gottesman-Knill theorem} \cite{gottesmanHeisenbergRepresentationQuantum1998}:
\begin{theorem}
    Any quantum computer performing only: a) Clifford group gates, b) measurements of Pauli group operators, and c) Clifford group operations
    conditioned on classical bits, which may be the results of earlier measurements, can be perfectly simulated in polynomial time on a
    probabilistic classical computer.
\end{theorem}
In the next chapter, after having introduced new tools, we will consider this theorem in more detail. For now, it should suffice to say that there
are \textbf{Clifford circuits} and \textbf{non-Clifford circuits}, where according to the theorem, it is possible to simulate the former efficiently on a classical
computer\sidenote{And we think that the latter cannot be simulated efficiently on a classical computer, but they can be executed efficiently on a quantum computer.}.

In order to see why one might be interested in \textbf{classical simulation}, we need to take a quick detour to experimental physics. Unsurprisingly,
the actual physical implementation of quantum computers poses yet another big challenge. For instance, we do not know which approach\sidenote{Superconductors,
ions, cold atoms, silicon dots, topological materials, photonic crystals, nitrogen vacancies, nuclear magnetic resonance etc.} will "win the race".
Each approach has certain advantages but also disadvantages, which are most often related to the number of qubits, the quality of the qubits and the circuit execution
speed \cite{pfaendlerAdvancementsQuantumComputing2024}. One of the problems that has to be solved is how to actually test whether the results a newly built
quantum computer produces are correct. This can be done by simulating a quantum computer on a classical computer. Such a \textbf{quantum circuit simulator} might
also be used to test or even come up with quantum algorithms \cite{xuHerculeanTaskClassical2023}.

Classical simulation can be split up into two domains \cite{KissingerWetering2024Book}:
\begin{enumerate}
    \item \textbf{Weak simulation} is concerned with sampling from the probability distribution that one would get when actually running the quantum circuit.
    \item \textbf{Strong simulation} is concerned with calculating the actual probability of observing a certain measurement outcome.
\end{enumerate}

Additionally, there exists a variety of different simulation methods based on ideas such as state vectors, density-matrices, matrix product states or tensor
networks. We will focus on another approach based on \textit{stabilizer decompositions}, which we will discuss in \refch{novel_state_decompositions} and
subsequent chapters. In \refsec{application_multiloop_feynman_diagrams} we will use a state vector approach implemented in Qiskit \cite{qiskit2024} in order to
check our own simulation results. Interested readers may refer to \cite{xuHerculeanTaskClassical2023} for a more detailed discussion of the various simulation
methods.
\setchapterpreamble[u]{\margintoc}
\chapter{Graphical Languages}
\labch{graphical_languages}

\begin{quote}
    \small \textit{A diagram to help write down the mathematical expressions.} \\
    \small --- \textup{Richard Feynman}
\end{quote}

\begin{marginfigure}[*2.9]
    \includegraphics{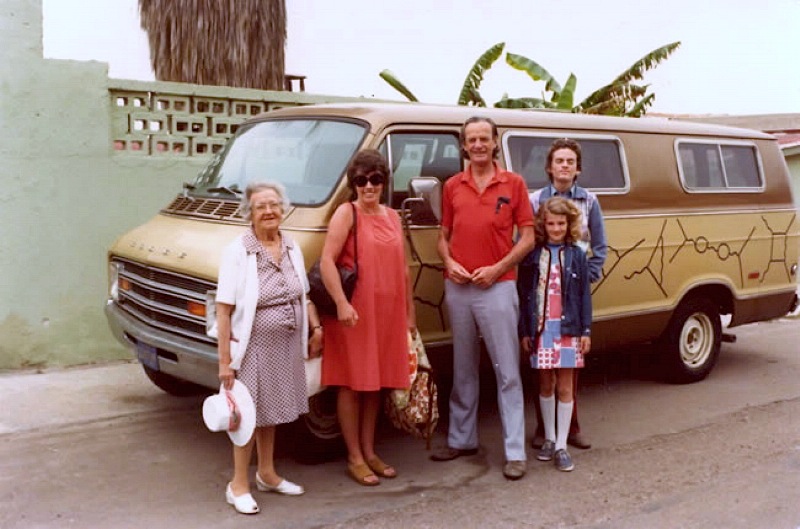}
    \caption[Feynman van]{Iconic Dodge Tradesman Maxivan \cite{FeynmanDsonVanjpgJPEGImage}.}
    \labfig{margin_feynman_van}
\end{marginfigure}

\begin{marginfigure}[*11.1]
    \includegraphics{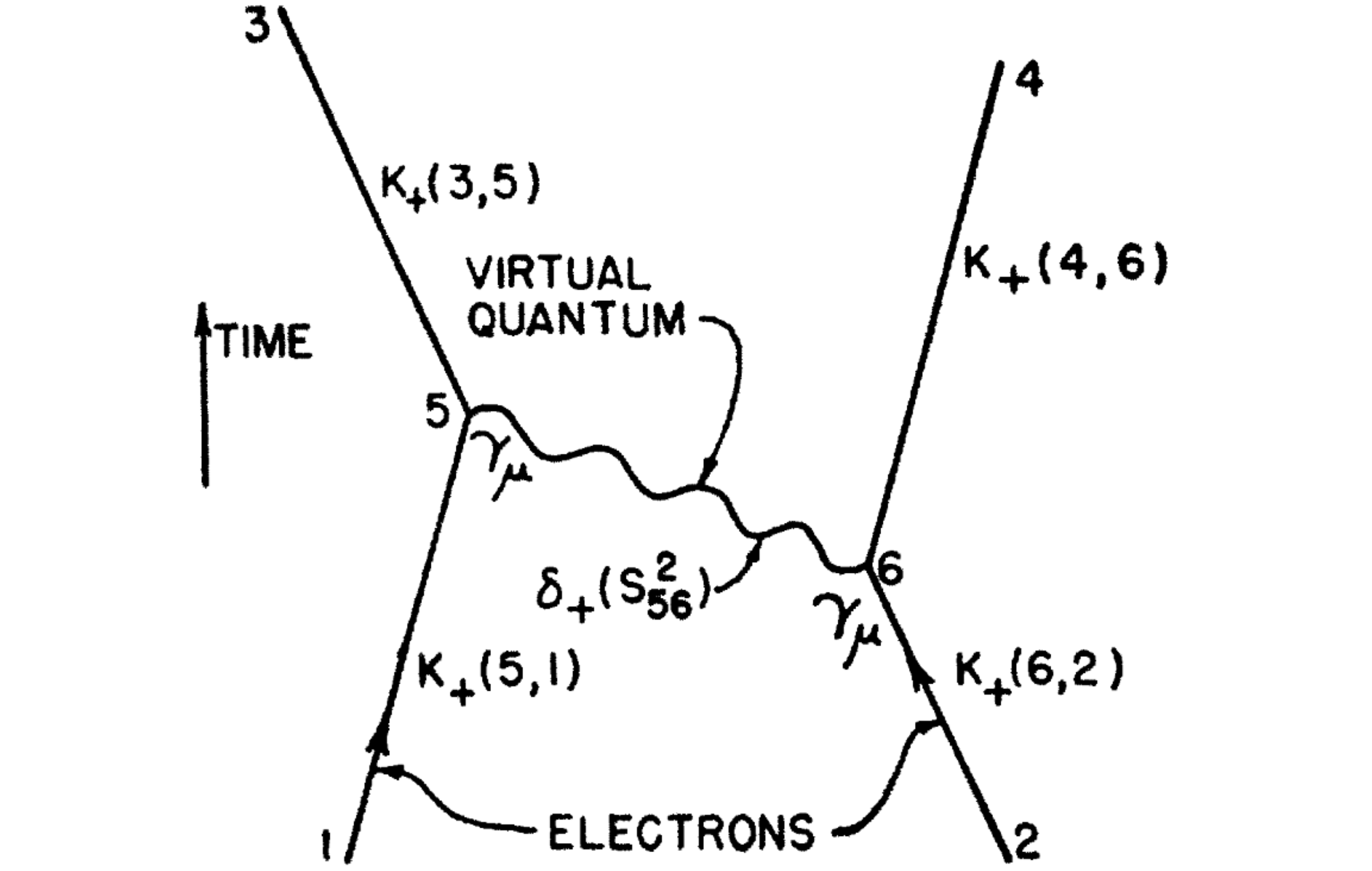}
    \caption[First Feynman diagram]{First published Feynman diagram \cite{feynmanSpaceTimeApproachQuantum1949}.}
    \labfig{margin_first_feynman_diagram}
\end{marginfigure}
In 1975, perphaps one of the most iconic cars, at least for a physicist, was bought (\cf \reffig{margin_feynman_van}). The well-sighted reader might have spotted the rather peculiar
drawings on the side of the van. Unsurprisingly, the owner of the van, Richard Feynman, came up with these types of diagrams himself. The first
\textit{Feynman diagram}, which can be seen in \reffig{margin_first_feynman_diagram}, was published in 1949 in \cite{feynmanSpaceTimeApproachQuantum1949}.

\begin{figure}[!ht]
	\includegraphics[width=\textwidth]{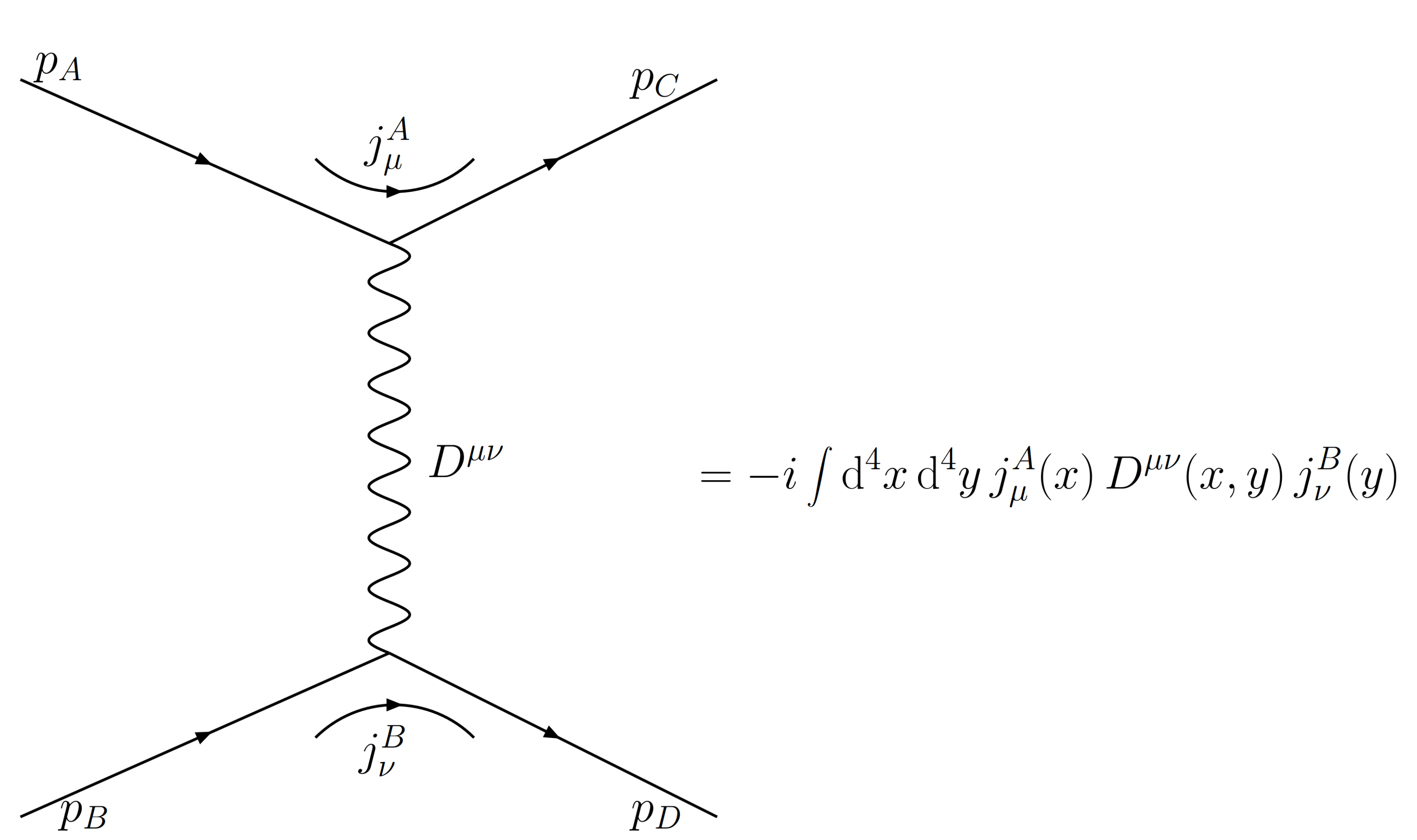}
	\caption[Feynman diagram formula]{Equivalence between Feynman diagram and formula \cite{seymourMeaningFeynmanDiagrams}.}
	\labfig{normal_feynman_diagram_formula}
\end{figure}
\begin{marginfigure}[*2]
    \includegraphics{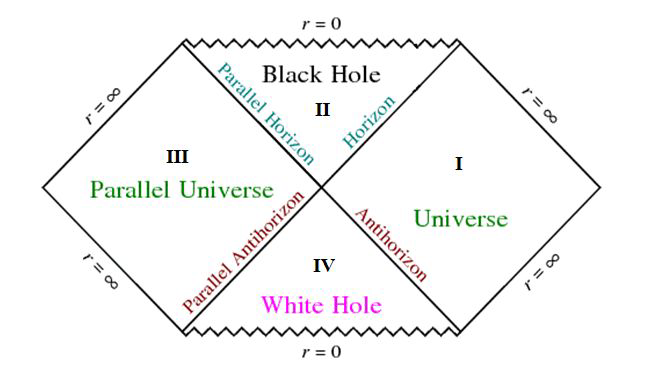}
    \caption[Penrose diagram]{Example of a Penrose diagram \cite{FigureFourDifferent}.}
    \labfig{margin_penrose_diagram}
\end{marginfigure}
\begin{marginfigure}[*10]
    \includegraphics{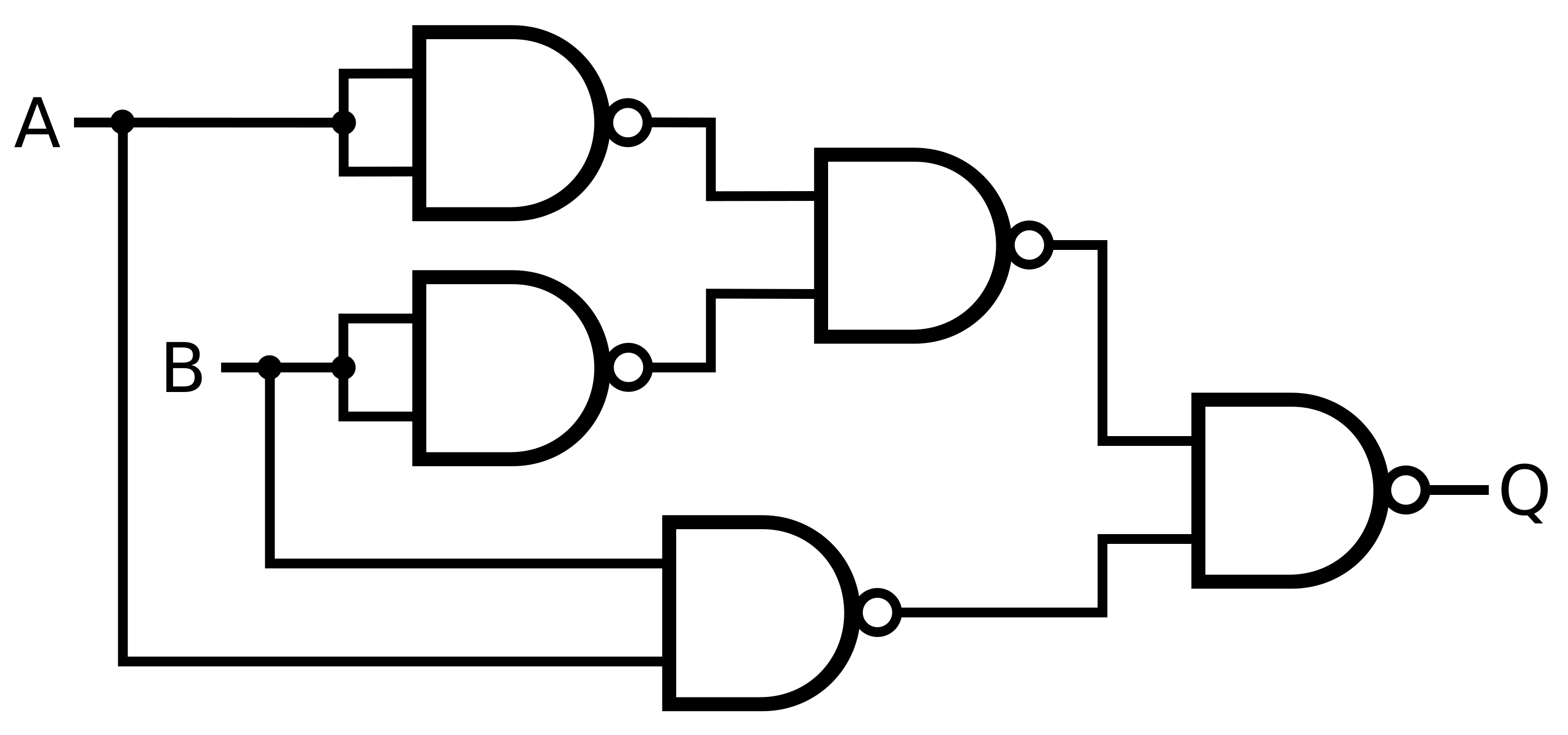}
    \caption[Logic circuit]{Example of a logic circuit \cite{LogicGate2024}.}
    \labfig{margin_logic_circuit}
\end{marginfigure}
\begin{marginfigure}[*17]
    \includegraphics{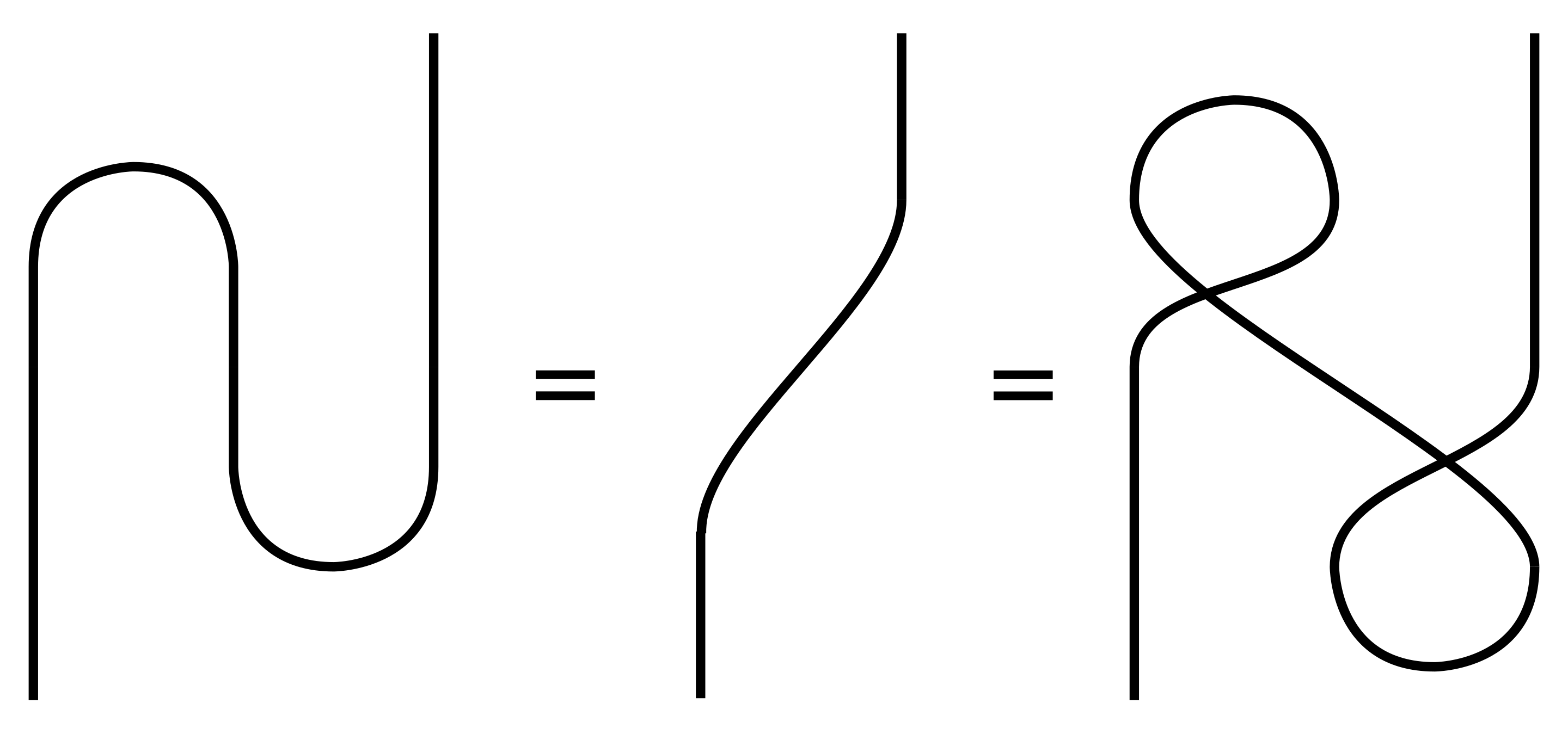}
    \caption[Tensor diagram notation]{Example of tensor diagram notation \cite{PenroseGraphicalNotation2024}.}
    \labfig{margin_tensor_diagram_notation}
\end{marginfigure}

Since then, Feynman diagrams have become an essential part in the toolkit of physicists studying \textit{Quantum Field Theory (QFT)}. In brief, QFT combines
classical field theory, special relativity and quantum mechanics into a powerful theoretical framework which is commonly used in particle physics and condensed matter
physics. The current standard model of particle physics is based on QFT \cite{QuantumFieldTheory2024}. Feynman diagrams cannot only help to visualize formulas, they can also help
to \textit{come up} with formulas. In fact, there are purely graphical rules that define how a diagram may be rewritten, without having to deal with the underlying equations.
\reffig{normal_feynman_diagram_formula} shows a Feynman diagram and the according formula that represents a scattering process, for example electron-muan scattering,
where particle $A$ scatters into particle $B$. More precisely, both the diagram and the formula represent the \textit{scattering amplitudes}.
The details, however, lie outside the scope of this thesis \cite{seymourMeaningFeynmanDiagrams}.

The importance of \textit{graphical languages} is apparent: Feynman diagrams simplify difficult QFT calculations, Penrose diagrams (\cf \reffig{margin_penrose_diagram})
help represent causal relations in spacetime \cite{PenroseDiagram2024}, logic circuits (\cf \reffig{margin_logic_circuit}) give a simple representation of complex compositions of Boolean functions \cite{LogicGate2024} and tensor diagram notation (\cf \reffig{margin_tensor_diagram_notation}) allows for simple graphical manipulations of tensors \cite{PenroseGraphicalNotation2024}.

Nevertheless, it is important to keep in mind that graphical languages are simply a tool, the same way a drill is great for screws, but rather bad for nails. You
could use it to hammer in a nail, but there are better tools like a hammer. This holds true for any tool a mathematician or scientist might use.

\section{ZX-calculus}

\enlargethispage{2\baselineskip}
A \textit{ZX-diagram} is a graphical representation of a linear map between qubits. It is similar to the tensor diagram notation presented in the last
section\sidenote{In fact, a ZX-diagram is a \textit{tensor network} \cite{KissingerWetering2024Book}.}. Moreover, it is linked to the quantum
circuit notation. As will soon be evident, we can convert \textit{any} quantum circuit into a ZX-diagram:
\begin{equation}\labeq{first_example_of_zx_eq}
    \input{2_Graphical_Languages/first_example_qcirc.tikz} \ \leadsto \ \input{2_Graphical_Languages/first_example_zx.tikz}
\end{equation}
The \textit{ZX-calculus} is a graphical language consisting of ZX-diagrams and \textit{rewrite rules} that allow one to manipulate such
diagrams \cite{vandeweteringZXcalculusWorkingQuantum2020}. It was introduced in 2008 by Coecke and Duncan \cite{coeckeInteractingQuantumObservables2008}.
\begin{example}\labexample{ghz_state_example}
    It is possible to diagrammatically prove that the circuit from \refeq{first_example_of_zx_eq} creates a so-called \textit{GHZ state}
    \begin{equation}
        \frac{\ket{000}+\ket{111}}{\sqrt{2}}
    \end{equation}
    using the rewrite rules\sidenote{Although not being introduced yet, the fact that one might already be able to guess certain rules, speaks for the
    intuitive nature of this tool.}:
    \begin{equation*}
        \input{2_Graphical_Languages/first_example_zx_step1.tikz}=\input{2_Graphical_Languages/first_example_zx_step2.tikz}
        =\input{2_Graphical_Languages/first_example_zx_step3.tikz}=\input{2_Graphical_Languages/first_example_zx_step4.tikz}
        =\begin{tikzpicture}
	\begin{pgfonlayer}{nodelayer}
		\node [style=none] (0) at (1, 1) {};
		\node [style=none] (1) at (1, 0) {};
		\node [style=none] (2) at (1, -1) {};
		\node [style=Z dot] (3) at (0, 0) {};
	\end{pgfonlayer}
	\begin{pgfonlayer}{edgelayer}
		\draw (1.center) to (3);
		\draw (3) to (0.center);
		\draw (3) to (2.center);
	\end{pgfonlayer}
\end{tikzpicture}

    \end{equation*}
\end{example}
Readers interested in more such examples or a more detailed introduction to the ZX-calculus than the one to follow, may refer
to \cite{vandeweteringZXcalculusWorkingQuantum2020}, where \refexample{ghz_state_example} is taken from, or, alternatively, to \cite{KissingerWetering2024Book}.

In the remainder of this section, we will more formally introduce the ZX-calculus\sidenote{If not noted otherwise, the primary source of information for the following
is \cite{KissingerWetering2024Book}.}. \refsec{related_calculi} will introduce certain variations of the
ZX-calculus and finally, \refsec{applications} will explore previous and current research areas by reviewing state-of-the-art applications.

ZX-diagrams are constructed by taking matrix products or tensor products\sidenote{The process is analogous to how we treated the matrices corresponding to
the gates in quantum circuits, in other words, matrix product for horizontal composition and tensor product for vertical composition.} of the following linear maps,
so-called \textit{generators}:
\begin{itemize}
    \item \textit{Z-spider}: \begin{equation}\labeq{z_spider_eq}
        Z_m^n[\alpha] \ = \ \input{2_Graphical_Languages/z_spider_def.tikz} \ = \ \ket{0}^{\otimes n}\bra{0}^{\otimes m}+e^{i\alpha}\ket{1}^{\otimes n}\bra{1}^{\otimes m}
    \end{equation}
    \item \textit{X-spider}: \begin{equation}\labeq{x_spider_eq}
        X_m^n[\alpha] \ = \ \input{2_Graphical_Languages/x_spider_def.tikz} \ = \ \ket{+}^{\otimes n}\bra{+}^{\otimes m}+e^{i\alpha}\ket{-}^{\otimes n}\bra{-}^{\otimes m}
    \end{equation}
\end{itemize}
\begin{remark}
    These generators form two families of linear maps that are given by
    \begin{equation*}
        \left\{ Z_m^n[\alpha] : (\mathbb{C}^2)^{\otimes m} \to (\mathbb{C}^2)^{\otimes n} \ \middle| \ m,n\in\mathbb{N},\alpha\in[0,2\pi) \right\}
    \end{equation*}
    and
    \begin{equation*}
        \left\{ X_m^n[\alpha] : (\mathbb{C}^2)^{\otimes m} \to (\mathbb{C}^2)^{\otimes n} \ \middle| \ m,n\in\mathbb{N},\alpha\in[0,2\pi) \right\}.
    \end{equation*}
\end{remark}

Because of how they look, these generators are called 'spiders'. Hence, the black lines are also referred to as \textit{legs}. Alternatively, we can see ZX-diagrams as graphs,
and thus call the spiders \textit{vertices}/\textit{nodes} and the legs \textit{edges}. In analogy to quantum circuits, the edges might also be called \textit{wires}. Since time
flows from left to right, wires on the left side are \textit{inputs} and wires on the right side are \textit{outputs}. A ZX-diagram with no inputs and no outputs is referred to
as \textit{scalar diagram}. The number $\alpha$ in \refeq{z_spider_eq} and \refeq{x_spider_eq} is the \textit{phase} of the spider. Finally, $n+m$ is
the \textit{arity} or \textit{(vertex) degree} of the spider.

Consider a Z-spider with one input edge, one output edge and a phase of zero:
\begin{equation*}
    \begin{tikzpicture}
	\begin{pgfonlayer}{nodelayer}
		\node [style=none] (0) at (-1, 0) {};
		\node [style=none] (1) at (1, 0) {};
		\node [style=Z phase dot] (2) at (0, 0) {$0$};
	\end{pgfonlayer}
	\begin{pgfonlayer}{edgelayer}
		\draw (0.center) to (2);
		\draw (2) to (1.center);
	\end{pgfonlayer}
\end{tikzpicture}

\end{equation*}
Per convention, a zero-phase is not written at all:
\begin{equation*}
    \begin{tikzpicture}
	\begin{pgfonlayer}{nodelayer}
		\node [style=none] (0) at (-1, 0) {};
		\node [style=none] (1) at (1, 0) {};
		\node [style=Z dot] (2) at (0, 0) {};
	\end{pgfonlayer}
	\begin{pgfonlayer}{edgelayer}
		\draw (0.center) to (2);
		\draw (2) to (1.center);
	\end{pgfonlayer}
\end{tikzpicture}

\end{equation*}
We can calculate the corresponding matrix
\begin{align*}
    \ket{0}^{\otimes 1}\bra{0}^{\otimes 1}+e^{i0}\ket{1}^{\otimes 1}\bra{1}^{\otimes 1} =& \ket{0}\bra{0}+\ket{1}\bra{1} \\
    =&\binom{1}{0}\left(\begin{array}{ll}
        1 & 0
        \end{array}\right)+\binom{0}{1}\left(\begin{array}{ll}
        0 & 1
        \end{array}\right) \\
    =&\left(\begin{array}{ll}
        1 & 0 \\
        0 & 0
        \end{array}\right)+\left(\begin{array}{ll}
        0 & 0 \\
        0 & 1
        \end{array}\right) \\
    =&\left(\begin{array}{ll}
        1 & 0 \\
        0 & 1
        \end{array}\right)
\end{align*}
which happens to be the identity matrix. We have thus derived one of the \textit{nine} rewrite rules\sidenote[][*-1]{The remaining proofs are outside the scope
of this thesis.}
\cite{KissingerWetering2024Book} \cite{sutcliffeProcedurallyOptimisedZXDiagram2024} \cite{eastSpinnetworksZXcalculus2022a} \cite{moyardIntroductionZXcalculus2023a}:
\begin{equation}\labeq{zx_rewrite_rules_eq}
    \input{2_Graphical_Languages/zx_rewrite_rules.tikz}
\end{equation}
\begin{table}[h!]
    \begin{tabular}{ l l l l }
    \SpiderFusion                   & \textbf{sp}ider fusion &
    \ColorChange                    & \textbf{c}olor \textbf{c}hange \\
    \PiCommutation                  & $\bm{\pi}$-commutation rule &
    \StateCopy                      & state \textbf{c}opy \\
    \Bialgebra                      & \textbf{b}ialgebra &
    \HadamardHadamardCancellation   & \textbf{h}adamard-\textbf{h}adamard cancellation \\
    \Hopf                           & \textbf{ho}pf &
    \Identity                       & \textbf{id}entity \\
    \EulerDecomposition             & \textbf{eu}ler decomposition of Hadamard \\
    \end{tabular}
\end{table}

It is worth mentioning that this set of rewrite rules is not minimal. In other words, there exists a smaller set of rewrite rules, with which
one could derive \StateCopy for instance. These 'obsolete' rules are still listed, as we will frequently encounter them throughout this thesis.

Alongside these rewrite rules, there exist two \textit{meta rules}:
\begin{itemize}
    \item All rewrite rules hold true when swapping red and green
    \item Only connectivity matters
\end{itemize}
The former should be clear. As for the latter, it simply means that one is allowed to 'bend' wires in any way.

For practical purposes, a new generator, the \textit{Hadamard box}, might be introduced, even though it can be represented using
more fundamental generators (\cf \EulerDecomposition). Instead of using a yellow box, we sometimes use the equivalent \textit{Hadamard edge}
for notational convenience\sidenote{It is especially used in programmatic implementations.}:
\begin{equation}
    \input{2_Graphical_Languages/had_box_blue_edge.tikz}
\end{equation}

\begin{remark}
    A \textit{scalar} $\xi_i$ in \refeq{zx_rewrite_rules_eq} is multiplied by the matrix representation of the diagram next to it. Since
    many applications in literature are not concerned with the global scalar of a diagram, one usually writes \cite{poorUniqueNormalForm}:
    \begin{equation*}
        \input{2_Graphical_Languages/approx_zx.tikz}
    \end{equation*}
    In the \textit{scalar-free ZX-calculus}\sidenote{This is the first small example of a variation of the ZX-calculus (\cf \refsec{related_calculi}).}, one might even write \cite{pehamEquivalenceCheckingQuantum2022}:
    \begin{equation*}
        \input{2_Graphical_Languages/scalar_free_zx.tikz}
    \end{equation*}
\end{remark}

Looking at \refeq{zx_rewrite_rules_eq}, the curious reader might wonder what values $\alpha$ and $\beta$ typically assume. However, the
answer is dependent on which \textit{fragment} is being used. Listed below are the two most commonly used fragments.
\begin{definition}\labdef{clifford_fragment}
    The \textit{Clifford-fragment}, \textit{$\frac{\pi}{2}$-fragment} or \textit{stabilizer-fragment}\sidenote{The first two terms are
    used in \cite{jeandelCompletenessZXCalculus2020a}. The term 'StabZX' is used in \cite{laakkonenGraphicalStabilizerDecompositions2022}, which for the sake of clarity we shall call \textit{stabilizer-fragment} instead. The definitions refer to the same.} restricts the ZX-calculus
    by requiring the phase of spiders to be $\alpha=k\frac{\pi}{2}, \ k\in\mathbb{Z}$.
\end{definition}
\begin{definition}\labdef{clifford_t_fragment}
    The \textit{Clifford+T-fragment} or \textit{$\frac{\pi}{4}$-fragment} restricts the ZX-calculus
    by requiring the phase of spiders to be $\alpha=k\frac{\pi}{4}, \ k\in\mathbb{Z}$.
\end{definition}
\begin{remark}
    Accordingly, we can talk about \textit{Clifford ZX-diagrams}/\textit{stabilizer ZX-diagrams}
    and \textit{non-Clifford ZX-diagrams}/\textit{non-stabilizer ZX-diagrams} \cite{KissingerWetering2024Book}\cite{laakkonenGraphicalStabilizerDecompositions2022}.
\end{remark}

\refdef{clifford_fragment} is of interest, since it can be shown that the \textit{simplification routines} required to evaluate a Clifford ZX-diagram
constructed from a Clifford circuit can be run in polynomial time \cite{laakkonenGraphicalStabilizerDecompositions2022}. As demonstrated in \cite{vandeweteringZXcalculusWorkingQuantum2020}, one can even prove the Gottesman-Knill theorem using the ZX-calculus.

On the other hand, non-Clifford diagrams resulting from \refdef{clifford_t_fragment} can be used to show \textit{universality} of the ZX-calculus.
In fact, this is one of the three main questions to ask when evaluating the sensibility of such a calculus with respect to its usefulness for reasoning
about quantum systems \cite{muussLinearCombinationsZXdiagramsa}:

\textbf{Is the ZX-calculus universal?} {} \textit{Can ZX-diagrams be used to represent any linear map between finite-dimensional Hilbert spaces?} The
answer is yes. In the following, we will outline a short proof.
\begin{proof}
    As a first step, we need to show that any scalar can be
    represented as a ZX-diagram. The following identities hold true:
    \begin{equation*}
        \input{2_Graphical_Languages/scalars_as_zx.tikz}
    \end{equation*}
    It is possible to show that one can construct any complex number using these scalars. As a next step, we will represent the CNOT, the Hadamard, the S and the
    T-gate as ZX-diagrams:
    \begin{equation*}
        \input{2_Graphical_Languages/universal_zx_gate_set.tikz}
    \end{equation*}
    Since this is a known \textit{universal quantum gate set} \cite{QuantumLogicGate2024}, we can thus, together with the correct scalars, represent any linear map between finite-dimensional Hilbert spaces \cite{muussLinearCombinationsZXdiagramsa}.
\end{proof}

\textbf{Is the ZX-calculus sound?} {} \textit{Is the interpretation of any ZX-diagram invariant under all rewrite rules?} The answer is again yes.
The proof involves showing invariance for each rewrite rule, which is outside the scope of this thesis \cite{muussLinearCombinationsZXdiagramsa}.

\textbf{Is the ZX-calculus complete?} {} \textit{Is it possible to prove any true equation for linear mappings using ZX-diagrams and a set of
rewrite rules of the ZX-calculus?} And again, the answer is yes. However, this result took many years until the most general version could be proven \cite{poorZXcalculusCompleteFiniteDimensional2024}. Simpler versions such as the completeness for the stabilizer-fragment were proven before
that \cite{backensZXcalculusCompleteStabilizer2014}.

\section{Related Calculi}
\labsec{related_calculi}

In the previous section, we have seen that new generators like the Hadamard box can be introduced in order to make diagrams more concise.
We will consider another example of such generators in \refsec{prev_work_stab_decomp}, where we will see that the introduction of \textit{triangles},
and for that matter also \textit{stars}, can make our diagrams much more compact and allow us to make easy-to-write rules.

As it turns out, one can much more drastically change the Z\textbf{X}-calculus, by replacing the fundamental generator \textbf{X} with the \textbf{H}
or \textbf{W} generator.

The first of these two generators is the \textit{H-box}, which is a generalization of the Hadamard box. Formally, it is defined as
\begin{equation}
    \input{2_Graphical_Languages/h_box_def.tikz}=\sum a^{i_1 \ldots i_m j_1 \ldots j_n}\ket{j_1 \ldots j_n}\bra{i_1 \ldots i_m}
\end{equation}
where the sum is over all $i_1, \ldots, i_m, j_1, \ldots, j_n \in\{0,1\}$ and $a\in\mathbb{C}$. It can thus be seen that all entries of the corresponding
matrix are $1$, expect for the bottom right element, which is equal to $a$. For example, we have that
\begin{equation*}
    \begin{tikzpicture}
	\begin{pgfonlayer}{nodelayer}
		\node [style=H] (0) at (0, 0) {$a$};
		\node [style=none] (1) at (-1, 0) {};
		\node [style=none] (2) at (1, 0) {};
	\end{pgfonlayer}
	\begin{pgfonlayer}{edgelayer}
		\draw (2.center) to (1.center);
	\end{pgfonlayer}
\end{tikzpicture}
=\left(\begin{array}{ll}
        1 & 1 \\
        1 & a
        \end{array}\right).
\end{equation*}
Notice that for $a=-1$, we simply get a rescaled Hadamard box:
\begin{equation}
    \input{2_Graphical_Languages/h_box_hadamard_box_corr.tikz}
\end{equation}
It is standard practice to omit the $-1$ label for H-boxes, which results in the same yellow box that is used for Hadamard boxes. Since this might lead to confusion,
it is important to clarify which convention is being used, if it is not clear from the context.

This generator alongside with the according rewrite-rules\sidenote{We refrain from showing the rewrite-rules, as they will not be used in this thesis.} results in
the \textit{ZH-calculus}, which was introduced in \cite{backensZHCompleteGraphical2019}.

There are many applications of this new calculus, some of which we will see in \refsec{applications}. One common use-case is the representation
of controlled unitaries. For instance, it is possible to show the following \cite{vandeweteringZXcalculusWorkingQuantum2020}:
\begin{equation*}
    \input{2_Graphical_Languages/zh_tof_qcirc_proof_idea1.tikz}
\end{equation*}
We can extract the resulting quantum circuit from the last diagram by using so-called \textit{phase gadgets}, giving us a particularly efficient construction:
\begin{equation*}
    \input{2_Graphical_Languages/zh_tof_qcirc_proof_idea2.tikz}
\end{equation*}

The second generator is the \textit{W-spider}, which is formally defined as
\begin{align}
    \input{2_Graphical_Languages/w_spider_def.tikz}= \ &|10 \ldots 0\rangle\langle 0 \ldots 0|+\ldots+|0 \ldots 01\rangle\langle 0 \ldots 0| \\
    + \ &|0 \ldots 0\rangle\langle 10 \ldots 0|+\ldots+|0 \ldots 0\rangle\langle 0 \ldots 01|.
\end{align}
where the sum is over all $\ket{x_1 \ldots x_n}\bra{y_1 \ldots y_m}$ such that only one of $x_i$ and $y_j$ is 1.
This generator creates the basis for a new calculus, the \textit{ZW-calculus}, which was introduced in \cite{hadzihasanovicDiagrammaticAxiomatisationQubit2015}.

At this point it is worth noting that a translation between the ZX-calculus, the ZH-calculus and the ZW-calculus is in fact possible, as demonstrated
in \cite{caretteRecipeQuantumGraphical2020}. This means that, in theory, all three calculi are capable of the same, however, in practice, one calculus might be
much easier to work with than the others. Furthermore, it is possible to combine aspects of the individual calculi, resulting in the \textit{ZXH-calculus} and
the \textit{ZXW-calculus} \cite{eastSpinnetworksZXcalculus2022a} \cite{shaikhHowSumExponentiate2023}.

As demonstrated in \cite{poorCompletenessArbitraryFinite2023b}, one might even consider a $d$-dimensional extension of the ZXW-calculus, which can be used
to reason about \textit{qudits}. One of the commonly used generators in that calculus is the \textit{W node}:
\begin{equation}\labeq{w_node_def}
    \raisebox{-0.40\height}{\includegraphics[height=1.0cm]{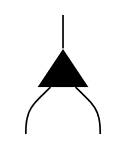}}=|00\rangle\langle 0|+\sum_{i=1}^{d-1}(|0 i\rangle+|i 0\rangle)\langle i|
\end{equation}
By setting $d=2$, we can see that \refeq{w_node_def} becomes equivalent to the following diagram involving W-spiders:
\begin{equation*}
    \input{2_Graphical_Languages/w_node_w_spider_equiv.tikz}
\end{equation*}

\section{Applications}
\labsec{applications}

At the start of this chapter, we have seen many graphical languages used in a wide variety of research fields. In the following, we will provide
an overview of what research directions have been pursued since the introduction of the ZX-calculus in 2008.

One of the first applications of the ZX-calculus was the description of graph states and measurement-based quantum computations \cite{duncanGraphsStatesNecessity2009}
\cite{duncanRewritingMeasurementBasedQuantum2010}. Another early use case was the verification of quantum protocols \cite{hillebrandQuantumProtocolsInvolving2011}.
Verification of oracle-based quantum algorithms was later enabled through the development of the scalable ZX-calculus \cite{caretteQuantumAlgorithmsOracles2021a}.
As it turns out, the ZX-calculus is also a great language to reason about lattice surgery on surface codes \cite{beaudrapZXCalculusLanguage2017}
\cite{beaudrapPauliFusionComputational2019}. It has also proved useful for finding and verifying the correctness of quantum error correcting codes
\cite{chancellorGraphicalStructuresDesign2018} \cite{duncanVerifyingSteaneCode2014} \cite{garvieVerifyingSmallestInteresting2018}. It has even been used to study
more exotic ideas like quantum natural language processing, which aims to study linguistic meanings and structures using quantum hardware
\cite{coeckeFoundationsNearTermQuantum2020}. Arguably one of the biggest research directions is concerned with classical simulation
\cite{kissingerClassicalSimulationQuantum2022} \cite{kissingerSimulatingQuantumCircuits2022} \cite{sutcliffeFastClassicalSimulation2024}, quantum circuit
optimization \cite{cowtanPhaseGadgetSynthesis2020} \cite{debeaudrapTechniquesReducePi2020} \cite{duncanGraphtheoreticSimplificationQuantum2020} and quantum
circuit equality validation \cite{kissingerReducingNumberNonClifford2020}. The ZX-calculus has been successfully applied to the boolean and counting satisfiability problem
\cite{debeaudrapTensorNetworkRewriting2021} \cite{laakkonenGraphicalStabilizerDecompositions2022}. In \refsec{this_work_state_decomps} we will see aspects
of another research direction, concerned with the analysis of barren plateaus \cite{zhaoAnalyzingBarrenPlateau2021} \cite{wangDifferentiatingIntegratingZX2022}.
Finally, the ZX-calculus enables us to study not only problems from computer science, but also many problems from physics. For example, it has been used
in condensed matter physics to describe AKLT-states \cite{eastAKLTStatesZXDiagramsDiagrammatic2022a}. Even links to quantum field theory have been made, by 
creating a categorical description of Feynman diagrams\sidenote{We will briefly talk about quantum field theory and Feynman diagrams in \refch{weighting_algorithms}.}
\cite{shaikhCategoricalSemanticsFeynman2022}, or more recently, by introducing Fock spiders, which allow for continuous-variable quantum computation. This
could enable native simulation of quantum field theories in the future \cite{shaikhFockedupZXCalculus2024}. The ZXH-calculus has also been used to describe
$\text{SU}(2)$ representation theory, more precisely to represent Yutsis diagrams and Penrose spin-networks, which creates a connection to loop quantum gravity
\cite{eastSpinnetworksZXcalculus2022a}. Lastly, \cite{gorardZXCalculusExtendedHypergraph2020} explored the relationship between the ZX-calculus and
the Wolfram model, a fundamental theory of physics that tries to describe the fundamental structure of the universe using hypergraph
rewriting systems \cite{WolframPhysicsProject}.

\pagelayout{wide} 
\addpart{Part II: Decompositions of Non-Stabilizer States}\labpart{part_ii}
\pagelayout{margin} 

\setchapterpreamble[u]{\margintoc}
\chapter{Novel State Decompositions}
\labch{novel_state_decompositions}

In \refsec{complexity_and_simulation}, we discussed different (classical) simulation methods. The state vector approach is a widely used method.
However, when applied on $n$ qubits, this approach leads to a computational complexity of $\mathcal{O}(2^n)$. One approach that tends to yield much better
results both memory-wise and time-wise is the concept of \textbf{stabilizer decompositions}. It allows us to do strong simulation.

\section{Previous Work}
\labsec{prev_work_stab_decomp}

Consider any given quantum circuit. Furthermore, without loss of generality, assume the input to be $\ket{0}^{\otimes n}$. The idea of a stabilizer
decomposition is to decompose the output $\ket{\psi}$ into a linear combination of Clifford states\sidenote{A (non-)Clifford state is a state obtained
from a (non-)Clifford circuit (\cf \refsec{complexity_and_simulation}).} $\ket{\phi_1},\dots,\ket{\phi_k}$
\begin{equation}
    \ket{\psi}=\sum_{i=1}^k a_i \ket{\phi_i}
\end{equation}
for complex coefficients $a_i$. The probability of any output $x$ is thus given by
\begin{equation}
    \braket{x|\psi}=\sum_{i=1}^k a_i \braket{x|\phi_i}.
\end{equation}
According to the Gottesman-Knill theorem, each term $\braket{x|\phi_i}$ can be computed in polynomial time, meaning the whole strong simulation has
a complexity of $\mathcal{O}(k)$.
\begin{remark}\labremark{no_conflict_stab_decomp_remark}
    This procedure creates no conflicts with the theoretical complexity of the problem. This can be seen by decomposing a non-Clifford state into
    two Clifford states. If we then want to use this decomposition for three of the original non-Clifford states combined using a tensor product,
    we have to multiply the Clifford states individually, resulting in $2^3=8$ Clifford states. We conclude that we still have exponential scaling.
\end{remark}
Consequently, minimizing the number of terms in such decompositions has become an active area of current research \cite{codsiCuttingEdgeGraphicalStabiliser}.
Such a minimization is limited by the \textit{stabilizer rank}.
\begin{definition}
    Given an arbitrary $n$-qubit state $\psi$, we define the \textit{stabilizer rank} of $\psi$ as the smallest non-negative integer $\chi$, denoted $\chi(\psi)$,
    such that it is possible to write $\ket{\psi}=\sum_{i=1}^{\chi} a_i \ket{\phi_i}$.
\end{definition}
One widely used example of a general stabilizer decomposition is the so-called \textbf{magic state}
$\ket{T}\coloneqq\frac{1}{\sqrt{2}}(\ket{0}+e^{i\frac{\pi}{4}}\ket{1})$, which can be expressed using ZX-diagrams:
\begin{equation}
    \begin{tikzpicture}[scale=2]
	\begin{pgfonlayer}{nodelayer}
		\node [style=Z phase dot] (0) at (0, -0.012) {$\frac\pi4$};
		\node [style=none] (1) at (0.525, -0.013) {};
	\end{pgfonlayer}
	\begin{pgfonlayer}{edgelayer}
		\draw (1.center) to (0);
	\end{pgfonlayer}
\end{tikzpicture}
 = \frac1{\sqrt2}\left(\begin{tikzpicture}
	\begin{pgfonlayer}{nodelayer}
		\node [style=Z dot] (0) at (-1, 0) {};
		\node [style=Z dot] (1) at (0, 0) {};
		\node [style=none] (2) at (1, 0) {};
	\end{pgfonlayer}
	\begin{pgfonlayer}{edgelayer}
		\draw (1) to (2.center);
		\draw [style=hadamard edge] (0) to (1);
	\end{pgfonlayer}
\end{tikzpicture}

    +e^{i\frac\pi4}\begin{tikzpicture}
	\begin{pgfonlayer}{nodelayer}
		\node [style=Z phase dot] (0) at (-1, 0) {$\pi$};
		\node [style=Z dot] (1) at (0.25, 0) {};
		\node [style=none] (2) at (1.25, 0) {};
	\end{pgfonlayer}
	\begin{pgfonlayer}{edgelayer}
		\draw (1) to (2.center);
		\draw [style=hadamard edge] (0) to (1);
	\end{pgfonlayer}
\end{tikzpicture}
\right)
\end{equation}

\newpage

Notice that such states can be unfused easily, a process similar to the magic state injection protocol used in other fields apart from the ZX-calculus:
\begin{equation}
    \input{3_Novel_State_Decompositions/T-spider-to-magic-state.tikz}
\end{equation}
According to \refremark{no_conflict_stab_decomp_remark}, $\begin{tikzpicture}[scale=2]
	\begin{pgfonlayer}{nodelayer}
		\node [style=Z phase dot] (0) at (0, -0.012) {$\frac\pi4$};
		\node [style=none] (1) at (0.525, -0.013) {};
	\end{pgfonlayer}
	\begin{pgfonlayer}{edgelayer}
		\draw (1.center) to (0);
	\end{pgfonlayer}
\end{tikzpicture}
^{\otimes t}$ would result in $2^t$
terms. A more sophisticated approach for $t=2$ would be to use:
\begin{equation}\labeq{2_T_decomp}
    \input{3_Novel_State_Decompositions/2-T-decomposition.tikz}
\end{equation}
Notice that the scaling for these decompositions is given by $p^{t/r}$, where $p$ is the number of terms and $r$ represents the number of non-Clifford terms
reduced, more concretely by how much the T-count was reduced. We will measure the scaling using the number $\alpha$, defined by
\begin{equation}
    p^{t/r}=2^{\alpha t}\implies \alpha=\frac{\log_2(p)}{r}
\end{equation}
for further comparisons. In \refeq{2_T_decomp}, for example, the number $\alpha$ is equal to $\log_2(2)/2=0.5$. The objective is thus to find stabilizer decompositions with a small $\alpha$ \cite{kissingerClassicalSimulationQuantum2022}.

One such decomposition that turns out to be very useful was first introduced in \cite{bravyiTradingClassicalQuantum2016} by Bravyi, Smith and Smolin, thus often
being referred to as the \textbf{BSS decomposition}. It was later reformulated using ZX-diagrams\sidenote[][*-1]{Note that whenever the diagrams are turned clockwise for space reasons,
time evolution is from top to bottom.} in \cite{kissingerSimulatingQuantumCircuits2022}:
\begin{equation}
    \input{3_Novel_State_Decompositions/bss.tikz}
\end{equation}
Since $\alpha=\log_2(7)/6\approx 0.468$, we call it a "6-to-7" decomposition.

Interestingly, there exist many other decomposition strategies that sometimes yield even better results. For example, \textbf{partial magic state decompositions}
use non-stabilizer decompositions, meaning the terms in the linear combination might only have \textit{less} T-spiders, but not necessarily \textit{zero}:
\begin{equation}
    \input{3_Novel_State_Decompositions/magic-state-5-decomp.tikz}
\end{equation}
So effectively, this is a 4-to-3 strategy, resulting in $\alpha=\log_2(3)/4\approx 0.396$ \cite{kissingerClassicalSimulationQuantum2022}.

Another approach introduced in \cite{kochSpeedyContractionZX2023a} focuses on triangles instead of T-spiders:
\begin{align}
\begin{split}
    \begin{tikzpicture}
	\begin{pgfonlayer}{nodelayer}
		\node [style=none] (0) at (-0.75, 0) {};
		\node [style=none] (1) at (0.75, 0) {};
		\node [style={l_triang}] (2) at (0, 0) {};
	\end{pgfonlayer}
	\begin{pgfonlayer}{edgelayer}
		\draw (0.center) to (1.center);
	\end{pgfonlayer}
\end{tikzpicture}
 
	~&=~ \begin{pmatrix} 1 & 1 \\ 0 & 1 \end{pmatrix}
	\\
	\begin{tikzpicture}
	\begin{pgfonlayer}{nodelayer}
		\node [style=none] (0) at (-0.75, 0) {};
		\node [style=none] (1) at (0.75, 0) {};
		\node [style={r_triang}] (2) at (0, 0) {};
	\end{pgfonlayer}
	\begin{pgfonlayer}{edgelayer}
		\draw (0.center) to (1.center);
	\end{pgfonlayer}
\end{tikzpicture}

	~&=~ \input{3_Novel_State_Decompositions/triangletransp2.tikz}
	~=~ \begin{pmatrix}	1 & 0 \\ 1 & 1 \end{pmatrix}
	\\
	\begin{tikzpicture}
	\begin{pgfonlayer}{nodelayer}
		\node [style=none] (0) at (-0.75, 0) {};
		\node [style=none] (1) at (1, 0) {};
		\node [style=none] (3) at (0.25, 0.45) {\scriptsize $-1$};
		\node [style={l_triang}] (4) at (0, 0) {};
	\end{pgfonlayer}
	\begin{pgfonlayer}{edgelayer}
		\draw (0.center) to (1.center);
	\end{pgfonlayer}
\end{tikzpicture}
 
	~&=~ \input{3_Novel_State_Decompositions/triangleinv.tikz} 
	~=~ \begin{pmatrix} 1 & -1 \\ 0 & 1 \end{pmatrix}
\end{split}
\end{align}
However, since these triangles essentially create directed graphs, one can use the symmetric \textbf{star} instead:
\begin{equation}\labeq{star_triang_relation}
    \begin{tikzpicture}
	\begin{pgfonlayer}{nodelayer}
		\node [style=none] (0) at (-0.75, 0) {};
		\node [style=none] (1) at (0.75, 0) {};
		\node [style={h_star}] (2) at (0, 0) {};
	\end{pgfonlayer}
	\begin{pgfonlayer}{edgelayer}
		\draw (0.center) to (1.center);
	\end{pgfonlayer}
\end{tikzpicture}
 
	~:=~ \begin{tikzpicture}
	\begin{pgfonlayer}{nodelayer}
		\node [style=none] (0) at (-0.75, 0) {};
		\node [style=none] (1) at (2, 0) {};
		\node [style=X phase dot] (3) at (1, 0) {$\pi$};
		\node [style={r_triang}] (4) at (0, 0) {};
	\end{pgfonlayer}
	\begin{pgfonlayer}{edgelayer}
		\draw (0.center) to (1.center);
	\end{pgfonlayer}
\end{tikzpicture}
 
	~=~ \begin{tikzpicture}
	\begin{pgfonlayer}{nodelayer}
		\node [style=none] (0) at (-1, 0) {};
		\node [style=none] (1) at (1.75, 0) {};
		\node [style=X phase dot] (3) at (0, 0) {$\pi$};
		\node [style={l_triang}] (4) at (1, 0) {};
	\end{pgfonlayer}
	\begin{pgfonlayer}{edgelayer}
		\draw (0.center) to (1.center);
	\end{pgfonlayer}
\end{tikzpicture}
 
	~=~ \begin{pmatrix}
			1 & 1 \\
			1 & 0
		\end{pmatrix}
\end{equation}
Similar as for Hadamard edges, we will use an orange dashed line to represent a \textbf{star edge}:
\begin{equation}
    \begin{tikzpicture}
	\begin{pgfonlayer}{nodelayer}
		\node [style=none] (0) at (-1.25, 0) {};
		\node [style=none] (1) at (1.25, 0) {};
		\node [style={h_star}] (2) at (0, 0) {};
	\end{pgfonlayer}
	\begin{pgfonlayer}{edgelayer}
		\draw (0.center) to (1.center);
	\end{pgfonlayer}
\end{tikzpicture}
 \quad\rightsquigarrow\quad \begin{tikzpicture}
	\begin{pgfonlayer}{nodelayer}
		\node [style=none] (0) at (0, 0) {};
		\node [style=none] (1) at (2, 0) {};
	\end{pgfonlayer}
	\begin{pgfonlayer}{edgelayer}
		\draw [style=star edge] (0.center) to (1.center);
	\end{pgfonlayer}
\end{tikzpicture}

\end{equation}
Technically, we could write triangles in terms of four T-spiders:
\begin{equation}
    \begin{tikzpicture}
	\begin{pgfonlayer}{nodelayer}
		\node [style=none] (1) at (-1, 0) {};
		\node [style=none] (2) at (1, 0) {};
		\node [style={r_triang}] (3) at (0, 0) {};
	\end{pgfonlayer}
	\begin{pgfonlayer}{edgelayer}
		\draw (1.center) to (2.center);
	\end{pgfonlayer}
\end{tikzpicture}
 ~=~ \textstyle{\frac{1}{2}}~ \input{3_Novel_State_Decompositions/intro_tri-2.tikz}
\end{equation}
This is suboptimal, since given $n$ T-spiders and $m$ triangles, the scaling is $2^{\alpha n + \beta m}$, where $\beta=4\alpha$. This results in a value
greater than $\beta=1$, which would already result from its naive decomposition. What will follow now is this decomposition, along with more efficient ones,
but formulated using star edges instead of triangles\sidenote[][*-1]{We will refer to \refeq{edge_decomp_1}, \refeq{edge_decomp_2} and \refeq{edge_decomp_1} as \textit{star decompositions}.}:

\begin{align}
    \labeq{edge_decomp_1}
	\begin{tikzpicture}
	\begin{pgfonlayer}{nodelayer}
		\node [style=none] (0) at (0, 1) {};
		\node [style=none] (1) at (0, -1) {};
	\end{pgfonlayer}
	\begin{pgfonlayer}{edgelayer}
		\draw [style=star edge] (0.center) to (1.center);
	\end{pgfonlayer}
\end{tikzpicture}
 ~=&~ \textstyle{
	\sqrt 2~~ \begin{tikzpicture}
	\begin{pgfonlayer}{nodelayer}
		\node [style=Z dot] (0) at (0, 0.5) {};
		\node [style=Z dot] (1) at (0, -0.5) {};
		\node [style=none] (2) at (0, 1.25) {};
		\node [style=none] (3) at (0, -1.25) {};
	\end{pgfonlayer}
	\begin{pgfonlayer}{edgelayer}
		\draw (3.center) to (1);
		\draw [style=hadamard edge] (0) to (2.center);
	\end{pgfonlayer}
\end{tikzpicture}
 ~~+~~ 
	2~~ \begin{tikzpicture}
	\begin{pgfonlayer}{nodelayer}
		\node [style=Z dot] (0) at (0, -0.5) {};
		\node [style=Z phase dot] (1) at (0, 0.5) {$\pi$};
		\node [style=none] (2) at (0, 1.25) {};
		\node [style=none] (3) at (0, -1.25) {};
	\end{pgfonlayer}
	\begin{pgfonlayer}{edgelayer}
		\draw [style=hadamard edge] (3.center) to (0);
		\draw [style=hadamard edge] (1) to (2.center);
	\end{pgfonlayer}
\end{tikzpicture}
}
    \\[5pt]
    \labeq{edge_decomp_2}
	\begin{tikzpicture}
	\begin{pgfonlayer}{nodelayer}
		\node [style=none] (0) at (0, 1) {};
		\node [style=none] (1) at (0, -1) {};
		\node [style=none] (3) at (0.75, 1) {};
		\node [style=none] (4) at (0.75, -1) {};
	\end{pgfonlayer}
	\begin{pgfonlayer}{edgelayer}
		\draw [style=star edge] (0.center) to (1.center);
		\draw [style=star edge] (3.center) to (4.center);
	\end{pgfonlayer}
\end{tikzpicture}
 ~=&~ \textstyle{
		\frac{1}{\sqrt 2}~ \begin{tikzpicture}
	\begin{pgfonlayer}{nodelayer}
		\node [style=Z dot] (0) at (0.75, 0.5) {};
		\node [style=Z dot] (1) at (0.75, -0.5) {};
		\node [style=none] (2) at (0.75, 1.25) {};
		\node [style=none] (3) at (0.75, -1.25) {};
		\node [style=none] (4) at (0, 1.25) {};
		\node [style=none] (5) at (0, -1.25) {};
	\end{pgfonlayer}
	\begin{pgfonlayer}{edgelayer}
		\draw (3.center) to (1);
		\draw (0) to (2.center);
		\draw [style=hadamard edge] (5.center) to (4.center);
	\end{pgfonlayer}
\end{tikzpicture}
 ~~+~~
		\frac{1}{\sqrt 2}~ \begin{tikzpicture}
	\begin{pgfonlayer}{nodelayer}
		\node [style=Z dot] (0) at (0, 0.5) {};
		\node [style=Z dot] (1) at (0, -0.5) {};
		\node [style=none] (2) at (0, 1.25) {};
		\node [style=none] (3) at (0, -1.25) {};
		\node [style=none] (4) at (0.75, 1.25) {};
		\node [style=none] (5) at (0.75, -1.25) {};
	\end{pgfonlayer}
	\begin{pgfonlayer}{edgelayer}
		\draw (3.center) to (1);
		\draw (0) to (2.center);
		\draw [style=hadamard edge] (5.center) to (4.center);
	\end{pgfonlayer}
\end{tikzpicture}
 ~~+~ ~
		4~~ \begin{tikzpicture}
	\begin{pgfonlayer}{nodelayer}
		\node [style=Z phase dot] (0) at (0, 0.5) {$\pi$};
		\node [style=Z phase dot] (1) at (1, 0.5) {$\pi$};
		\node [style=Z phase dot] (2) at (0, -0.5) {$\pi$};
		\node [style=Z phase dot] (3) at (1, -0.5) {$\pi$};
		\node [style=none] (4) at (0, 1.25) {};
		\node [style=none] (5) at (1, 1.25) {};
		\node [style=none] (6) at (1, -1.25) {};
		\node [style=none] (7) at (0, -1.25) {};
	\end{pgfonlayer}
	\begin{pgfonlayer}{edgelayer}
		\draw [style=hadamard edge] (7.center) to (2);
		\draw [style=hadamard edge] (0) to (4.center);
		\draw [style=hadamard edge] (5.center) to (1);
		\draw [style=hadamard edge] (6.center) to (3);
	\end{pgfonlayer}
\end{tikzpicture}
}
	\\[5pt]
    \labeq{edge_decomp_3}
	\input{3_Novel_State_Decompositions/star-3/star.tikz} ~=&~ \textstyle{
		\frac{1}{2\sqrt 2}~~ \input{3_Novel_State_Decompositions/star-3/1.tikz} ~~+~~
		\frac{1}{2\sqrt 2}~~ \input{3_Novel_State_Decompositions/star-3/2.tikz} ~~+~~
		\frac{1}{2\sqrt 2}~~ \input{3_Novel_State_Decompositions/star-3/3.tikz}}
        \\
        \notag
        & \textstyle{~+~~ \frac{1}{\sqrt 2}~~ \input{3_Novel_State_Decompositions/star-3/4.tikz} ~~+~~
		8~~ \input{3_Novel_State_Decompositions/star-3/5.tikz}}
\end{align}
The corresponding scalings for \refeq{edge_decomp_1}, \refeq{edge_decomp_2} and \refeq{edge_decomp_3} are given by $\beta=\log_2(2)/1=1$,
$\beta=\log_2(3)/2\approx0.792$ and $\beta=\log_2(5)/3\approx0.774$, respectively. This can be further improved by considering special
cases, namely when star edges are connected to Z-spiders\sidenote{We will refer to such states as \textit{star states}.}, creating \textbf{non-stabilizer states}, which can be decomposed into \textbf{stabilizer states}.

\begin{align}
	\input{3_Novel_State_Decompositions/star-3-state/star-0.tikz} ~=&~ \textstyle{
	3~~ \input{3_Novel_State_Decompositions/star-3-state/1.tikz} ~~-~~
	\input{3_Novel_State_Decompositions/star-3-state/2.tikz} ~~+~~
	\frac{3}{\sqrt{2}}~~ \input{3_Novel_State_Decompositions/star-3-state/3.tikz} ~~-~~
	\frac{3}{2\sqrt{2}}~~ \input{3_Novel_State_Decompositions/star-3-state/4.tikz}}
	\\[5pt]
	\input{3_Novel_State_Decompositions/star-3-state/star-pi2.tikz} ~=&~ \textstyle{
	\frac{1 \pm 3i}{2}~~ \input{3_Novel_State_Decompositions/star-3-state/1.tikz} ~~+~~
	\frac{1 \mp i}{2}~~ \input{3_Novel_State_Decompositions/star-3-state/2.tikz} ~~-~~
	\frac{3-i}{2\sqrt{2}}~~ \input{3_Novel_State_Decompositions/star-3-state/3.tikz}}
    \\
    \notag
    & \textstyle{~~+~~ \frac{1 \mp i}{2\sqrt{2}}~~ \input{3_Novel_State_Decompositions/star-3-state/4.tikz}}
\end{align}
The corresponding scaling here is $\beta=\log_2(4)/3\approx0.667$ for both decompositions. Note that we do need to consider a $\pi$-phase, since
the resulting state is in fact Clifford:
\begin{equation}
    \begin{tikzpicture}
	\begin{pgfonlayer}{nodelayer}
		\node [style=Z phase dot] (0) at (0, 0) {$\pi$};
		\node [style=none] (1) at (1.5, 0) {};
	\end{pgfonlayer}
	\begin{pgfonlayer}{edgelayer}
		\draw [style=star edge] (1.center) to (0);
	\end{pgfonlayer}
\end{tikzpicture}
 ~=~ \frac{1}{\sqrt 2}~ \begin{tikzpicture}
	\begin{pgfonlayer}{nodelayer}
		\node [style=Z phase dot] (0) at (0, 0) {$\pi$};
		\node [style=none] (1) at (1.5, 0) {};
	\end{pgfonlayer}
	\begin{pgfonlayer}{edgelayer}
		\draw [style=hadamard edge] (1.center) to (0);
	\end{pgfonlayer}
\end{tikzpicture}

\end{equation}
At this point, it is worth mentioning that although this \textit{star formalism} was introduced in \cite{kochSpeedyContractionZX2023a}, the underlying
linear map is equivalent up to a scalar to the zero-labelled H-box in the ZH-calculus:
\begin{equation}\labeq{star_H_box_equiv}
    \input{3_Novel_State_Decompositions/star_H_box_equiv.tikz}
\end{equation}
In fact, such stabilizer decompositions utilizing the H-box formalism had already been introduced in \cite{laakkonenGraphicalStabilizerDecompositions2022},
most notably they included one decomposition\sidenote{Note that they use a slightly different, but equivalent notation to the more traditional notation
described in \refsec{related_calculi}.} similar to the stabilizer state decompositions from before, but with a better scaling
of $\beta=\log_2(5)/4\approx0.580$, which can be seen in \reffig{sat_4_state_decomp}:
\begin{figure}[!ht]
	\includegraphics[width=\textwidth]{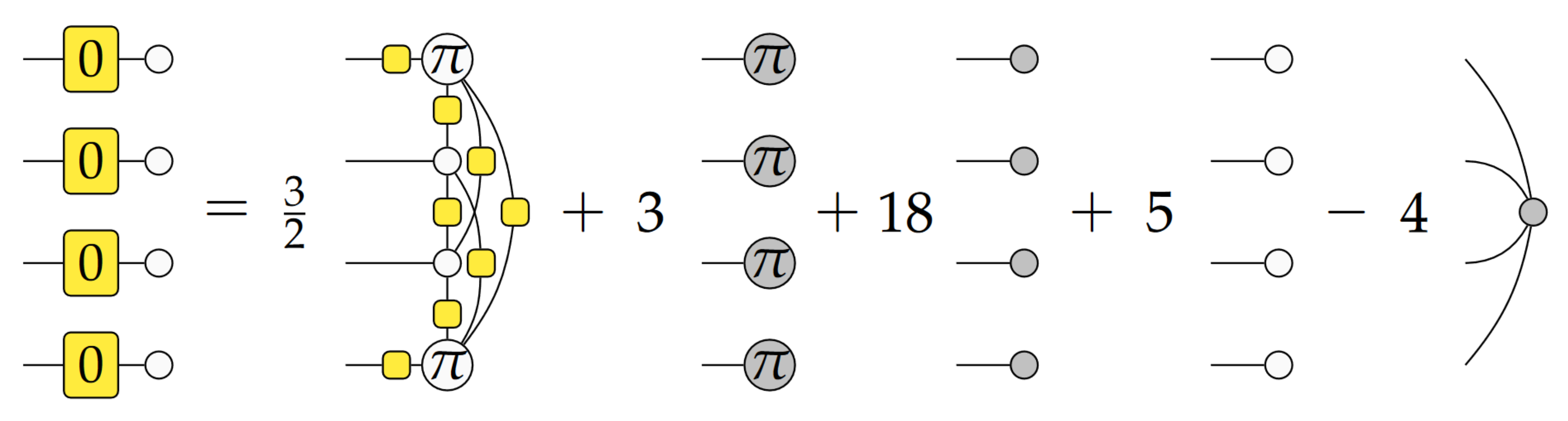}
	\caption[SAT 4 state decomp]{Four states tensored together with their according decomposition taken from \cite{laakkonenGraphicalStabilizerDecompositions2022}.}
	\labfig{sat_4_state_decomp}
\end{figure}

\section{This Work}
\labsec{this_work_state_decomps}

One of the first experiments we conducted after having started learning about the ZX-calculus was an attempt at proving certain theorems required to check
the validity of a quantum query algorithm (\cf \refsec{application_multiloop_feynman_diagrams}) introduced in an undergraduate thesis\sidenote{As it has not been published yet, we cannot cite it here.}. It involved
the calculation of the expectation value $\mathbb{E}_{\theta}(f)$
given by the integral
\begin{equation}
    \mathbb{E}_{\theta}(f)=\left(\frac{1}{2 \pi}\right)^n \int_{[0,2 \pi]^n} f(\theta) d \theta,
\end{equation}
where $n$ is the number of qubits and $f$ is a placeholder for a matrix composed of permutation and rotation matrices, which are dependent on $\theta$.

Thus having found \cite{wangDifferentiatingIntegratingZX2022} proved beneficial, as it included a theorem regarding \textit{diagrammatic integration}.
Later, however, we developed a classical algorithm to solve the same problem as the quantum query algorithm. Even for simple starting configurations,
the classical algorithm heavily outperformed the quantum algorithm. Of course, the theoretical study could still be of interest, but it was
decided to investigate another, more practical aspect of \cite{wangDifferentiatingIntegratingZX2022}, namely the concept of \textit{barren plateau
detection}.

The work in \cite{wangDifferentiatingIntegratingZX2022} introduced a theorem regarding \textit{diagrammatic variance calculation}\sidenote{The variance
can be calculated using the expectation value, which again can be calculated using diagrammatic integration.}. It shows that given a certain type
of Hamiltonian $H$, and therefore also its expectation value $E_H$, the equality in \refeq{diagrammatic_var_calc} holds. Note that since $E_H$ is
not only dependent on the Hamiltonian, but also on the \textit{circuit ansatz} $U(\bm{\theta})$ for the $n$-tuple $\bm{\theta}=(\theta_1,\theta_2,\dots)$, the
resulting ZX-diagram contains spiders with phase $\theta_j$.
\begin{equation}\labeq{diagrammatic_var_calc}
    \mathbf{Var}\left(\frac{ \partial \braket{H}}{\partial\theta_j}\right)= \ \raisebox{-0.45\height}{\includegraphics[height=5.2cm]{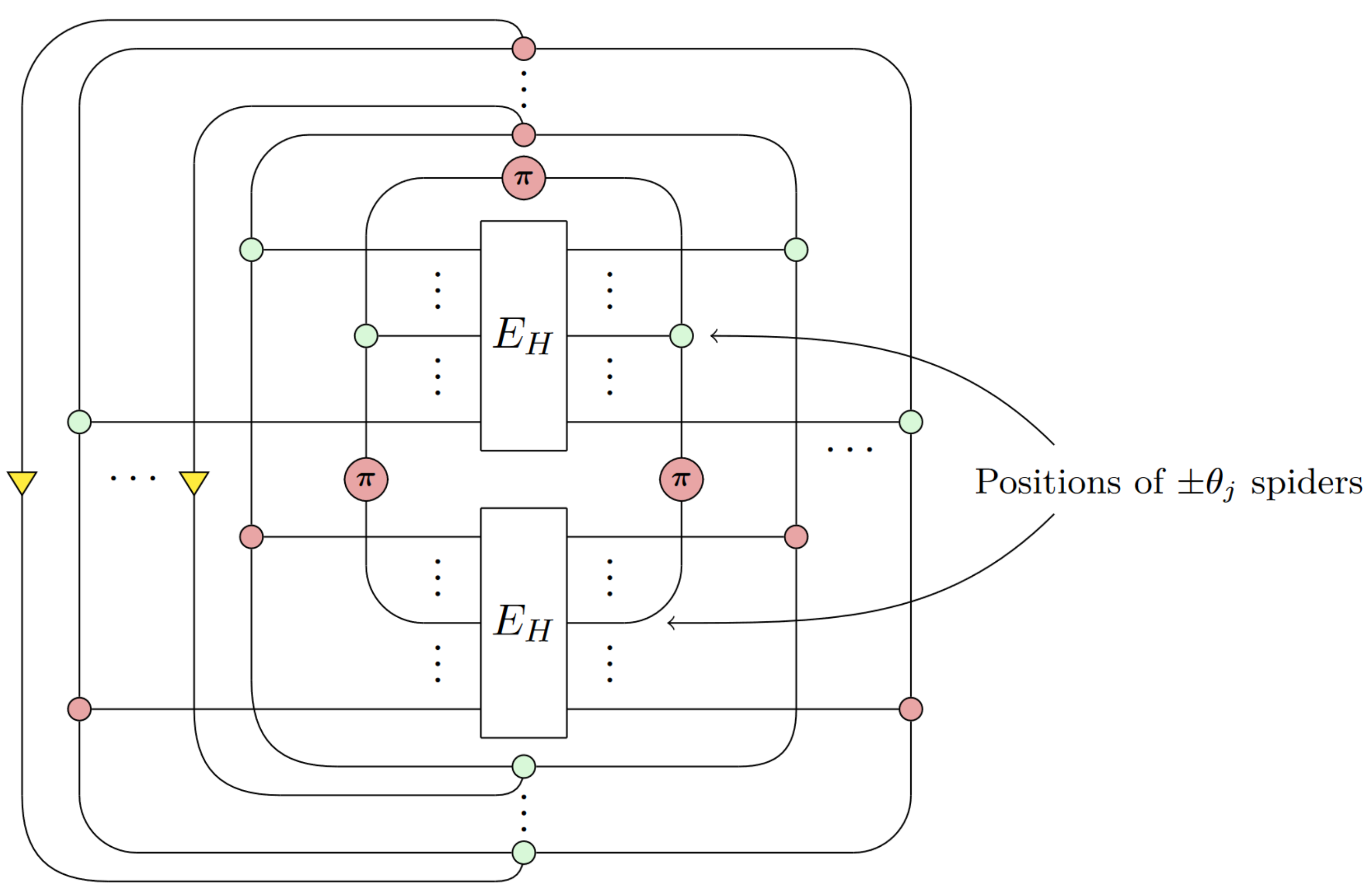}}
\end{equation}
The reason for the appearance of the Hamiltonian is that the type of quantum circuits being studied here is part of a field called \textit{quantum machine learning}.
In \textit{quantum chemistry}, a \textit{circuit ansatz} might be used to find the ground state of a Hamiltonian. Such ansätze are often optimized
using gradient-based techniques, much like how it is done in classical machine learning. However, it has been shown that many ansätze suffer
from an exponentially flat training landscape, making gradient descent impossible. This type of landscape is referred to as \textit{barren plateau}.
One common approach to detect barren plateaus is to check whether $\mathbf{Var}\left(\frac{ \partial \braket{H}}{\partial\theta_j}\right)\approx0$,
in which case a barren plateau is likely to start and thus disturb the learning progression \cite{mccleanBarrenPlateausQuantum2018a}.

For the next project, we attempted to brute force linear combinations resulting in decompositions that could be used to simulate \refeq{diagrammatic_var_calc}.
More precisely, the goal was to decompose the non-Clifford \textit{cycle}:
\begin{equation}
	\input{3_Novel_State_Decompositions/variance_cycle.tikz}
\end{equation}
For this purpose, a big set of potential decomposition terms was generated, hence also obtaining the according matrices. These matrices were
combined in a linear combination in Mathematica\sidenote{The code is very similar to the implementation in
\cite{Tix3DevNovel_star_state_decompositions}.} \cite{Mathematica}. By a process of elimination, a decomposition consisting of two terms was
indeed found\sidenote{In retrospect, it is clear that we had basically found the standard decomposition of triangles/star edges as in \refeq{edge_decomp_1}, only
that it had been embedded in the diagram of the cycle. This was not obvious at the time, as we did not know of these decompositions.}. One could thus compute the variance using a sum of $2^{m-1}$ terms, where $m$ is the number of parameters in the considered ansatz.
At the time this seemed to be a nice improvement over the approach introduced in \cite{zhaoAnalyzingBarrenPlateau2021}, which resulted in a sum
of $3^{m-1}$ terms.

It was only later that we realized that \cite{kochSpeedyContractionZX2023a} had already introduced a better approach utilizing decompositions
of star edges as described in \refsec{prev_work_stab_decomp}, resulting in a scaling of $2^{\beta(m-1)}$.

Since finding better decompositions than the ones introduced in \cite{kochSpeedyContractionZX2023a} and \cite{laakkonenGraphicalStabilizerDecompositions2022}
would most likely require a better approach than our brute force method, we investigated the techniques that were originally used to find
such decompositions. As it turns out, there exists an open-source implementation \cite{laakkonenTuomas56Cliffs2024} of the technique used
in \cite{bravyiTradingClassicalQuantum2016} to find the (asymptotic) 6-to-7 stabilizer decomposition.

The algorithm uses \textit{simulated annealing} to search over the solution space of the real Cliffords\sidenote{It can be proven that the optimal
decomposition of a real state consists only of real Cliffords \cite{laakkonenTuomas56Cliffs2024}.}. Simulated annealing is a probabilistic optimization algorithm that explores
the solution space using a random walk and probabilistically accepting worse solutions in order to escape local optima, whilst gradually
reducing this acceptance probability, aiming to converge on a global optimum \cite{bravyiTradingClassicalQuantum2016} \cite{SimulatedAnnealing2024}.

By configuring the program to search for specific states using different annealing parameters, we were in fact able to find novel decompositions.
It is worth noting that the parameters are crucial, since the solution is easily overlooked when reducing the acceptance probability too
quickly or sometimes even too slowly. As the program did not provide the coefficients of the decomposition, we created a small routine to find
them. This routine, together with a program to check the validity of the decompositions, can be found on GitHub \cite{Tix3DevNovel_star_state_decompositions}.

These efforts led to the following five novel decompositions\sidenote{The coefficients are written separately for space reasons.}:

\begin{align}\labeq{5_state_phase_0}
  \input{3_Novel_State_Decompositions/my_novel_decomps/5_states_phase_0/star.tikz} ~ & =~ \textstyle{
    \xi_1~~ \input{3_Novel_State_Decompositions/my_novel_decomps/5_states_phase_0/0.tikz} ~~+~~
    \xi_2~~ \input{3_Novel_State_Decompositions/my_novel_decomps/5_states_phase_0/1.tikz} ~~+~~
    \xi_3~~ \input{3_Novel_State_Decompositions/my_novel_decomps/5_states_phase_0/2.tikz}}
  \\[15pt]
  \notag
                                                                                          & \textstyle{~~+~~
    \xi_4~~ \input{3_Novel_State_Decompositions/my_novel_decomps/5_states_phase_0/3.tikz} ~~+~~
    \xi_5~~ \input{3_Novel_State_Decompositions/my_novel_decomps/5_states_phase_0/4.tikz} ~~+~~
    \xi_6~~ \input{3_Novel_State_Decompositions/my_novel_decomps/5_states_phase_0/5.tikz}}
\end{align}
\begin{equation*}
	\xi_1=-192,\ \xi_2=\frac{15\sqrt{2}}{8},\ \xi_3=10\sqrt{2},\ \xi_4=20\sqrt{2},\ \xi_5=48\sqrt{2},\ \xi_6=15
\end{equation*}
\refeq{5_state_phase_0} yields $\beta=\frac{\log_2{6}}{5} \approx 0.517$.
\begin{align}\labeq{5_state_phase_pi_ov_2}
  \input{3_Novel_State_Decompositions/my_novel_decomps/5_states_phase_pi_ov_2_and_minus/star.tikz} ~ & =~ \textstyle{
    \xi_1~~ \input{3_Novel_State_Decompositions/my_novel_decomps/5_states_phase_pi_ov_2_and_minus/0.tikz} ~~+~~
    \xi_2~~ \input{3_Novel_State_Decompositions/my_novel_decomps/5_states_phase_pi_ov_2_and_minus/1.tikz} ~~+~~
    \xi_3~~ \input{3_Novel_State_Decompositions/my_novel_decomps/5_states_phase_pi_ov_2_and_minus/2.tikz}}
  \\[15pt]
  \notag
                                                                                                          & \textstyle{~~+~~
    \xi_4~~ \input{3_Novel_State_Decompositions/my_novel_decomps/5_states_phase_pi_ov_2_and_minus/3.tikz} ~~+~~
    \xi_5~~ \input{3_Novel_State_Decompositions/my_novel_decomps/5_states_phase_pi_ov_2_and_minus/4.tikz} ~~+~~
	\xi_6~~ \input{3_Novel_State_Decompositions/my_novel_decomps/5_states_phase_pi_ov_2_and_minus/5.tikz}}
\end{align}
\begin{align*}
    &\xi_1 = -\frac{5i}{4\sqrt{2}}, \ \xi_2 = 56+8i, \ \xi_3 = -5, \ \xi_4 = -(64+32i), \ \xi_5 = -(7\sqrt{2}-\sqrt{2}i), \\
	&\xi_6 = -(16-48i)
\end{align*}
\refeq{5_state_phase_pi_ov_2} yields $\beta=\frac{\log_2{6}}{5} \approx 0.517$.

\newpage

\begin{align}\labeq{5_state_phase_pi_ov_2_minus}
	\input{3_Novel_State_Decompositions/my_novel_decomps/5_states_phase_pi_ov_2_and_minus/star_minus.tikz} ~ & =~ \textstyle{
	  \xi_1~~ \input{3_Novel_State_Decompositions/my_novel_decomps/5_states_phase_pi_ov_2_and_minus/0.tikz} ~~+~~
	  \xi_2~~ \input{3_Novel_State_Decompositions/my_novel_decomps/5_states_phase_pi_ov_2_and_minus/1.tikz} ~~+~~
	  \xi_3~~ \input{3_Novel_State_Decompositions/my_novel_decomps/5_states_phase_pi_ov_2_and_minus/2.tikz}}
	\\[15pt]
	\notag
																												  & \textstyle{~~+~~
	  \xi_4~~ \input{3_Novel_State_Decompositions/my_novel_decomps/5_states_phase_pi_ov_2_and_minus/3.tikz} ~~+~~
	  \xi_5~~ \input{3_Novel_State_Decompositions/my_novel_decomps/5_states_phase_pi_ov_2_and_minus/4.tikz} ~~+~~
	  \xi_6~~ \input{3_Novel_State_Decompositions/my_novel_decomps/5_states_phase_pi_ov_2_and_minus/5.tikz}}
\end{align}
\begin{align*}
    &\xi_1 = \frac{5i}{4\sqrt{2}}, \ \xi_2 = 56-8i, \ \xi_3 = -5, \ \xi_4 = -(64-32i), \ \xi_5 = -(\sqrt{2}+3\sqrt{2}i), \\
	&\xi_6 = -(16+48i)
\end{align*}
\refeq{5_state_phase_pi_ov_2_minus} yields $\beta=\frac{\log_2{6}}{5} \approx 0.517$.

\begin{align}\labeq{4_state_phase_pi_ov_2}
  \input{3_Novel_State_Decompositions/my_novel_decomps/4_states_phase_pi_ov_2_and_minus/star.tikz} ~ & =~ \textstyle{
    \xi_1~~ \input{3_Novel_State_Decompositions/my_novel_decomps/4_states_phase_pi_ov_2_and_minus/0.tikz} ~~+~~
    \xi_2~~ \input{3_Novel_State_Decompositions/my_novel_decomps/4_states_phase_pi_ov_2_and_minus/1.tikz} ~~+~~
    \xi_3~~ \input{3_Novel_State_Decompositions/my_novel_decomps/4_states_phase_pi_ov_2_and_minus/2.tikz}}
  \\[15pt]
  \notag
                                                                                                          & \textstyle{~~+~~
    \xi_4~~ \input{3_Novel_State_Decompositions/my_novel_decomps/4_states_phase_pi_ov_2_and_minus/3.tikz} ~~+~~
    \xi_5~~ \input{3_Novel_State_Decompositions/my_novel_decomps/4_states_phase_pi_ov_2_and_minus/4.tikz}}
\end{align}
\begin{equation*}
    \xi_1 = -6+2i, \ \xi_2 = -\frac{5+5i}{2}, \ \xi_3 = -(3\sqrt{2}-\sqrt{2}i), \ \xi_4 = -(6-18i), \ \xi_5 = \frac{7+9i}{2}
\end{equation*}
\refeq{4_state_phase_pi_ov_2} yields $\beta=\frac{\log_2{5}}{4} \approx 0.580$.
\begin{align}\labeq{4_state_phase_pi_ov_2_minus}
  \input{3_Novel_State_Decompositions/my_novel_decomps/4_states_phase_pi_ov_2_and_minus/star_minus.tikz} ~ & =~ \textstyle{
    \xi_1~~ \input{3_Novel_State_Decompositions/my_novel_decomps/4_states_phase_pi_ov_2_and_minus/0.tikz} ~~+~~
    \xi_2~~ \input{3_Novel_State_Decompositions/my_novel_decomps/4_states_phase_pi_ov_2_and_minus/1.tikz} ~~+~~
    \xi_3~~ \input{3_Novel_State_Decompositions/my_novel_decomps/4_states_phase_pi_ov_2_and_minus/2.tikz}}
  \\[15pt]
  \notag
                                                                                                                & \textstyle{~~+~~
    \xi_4~~ \input{3_Novel_State_Decompositions/my_novel_decomps/4_states_phase_pi_ov_2_and_minus/3.tikz} ~~+~~
    \xi_5~~ \input{3_Novel_State_Decompositions/my_novel_decomps/4_states_phase_pi_ov_2_and_minus/4.tikz}}
\end{align}
\begin{equation*}
    \xi_1 = -6+18i, \ \xi_2 = -\frac{5+5i}{2}, \ \xi_3 = -(3\sqrt{2}+11\sqrt{2}i), \ \xi_4 = -(6-2i), \ \xi_5 = -\frac{1+3i}{2}
\end{equation*}
\refeq{4_state_phase_pi_ov_2_minus} yields $\beta=\frac{\log_2{5}}{4} \approx 0.580$.

As of writing this thesis, we have not come across comparable decompositions.

To test the limitations of this approach, the program was run on an AWS EC2 C6a instance (c6a.32xlarge) \cite{AmazonEC2C6a} to search for
more complex decompositions. Although using parallelization on 128 vCPUs, together with 256 GiB of memory, we did not find any novel
decompositions. It is presumed that one fundamental limitation of this approach is that it either finds a result relatively quickly,
or it never finds one, thus not being able to take advantage of high-performance computers.

The practicality of the novel decompositions alongside with the star edge decompositions listed in \refsec{prev_work_stab_decomp} will
be discussed in \refsec{comparison_with_state_decomps}.

\pagelayout{wide} 
\addpart{Part III: Multi-Control Toffoli Gate Dense Quantum Circuits}\labpart{part_iii}
\pagelayout{margin} 

\setchapterpreamble[u]{\margintoc}
\chapter{Dynamic Decompositions}
\labch{dynamic_decompositions}

In \refch{novel_state_decompositions}, we discussed state decompositions. Mathematically, these correspond to decomposing a \textit{vector} into
a sum of other \textit{vectors}. In the following, we will demonstrate that we can also decompose \textit{matrices}.

\section{Overview}

Using the concept of decompositions for quantum computation may seem like a rather exotic idea, however, one of the most
elementary definitions already allows us to obtain a very useful decomposition:
\begin{theorem}\labthm{elementary_decomp}
    \begin{equation}
        \input{2_Graphical_Languages/z_spider_def.tikz} = \frac{1}{\sqrt{2}^{n+m}} \ \begin{tikzpicture}
	\begin{pgfonlayer}{nodelayer}
		\node [style=none] (0) at (1.25, -1) {};
		\node [style=none] (1) at (1.25, 1) {};
		\node [style=none] (2) at (-1.25, -1) {};
		\node [style=none] (3) at (-1.25, 1) {};
		\node [style=none, rotate=90] (4) at (-1, 0) {$...$};
		\node [style=none, rotate=90] (5) at (1, 0) {$...$};
		\node [style=none] (6) at (-1.75, -1) {};
		\node [style=none] (7) at (-1.75, 1) {};
		\node [style=none] (8) at (1.75, 1) {};
		\node [style=none] (9) at (1.75, -1) {};
		\node [style=none] (10) at (-2.25, 0) {$m$};
		\node [style=none] (11) at (2.25, 0) {$n$};
		\node [style=X dot] (16) at (0.5, -0.5) {};
		\node [style=X dot] (17) at (-0.5, -0.5) {};
		\node [style=X dot] (18) at (-0.5, 0.5) {};
		\node [style=X dot] (19) at (0.5, 0.5) {};
	\end{pgfonlayer}
	\begin{pgfonlayer}{edgelayer}
		\draw [style=brace edge, rotate=180] (6.center) to (7.center);
		\draw [style=brace edge] (8.center) to (9.center);
		\draw (19) to (1.center);
		\draw (16) to (0.center);
		\draw (17) to (2.center);
		\draw (18) to (3.center);
	\end{pgfonlayer}
\end{tikzpicture}
 +
        \frac{e^{i\alpha}}{\sqrt{2}^{n+m}} \ \begin{tikzpicture}
	\begin{pgfonlayer}{nodelayer}
		\node [style=none] (1) at (1.25, -1) {};
		\node [style=none] (2) at (1.25, 1) {};
		\node [style=none] (3) at (-1.25, -1) {};
		\node [style=none] (4) at (-1.25, 1) {};
		\node [style=none, rotate=90] (5) at (-1, 0) {$...$};
		\node [style=none, rotate=90] (6) at (1, 0) {$...$};
		\node [style=none] (7) at (-1.75, -1) {};
		\node [style=none] (8) at (-1.75, 1) {};
		\node [style=none] (9) at (1.75, 1) {};
		\node [style=none] (10) at (1.75, -1) {};
		\node [style=none] (11) at (-2.25, 0) {$m$};
		\node [style=none] (12) at (2.25, 0) {$n$};
		\node [style=X phase dot] (13) at (-0.5, -0.5) {$\pi$};
		\node [style=X phase dot] (14) at (0.5, -0.5) {$\pi$};
		\node [style=X phase dot] (15) at (0.5, 0.5) {$\pi$};
		\node [style=X phase dot] (16) at (-0.5, 0.5) {$\pi$};
	\end{pgfonlayer}
	\begin{pgfonlayer}{edgelayer}
		\draw [style=brace edge, rotate=180] (7.center) to (8.center);
		\draw [style=brace edge] (9.center) to (10.center);
		\draw (15) to (2.center);
		\draw (14) to (1.center);
		\draw (13) to (3.center);
		\draw (4.center) to (16);
	\end{pgfonlayer}
\end{tikzpicture}

    \end{equation}
\end{theorem}
\begin{proof}
    From \refeq{z_spider_eq} we have
    \begin{equation}\labeq{z_spider_eq_redef}
        \input{2_Graphical_Languages/z_spider_def.tikz} \ = \ \ket{0}^{\otimes n}\bra{0}^{\otimes m}+e^{i\alpha}\ket{1}^{\otimes n}\bra{1}^{\otimes m}.
    \end{equation}
    Furthermore, we know that
    \begin{equation}\labeq{x_state_def}
        \ket{0}=\frac{1}{\sqrt{2}}(\ket{+}+\ket{-})=\frac{1}{\sqrt{2}} \ \begin{tikzpicture}
	\begin{pgfonlayer}{nodelayer}
		\node [style=none] (0) at (0.5, 0) {};
		\node [style=X dot] (1) at (-0.5, 0) {};
	\end{pgfonlayer}
	\begin{pgfonlayer}{edgelayer}
		\draw (1) to (0.center);
	\end{pgfonlayer}
\end{tikzpicture}

    \end{equation}
    and
    \begin{equation}\labeq{x_pi_state_def}
        \ket{1}=\frac{1}{\sqrt{2}}(\ket{+}-\ket{-})=\frac{1}{\sqrt{2}} \ \begin{tikzpicture}
	\begin{pgfonlayer}{nodelayer}
		\node [style=none] (0) at (0.5, 0) {};
		\node [style=X phase dot] (2) at (-0.5, 0) {$\pi$};
	\end{pgfonlayer}
	\begin{pgfonlayer}{edgelayer}
		\draw (2) to (0.center);
	\end{pgfonlayer}
\end{tikzpicture}
.
    \end{equation}
    Therefore, plugging in \refeq{x_state_def} and \refeq{x_pi_state_def} into \refeq{z_spider_eq_redef}, we get:
    \begin{equation*}
        \input{2_Graphical_Languages/z_spider_def.tikz} = \frac{1}{\sqrt{2}^{n+m}} \  +
        \frac{e^{i\alpha}}{\sqrt{2}^{n+m}} \  
    \end{equation*}
\end{proof}
\begin{remark}
    \refthm{elementary_decomp} can be rewritten using the \StateCopy-rule, resulting in
    \begin{equation}
        \input{2_Graphical_Languages/z_spider_def.tikz} = \frac{1}{\sqrt{2}} \ \begin{tikzpicture}
	\begin{pgfonlayer}{nodelayer}
		\node [style=none] (1) at (1, -1) {};
		\node [style=none] (2) at (1, 1) {};
		\node [style=none] (3) at (-1, -1) {};
		\node [style=none] (4) at (-1, 1) {};
		\node [style=none, rotate=90] (5) at (-1, 0) {$...$};
		\node [style=none, rotate=90] (6) at (0.75, 0) {$...$};
		\node [style=none] (7) at (-1.5, -1) {};
		\node [style=none] (8) at (-1.5, 1) {};
		\node [style=none] (9) at (-2, 0) {$m$};
		\node [style=none] (10) at (1.5, 1) {};
		\node [style=none] (11) at (1.5, -1) {};
		\node [style=none] (12) at (2, 0) {$n$};
		\node [style=Z dot] (13) at (0, 0) {};
		\node [style=X dot] (14) at (0, 1.5) {};
	\end{pgfonlayer}
	\begin{pgfonlayer}{edgelayer}
		\draw [style=brace edge, rotate=180] (7.center) to (8.center);
		\draw [style=brace edge] (10.center) to (11.center);
		\draw (3.center) to (2.center);
		\draw (4.center) to (1.center);
		\draw (14) to (13);
	\end{pgfonlayer}
\end{tikzpicture}
 +
        \frac{e^{i\alpha}}{\sqrt{2}} \ \begin{tikzpicture}
	\begin{pgfonlayer}{nodelayer}
		\node [style=none] (0) at (1, -1) {};
		\node [style=none] (1) at (1, 1) {};
		\node [style=none] (2) at (-1, -1) {};
		\node [style=none] (3) at (-1, 1) {};
		\node [style=none, rotate=90] (4) at (-1, 0) {$...$};
		\node [style=none, rotate=90] (5) at (0.75, 0) {$...$};
		\node [style=none] (6) at (-1.5, -1) {};
		\node [style=none] (7) at (-1.5, 1) {};
		\node [style=none] (8) at (-2, 0) {$m$};
		\node [style=none] (9) at (1.5, 1) {};
		\node [style=none] (10) at (1.5, -1) {};
		\node [style=none] (11) at (2, 0) {$n$};
		\node [style=Z dot] (12) at (0, 0) {};
		\node [style=X phase dot] (13) at (0, 1.5) {$\pi$};
	\end{pgfonlayer}
	\begin{pgfonlayer}{edgelayer}
		\draw [style=brace edge, rotate=180] (6.center) to (7.center);
		\draw [style=brace edge] (9.center) to (10.center);
		\draw (2.center) to (1.center);
		\draw (3.center) to (0.center);
		\draw (13) to (12);
	\end{pgfonlayer}
\end{tikzpicture}
.
    \end{equation}
    This form is particularly handy for programmatic implementations, as one only has to attach a single state to the center node and multiply
    by a scalar\sidenote[][*-1]{And then also do a simplification that would be necessary in either case.}.
\end{remark}

\newpage

In \cite{kochSpeedyContractionZX2023a}, \refthm{elementary_decomp} was used to derive the following decomposition:
\begin{theorem}\labthm{speedy_elementary_decomp}
	\begin{equation} \textstyle{
	\begin{tikzpicture}
	\begin{pgfonlayer}{nodelayer}
		\node [style=Z phase dot] (0) at (1.5, 1) {$\gamma_1$};
		\node [style=none] (1) at (2.25, 1.5) {};
		\node [style=none] (2) at (2.25, 0.5) {};
		\node [style=none, rotate=90] (3) at (2.125, 1) {$...$};
		\node [style=Z phase dot] (4) at (1.5, -1) {$\gamma_m$};
		\node [style=none] (5) at (2.25, -0.5) {};
		\node [style=none] (6) at (2.25, -1.5) {};
		\node [style=none, rotate=90] (7) at (2.125, -1) {$...$};
		\node [style=none, rotate=90] (8) at (1.5, 0) {$...$};
		\node [style=Z phase dot] (9) at (-1.5, 1) {$\beta_1$};
		\node [style=none] (10) at (-2.25, 1.5) {};
		\node [style=none] (11) at (-2.25, 0.5) {};
		\node [style=none, rotate=90] (12) at (-2.125, 1) {$...$};
		\node [style=Z phase dot] (13) at (-1.5, -1) {$\beta_n$};
		\node [style=none] (14) at (-2.25, -0.5) {};
		\node [style=none] (15) at (-2.25, -1.5) {};
		\node [style=none, rotate=90] (16) at (-2.125, -1) {$...$};
		\node [style=none, rotate=90] (17) at (-1.5, 0) {$...$};
		\node [style=Z phase dot] (18) at (0, 0) {$\alpha$};
	\end{pgfonlayer}
	\begin{pgfonlayer}{edgelayer}
		\draw (2.center) to (0);
		\draw (0) to (1.center);
		\draw (6.center) to (4);
		\draw (4) to (5.center);
		\draw (11.center) to (9);
		\draw (9) to (10.center);
		\draw (15.center) to (13);
		\draw (13) to (14.center);
		\draw [style=hadamard edge] (9) to (18);
		\draw [style=hadamard edge] (18) to (13);
		\draw [style=star edge] (18) to (0);
		\draw [style=star edge] (18) to (4);
	\end{pgfonlayer}
\end{tikzpicture}
 
	~=~~ \frac{1}{\sqrt 2^{n}}~ \begin{tikzpicture}
	\begin{pgfonlayer}{nodelayer}
		\node [style=Z phase dot] (0) at (0.75, 1) {$\gamma_1$};
		\node [style=none] (1) at (1.5, 1.5) {};
		\node [style=none] (2) at (1.5, 0.5) {};
		\node [style=none, rotate=90] (3) at (1.375, 1) {$...$};
		\node [style=Z phase dot] (4) at (0.75, -1) {$\gamma_m$};
		\node [style=none] (5) at (1.5, -0.5) {};
		\node [style=none] (6) at (1.5, -1.5) {};
		\node [style=none, rotate=90] (7) at (1.375, -1) {$...$};
		\node [style=none, rotate=90] (8) at (0.75, 0) {$...$};
		\node [style=Z phase dot] (9) at (-0.75, 1) {$\beta_1$};
		\node [style=none] (10) at (-1.5, 1.5) {};
		\node [style=none] (11) at (-1.5, 0.5) {};
		\node [style=none, rotate=90] (12) at (-1.375, 1) {$...$};
		\node [style=Z phase dot] (13) at (-0.75, -1) {$\beta_n$};
		\node [style=none] (14) at (-1.5, -0.5) {};
		\node [style=none] (15) at (-1.5, -1.5) {};
		\node [style=none, rotate=90] (16) at (-1.375, -1) {$...$};
		\node [style=none, rotate=90] (17) at (-0.75, 0) {$...$};
	\end{pgfonlayer}
	\begin{pgfonlayer}{edgelayer}
		\draw (2.center) to (0);
		\draw (0) to (1.center);
		\draw (6.center) to (4);
		\draw (4) to (5.center);
		\draw (11.center) to (9);
		\draw (9) to (10.center);
		\draw (15.center) to (13);
		\draw (13) to (14.center);
	\end{pgfonlayer}
\end{tikzpicture}
 ~~+~~
    \frac{e^{i\alpha}}{\sqrt 2^{n+m}}~ \begin{tikzpicture}
	\begin{pgfonlayer}{nodelayer}
		\node [style=Z phase dot] (0) at (2, 1) {$\gamma_1$};
		\node [style=none] (1) at (2.75, 1.5) {};
		\node [style=none] (2) at (2.75, 0.5) {};
		\node [style=none, rotate=90] (3) at (2.625, 1) {$...$};
		\node [style=Z phase dot] (4) at (2, -1) {$\gamma_m$};
		\node [style=none] (5) at (2.75, -0.5) {};
		\node [style=none] (6) at (2.75, -1.5) {};
		\node [style=none, rotate=90] (7) at (2.625, -1) {$...$};
		\node [style=none, rotate=90] (8) at (2, 0) {$...$};
		\node [style=Z phase dot] (9) at (-1, 1) {$\beta_1+\pi$};
		\node [style=none] (10) at (-2.25, 1.5) {};
		\node [style=none] (11) at (-2.25, 0.5) {};
		\node [style=none, rotate=90] (12) at (-2.125, 1) {$...$};
		\node [style=Z phase dot] (13) at (-1, -1) {$\beta_n+\pi$};
		\node [style=none] (14) at (-2.25, -0.5) {};
		\node [style=none] (15) at (-2.25, -1.5) {};
		\node [style=none, rotate=90] (16) at (-2.125, -1) {$...$};
		\node [style=none, rotate=90] (17) at (-1, 0) {$...$};
		\node [style=Z dot] (20) at (0.75, 1) {};
		\node [style=Z dot] (21) at (0.75, -1) {};
	\end{pgfonlayer}
	\begin{pgfonlayer}{edgelayer}
		\draw (2.center) to (0);
		\draw (0) to (1.center);
		\draw (6.center) to (4);
		\draw (4) to (5.center);
		\draw (11.center) to (9);
		\draw (9) to (10.center);
		\draw (15.center) to (13);
		\draw (13) to (14.center);
		\draw [style=hadamard edge] (21) to (4);
		\draw [style=hadamard edge] (0) to (20);
	\end{pgfonlayer}
\end{tikzpicture}
} \labeq{tof_decomp}
	\end{equation}
\end{theorem}

This yields a scaling of $\beta=\frac{1}{m}$ for $m$ star edges. Since this decomposition holds for any $m$, we may call this
a \textit{dynamic decomposition}\sidenote{The term does not occur in literature, but we will be referring to this decomposition a lot, so
we have thus introduced it.}.

We can use these theorems to find the stabilizer rank of a multi-control Toffoli gate with $n$ control bits. The first step is to translate the according quantum circuit
into a ZX-diagram \cite{kochSpeedyContractionZX2023a}:
\begin{align*}
    \input{4_Dynamic_Decompositions/multi-control-toffoli-gate_qcirc.tikz} \leadsto& \input{4_Dynamic_Decompositions/multi-control-toffoli-gate_zx_triang.tikz} \\
    =& \input{4_Dynamic_Decompositions/multi-control-toffoli-gate_zx_star1.tikz}
    = \input{4_Dynamic_Decompositions/multi-control-toffoli-gate_zx_star2.tikz} \labeq{multi_ctrl_toff_zx}
\end{align*}
For the sake of simplicity, we will use \refthm{elementary_decomp} instead of \refthm{speedy_elementary_decomp}:
\begin{gather}
    \input{4_Dynamic_Decompositions/multi-control-toffoli-gate_zx_star2.tikz} \notag \\
    = \frac{1}{\sqrt{2}^{n+1}} \input{4_Dynamic_Decompositions/multi-control-toffoli-gate_zx_star2_decomp1.tikz}
    + \frac{1}{\sqrt{2}^{n+1}} \input{4_Dynamic_Decompositions/multi-control-toffoli-gate_zx_star2_decomp2.tikz} \labeq{result_tof_rank}
\end{gather}
At this point, it is worth introducing a few simple to derive equalities:
\begin{align}
    \input{4_Dynamic_Decompositions/star_state_x_pi_equality.tikz} \labeq{star_state_x_pi}\\
    \input{4_Dynamic_Decompositions/star_state_x_equality.tikz} \labeq{star_state_x} \\
    \input{4_Dynamic_Decompositions/star_state_z_pi_equality.tikz} \labeq{star_state_z_pi}
\end{align}
One can therefore observe that both terms in \refeq{result_tof_rank} are stabilizer terms, since all non-trivial leftover states can be brought into the form
of \refeq{star_state_x_pi}, \refeq{star_state_x} or \refeq{star_state_z_pi}. We can conclude that the stabilizer rank of a multi-control Toffoli gates is two.

In \refch{weighting_algorithms}, we will use the elementary decomposition from \refthm{elementary_decomp} to do more sophisticated procedures.

\section{Comparison with State Decompositions}
\labsec{comparison_with_state_decomps}

One of the authors of \cite{kochSpeedyContractionZX2023a}, Mark Koch, was kind enough to share the source code for their algorithm, including
the benchmark for barren plateau detection.

We repeated the benchmark after having slightly modified the code, which allowed us to capture information regarding which decomposition was applied for
a given ansatz. The results\sidenote{The log file can be found in \cite{Novel_star_state_decompositionsZxbarrenprivateResulttxt}.} show that around
two-thirds of the decompositions were star decompositions, and only one-third utilized dynamic decompositions. State decompositions did not occur.

It is evident that decompositions of matrices occur much more often than state decompositions, as they can be used more generally. Since not even the state
decompositions introduced in the original paper occurred, we can conclude that our novel decompositions are of little use for this concrete application.
Nevertheless, it is probable to assume that other applications may involve such star states to a greater extent.

In the next chapter, we will focus on quantum circuits, where, according to our testing, even star decompositions are exceedingly rare or nonexistent, despite
the presence of numerous star edges. Consequently, dynamic decompositions will occur more frequently, which will be beneficial for the concept we will
introduce, namely \textit{weighting algorithms}.
\setchapterpreamble[u]{\margintoc}
\chapter{Weighting Algorithms}
\labch{weighting_algorithms}

Whilst exploring further areas of the ZX-calculus, we encountered an interesting paper \cite{sutcliffeProcedurallyOptimisedZXDiagram2024}. They propose
and demonstrate a new approach for strong simulation. Traditionally, programs would check which decompositions\sidenote[][*5.5]{For example, the BSS decomposition.}
currently apply, and then greedily choose the one with the smallest scaling factor. The novel algorithm, in short, uses a weighting algorithm that
considers more information about the graph than previously done. They implemented a Python Proof-of-Concept (PoC) \cite{mjsutcliffe99Mjsutcliffe99ProcOptCut2024}.
The results can be seen in \reffig{procedurally_results1} and \reffig{procedurally_results2}.
\begin{marginfigure}[*2.5]
    \includegraphics{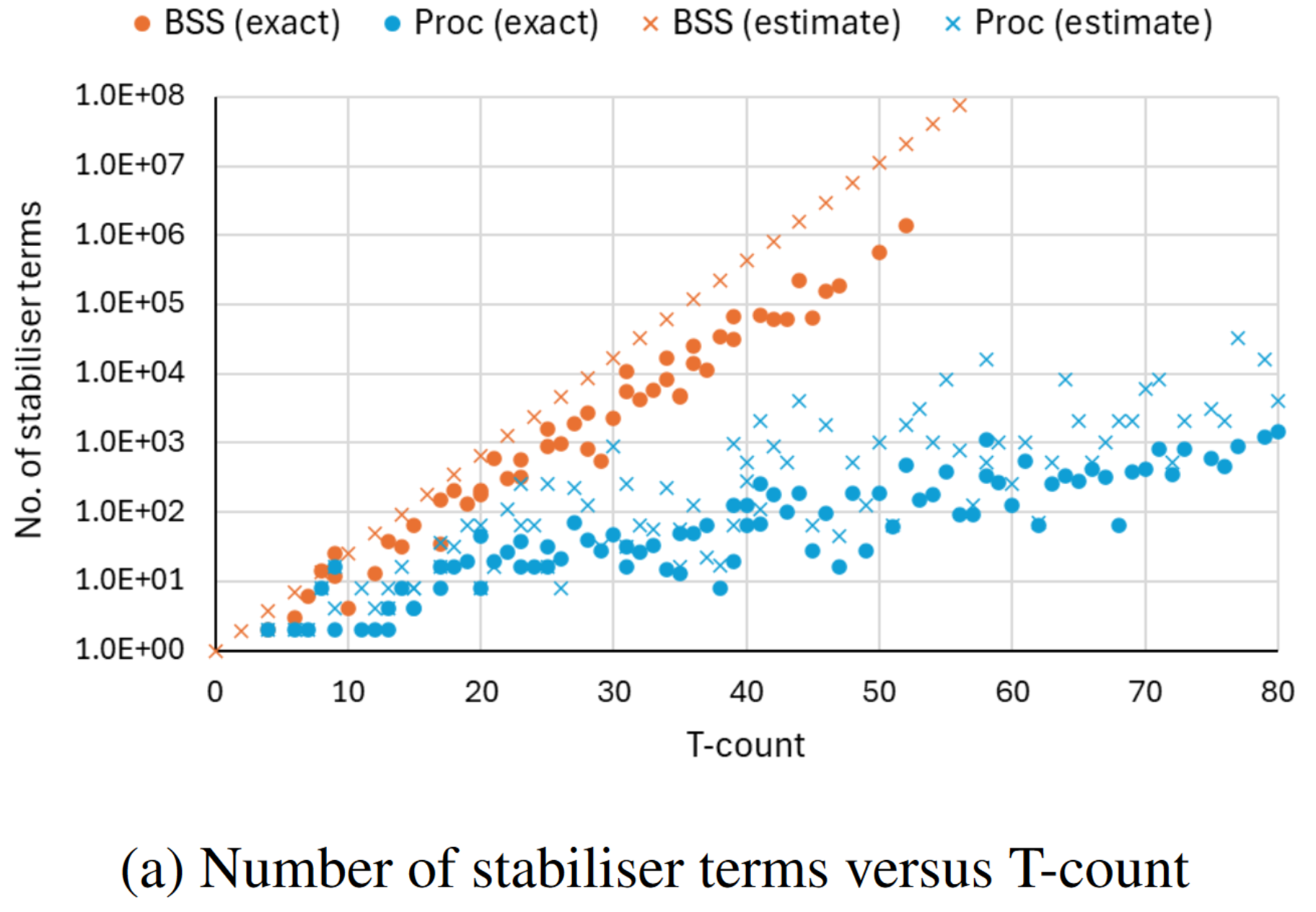}
    \caption[Procedurally Optimised paper results 1]{Results from PoC \cite{sutcliffeProcedurallyOptimisedZXDiagram2024}.}
    \labfig{procedurally_results1}
\end{marginfigure}
\begin{marginfigure}[*11.5]
    \includegraphics{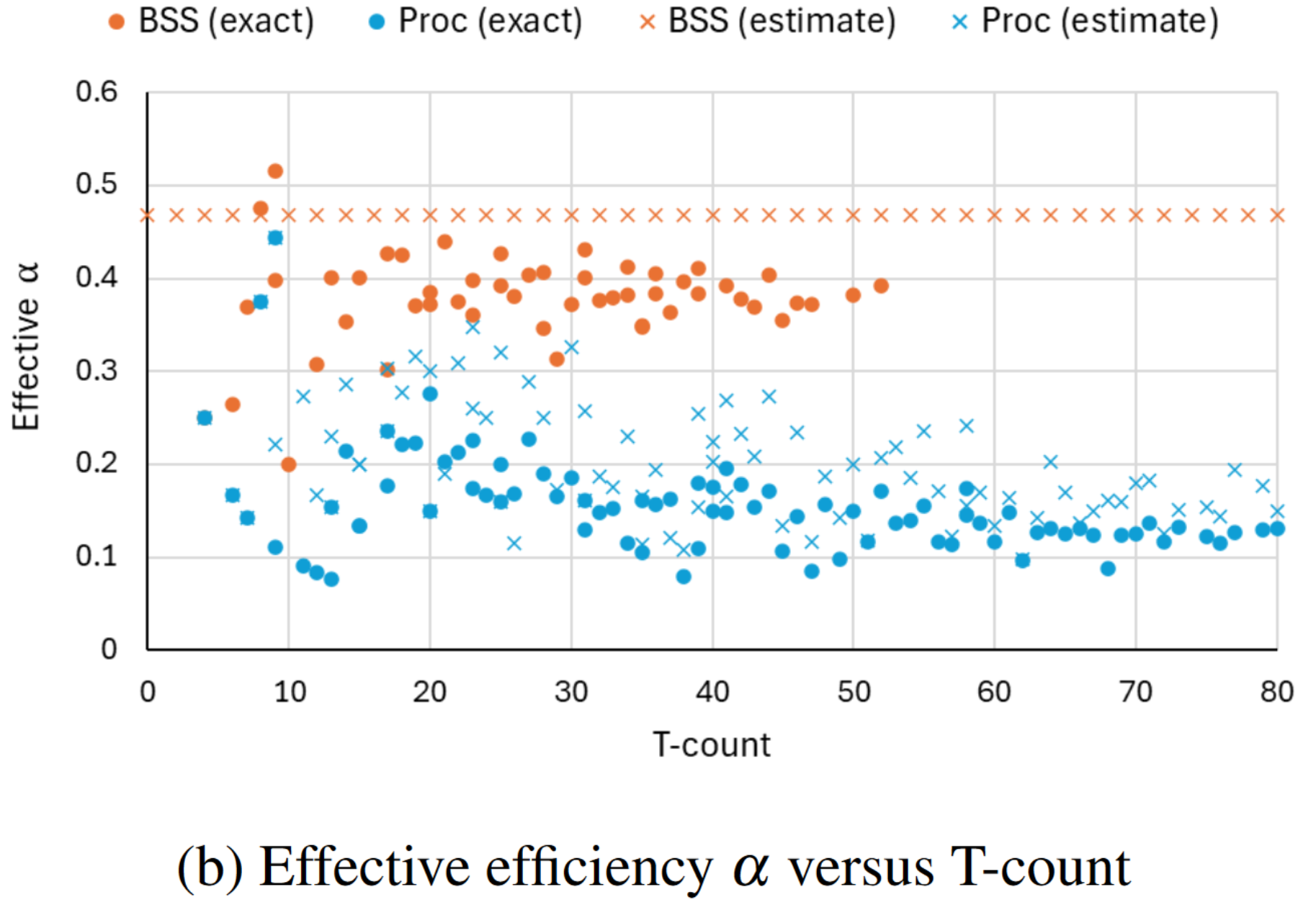}
    \caption[Procedurally Optimised paper results 2]{Results from PoC \cite{sutcliffeProcedurallyOptimisedZXDiagram2024}.}
    \labfig{procedurally_results2}
\end{marginfigure}

In the following, we will briefly explain their method, and then we will show whether these ideas can be translated into a similar problem, namely
the strong simulation of quantum circuits consisting of many multi-control Toffoli gates, rather than many T-gates.

\section{Basic Idea}
\labsec{basic_idea}

Consider the following ZX-diagram taken from \cite{sutcliffeProcedurallyOptimisedZXDiagram2024}:
\begin{equation*}
    \input{5_Weighting_Algorithms/picom_grouped_start.tikz}
\end{equation*}
The usual approach would be to unfuse four $\frac{\pi}{4}$-states and then use decompositions similar to the ones reviewed in \refsec{prev_work_stab_decomp}.
However, one might instead use \refthm{elementary_decomp} in order to \textit{cut} a Z-spider:
\begin{equation*}
    \input{5_Weighting_Algorithms/picom_grouped.tikz}
\end{equation*}
This should be very compelling, since we have just seen that using a decomposition with a worse scaling than traditional decompositions can yield
a better \textit{effective scaling}, i.e. the scaling calculated from the final number of terms, in this case\sidenote{By extending
the diagram to include more separated T-gates whilst maintaining a common Z-spider, one cut would still suffice, thus giving $\alpha \rightarrow 0$.} $\alpha=0.25$.

The above example is part of an idea called \textit{CNOT-grouping}. More generally, in \cite{sutcliffeProcedurallyOptimisedZXDiagram2024} it was realized
that CNOTs can \textit{block} $\pm\frac{\pi}{4}$-spiders from fusing, which would result in a phase of either $0$ or $\frac{\pi}{2}$, hence being reduced
to a Clifford (sub-)diagram. By defining strict rules, one can classify each spider by considering its neighbors and counting how many T-gates could
potentially be fused, up until some depth\sidenote[][*-2]{The actual procedure employed is slightly more sophisticated, utilizing so-called \textit{tiers}.}.

The results in \reffig{procedurally_results1} and \reffig{procedurally_results2} clearly show that the new approach can yield a benefit over
traditional methods. Nevertheless, a few important things need to be mentioned:
\begin{enumerate}
    \item The experiments were conducted using "random" circuits. Here, the challenge is to achieve a natural balance between highly random circuits, which may prevent the algorithm from demonstrating its effectiveness, and highly structured circuits, where the results may be overly biased.
    \item By only implementing a PoC, runtime comparisons are not possible with state-of-the-art implementations like \cite{ZxcalcQuizx2024}. This only leaves the final number of terms as a usable metric for comparison. Nevertheless, they did provide a rough upper-bound for the computational complexity of $\mathcal{O}(V^2)$ for $V$ being the
    number of non-trivial vertices. It is thus argued to be negligible compared to the exponential scaling produced by stabilizer terms.
    \item The algorithm only uses partial simplifications\sidenote{This means that they, for example, do not use \textit{pivoting} or
    \textit{local complementation}.}, as this avoids the risk of losing relevant patterns. Later, however, it is argued that this might have been the reason for not even better results.
\end{enumerate}
Similarly, the following points should be noted for our own algorithm, as described in \refsec{studied_approach} and \refsec{application_multiloop_feynman_diagrams}:
\begin{enumerate}
    \item The multiloop Feynman diagram application might be a good or a bad fit, which is why we will conduct experiments using "more random" circuits too.
    \item We will also only implement a PoC in Python. We will not consider the space-time complexity in more detail, however, some further comments on the runtime will be made.
    \item Our algorithm will also only use a partial simplification strategy to mimic the original algorithm and potentially avoid the same problem
    of losing relevant patterns. This choice will be further discussed in \refsec{importance_of_simp_strat}.
\end{enumerate}

We will now demonstrate an attempt to replicate the idea of decomposing spiders that block further improvements in diagrams consisting of star edges
instead of T-gates. To start with, consider the
following diagram:
\begin{equation}\labeq{CNOT_example_1}
    \begin{tikzpicture}
	\begin{pgfonlayer}{nodelayer}
		\node [style=Z dot] (0) at (-2, 2) {};
		\node [style=Z dot] (1) at (0, 2) {};
		\node [style=Z dot] (2) at (2, 2) {};
		\node [style=X dot] (3) at (-2, 1) {};
		\node [style=X dot] (4) at (0, 1) {};
		\node [style=X dot] (5) at (2, 1) {};
		\node [style=Z dot] (6) at (-3, 1) {};
		\node [style=Z dot] (7) at (-3, 0) {};
		\node [style=Z dot] (8) at (-1, -1) {};
		\node [style=Z dot] (9) at (-1, 1) {};
		\node [style=Z dot] (10) at (1, 1) {};
		\node [style=Z dot] (11) at (1, -2) {};
		\node [style=Z dot] (12) at (3, 1) {};
		\node [style=Z dot] (13) at (3, 0) {};
		\node [style=X phase dot] (14) at (2, 0) {$\frac{\pi}{2}$};
		\node [style=none] (15) at (-4, 2) {};
		\node [style=none] (16) at (-4, 1) {};
		\node [style=none] (17) at (-4, 0) {};
		\node [style=none] (18) at (-4, -1) {};
		\node [style=none] (19) at (-4, -2) {};
		\node [style=none] (20) at (4, 2) {};
		\node [style=none] (21) at (4, 1) {};
		\node [style=none] (22) at (4, 0) {};
		\node [style=none] (23) at (4, -1) {};
		\node [style=none] (24) at (4, -2) {};
	\end{pgfonlayer}
	\begin{pgfonlayer}{edgelayer}
		\draw (15.center) to (20.center);
		\draw (21.center) to (16.center);
		\draw (17.center) to (22.center);
		\draw (23.center) to (18.center);
		\draw (19.center) to (24.center);
		\draw (1) to (4);
		\draw (0) to (3);
		\draw (2) to (5);
		\draw [style=star edge] (6) to (7);
		\draw [style=star edge] (9) to (8);
		\draw [style=star edge] (10) to (11);
		\draw [style=star edge] (12) to (13);
	\end{pgfonlayer}
\end{tikzpicture}

\end{equation}
By fusing the control bits, we get
\begin{equation}\labeq{CNOT_example_2}
    \begin{tikzpicture}
	\begin{pgfonlayer}{nodelayer}
		\node [style=Z dot] (1) at (0, 2) {};
		\node [style=X dot] (3) at (-2, 1) {};
		\node [style=X dot] (4) at (0, 1) {};
		\node [style=X dot] (5) at (2, 1) {};
		\node [style=Z dot] (6) at (-3, 1) {};
		\node [style=Z dot] (7) at (-3, 0) {};
		\node [style=Z dot] (8) at (-1, -1) {};
		\node [style=Z dot] (9) at (-1, 1) {};
		\node [style=Z dot] (10) at (1, 1) {};
		\node [style=Z dot] (11) at (1, -2) {};
		\node [style=Z dot] (12) at (3, 1) {};
		\node [style=Z dot] (13) at (3, 0) {};
		\node [style=X phase dot] (14) at (2, 0) {$\frac{\pi}{2}$};
		\node [style=none] (15) at (-4, 2) {};
		\node [style=none] (16) at (-4, 1) {};
		\node [style=none] (17) at (-4, 0) {};
		\node [style=none] (18) at (-4, -1) {};
		\node [style=none] (19) at (-4, -2) {};
		\node [style=none] (20) at (4, 2) {};
		\node [style=none] (21) at (4, 1) {};
		\node [style=none] (22) at (4, 0) {};
		\node [style=none] (23) at (4, -1) {};
		\node [style=none] (24) at (4, -2) {};
	\end{pgfonlayer}
	\begin{pgfonlayer}{edgelayer}
		\draw (15.center) to (20.center);
		\draw (21.center) to (16.center);
		\draw (17.center) to (22.center);
		\draw (23.center) to (18.center);
		\draw (19.center) to (24.center);
		\draw (1) to (4);
		\draw [style=star edge] (6) to (7);
		\draw [style=star edge] (9) to (8);
		\draw [style=star edge] (10) to (11);
		\draw [style=star edge] (12) to (13);
		\draw (3) to (1);
		\draw (1) to (5);
	\end{pgfonlayer}
\end{tikzpicture}
.
\end{equation}
It is now possible to decompose the resulting spider according to \refthm{elementary_decomp}, which gives us
\begin{equation*}
    \begin{tikzpicture}
	\begin{pgfonlayer}{nodelayer}
		\node [style=Z dot] (4) at (-3, 1) {};
		\node [style=Z dot] (5) at (-3, 0) {};
		\node [style=Z dot] (6) at (-1, -1) {};
		\node [style=Z dot] (7) at (-1, 1) {};
		\node [style=Z dot] (8) at (1, 1) {};
		\node [style=Z dot] (9) at (1, -2) {};
		\node [style=Z dot] (10) at (3, 1) {};
		\node [style=Z dot] (11) at (3, 0) {};
		\node [style=X phase dot] (12) at (2, 0) {$\frac{\pi}{2}$};
		\node [style=none] (13) at (-4, 2) {};
		\node [style=none] (14) at (-4, 1) {};
		\node [style=none] (15) at (-4, 0) {};
		\node [style=none] (16) at (-4, -1) {};
		\node [style=none] (17) at (-4, -2) {};
		\node [style=none] (18) at (4, 2) {};
		\node [style=none] (19) at (4, 1) {};
		\node [style=none] (20) at (4, 0) {};
		\node [style=none] (21) at (4, -1) {};
		\node [style=none] (22) at (4, -2) {};
		\node [style=X dot] (23) at (-1, 2) {};
		\node [style=X dot] (24) at (1, 2) {};
	\end{pgfonlayer}
	\begin{pgfonlayer}{edgelayer}
		\draw (19.center) to (14.center);
		\draw (15.center) to (20.center);
		\draw (21.center) to (16.center);
		\draw (17.center) to (22.center);
		\draw [style=star edge] (4) to (5);
		\draw [style=star edge] (7) to (6);
		\draw [style=star edge] (8) to (9);
		\draw [style=star edge] (10) to (11);
		\draw (13.center) to (23);
		\draw (24) to (18.center);
	\end{pgfonlayer}
\end{tikzpicture}
+\begin{tikzpicture}
	\begin{pgfonlayer}{nodelayer}
		\node [style=Z dot] (0) at (-3, 1) {};
		\node [style=Z dot] (1) at (-3, 0) {};
		\node [style=Z dot] (2) at (-1, -1) {};
		\node [style=Z dot] (3) at (-1, 1) {};
		\node [style=Z dot] (4) at (1, 1) {};
		\node [style=Z dot] (5) at (1, -2) {};
		\node [style=Z dot] (6) at (3, 1) {};
		\node [style=Z dot] (7) at (3, 0) {};
		\node [style=X phase dot] (8) at (2, 0) {$\frac{\pi}{2}$};
		\node [style=none] (9) at (-4, 2) {};
		\node [style=none] (10) at (-4, 1) {};
		\node [style=none] (11) at (-4, 0) {};
		\node [style=none] (12) at (-4, -1) {};
		\node [style=none] (13) at (-4, -2) {};
		\node [style=none] (14) at (4, 2) {};
		\node [style=none] (15) at (4, 1) {};
		\node [style=none] (16) at (4, 0) {};
		\node [style=none] (17) at (4, -1) {};
		\node [style=none] (18) at (4, -2) {};
		\node [style=X phase dot] (19) at (-1, 2) {$\pi$};
		\node [style=X phase dot] (20) at (1, 2) {$\pi$};
		\node [style=X phase dot] (21) at (-2, 1) {$\pi$};
		\node [style=X phase dot] (22) at (0, 1) {$\pi$};
		\node [style=X phase dot] (23) at (2, 1) {$\pi$};
	\end{pgfonlayer}
	\begin{pgfonlayer}{edgelayer}
		\draw (15.center) to (10.center);
		\draw (11.center) to (16.center);
		\draw (17.center) to (12.center);
		\draw (13.center) to (18.center);
		\draw [style=star edge] (0) to (1);
		\draw [style=star edge] (3) to (2);
		\draw [style=star edge] (4) to (5);
		\draw [style=star edge] (6) to (7);
		\draw (19) to (9.center);
		\draw (14.center) to (20);
	\end{pgfonlayer}
\end{tikzpicture}
.
\end{equation*}
The left term can be simplified by fusing the four Z-spiders. According to \refthm{speedy_elementary_decomp}, we can decompose the resulting node, which gives
us two terms.

The right term is a bit more difficult due to the new obstructions caused by NOT gates. However, it can be shown that the right term can also be decomposed
into two terms, according to \reflemma{NOT_obstruction}. Therefore, \refeq{CNOT_example_1} can be decomposed into four terms. But, what would be the final
number of terms if a traditional method was used? At first, it might seem that \refeq{edge_decomp_2} is the best choice. Applying it twice,
one obtains nine terms. While this may be true in theory, algorithms that greedily choose the decomposition with the smallest scaling factor will
always simplify the diagram after having applied one decomposition. This is exemplified after the following lemma and proof.

\begin{lemma}
    \lablemma{NOT_obstruction}
    Any (sub-)diagram of the form
    \begin{equation}
        \input{5_Weighting_Algorithms/NOT_lemma/1.tikz}
    \end{equation}
    with $m$ star edges can be rewritten as
    \begin{equation}
        \input{5_Weighting_Algorithms/NOT_lemma/5.tikz}
    \end{equation}
\end{lemma}

\newpage

\begin{proof}
    The diagram
    \begin{equation*}
        \input{5_Weighting_Algorithms/NOT_lemma/1.tikz}
    \end{equation*}
    with $m$ star edges, can be rewritten by using \refeq{star_triang_relation}, giving us
    \begin{align}
        &\input{5_Weighting_Algorithms/NOT_lemma/2.tikz} \labeq{stacks_repr_use_case} \\[1.0ex]
        = \ &\input{5_Weighting_Algorithms/NOT_lemma/3.tikz} \labeq{stacks_repr}
    \end{align}
    The star edges will always be connected to one of the two Z-spiders. We will sometimes refer to these two sides as \textit{stacks}. By decomposing
    both stacks using \refthm{speedy_elementary_decomp}, one obtains four terms, two of which cancel out as they will include an isolated X-spiders with
    a phase of $\pi$, which is equivalent to a zero scalar. The final two terms (for both $m$ being odd and even) are given by
    \begin{equation}
        \input{5_Weighting_Algorithms/NOT_lemma/4_full_odd.tikz}
    \end{equation}
    and
    \begin{equation}
        \input{5_Weighting_Algorithms/NOT_lemma/4_full_even.tikz}
    \end{equation}
    which can be brought into the desired form of \reflemma{NOT_obstruction}.
\end{proof}
Getting back to the reason why such CNOT blockades cannot be exploited like in the T-gate case, let us decompose the first star edge in \refeq{CNOT_example_2}
using \refeq{edge_decomp_1}:
\begin{equation*}
    \begin{tikzpicture}
	\begin{pgfonlayer}{nodelayer}
		\node [style=Z dot] (0) at (0, 2) {};
		\node [style=X dot] (1) at (-2, 1) {};
		\node [style=X dot] (2) at (0, 1) {};
		\node [style=X dot] (3) at (2, 1) {};
		\node [style=Z dot] (4) at (-3, 1) {};
		\node [style=Z dot] (5) at (-3, -1.5) {};
		\node [style=Z dot] (6) at (-1, -2.5) {};
		\node [style=Z dot] (7) at (-1, 1) {};
		\node [style=Z dot] (8) at (1, 1) {};
		\node [style=Z dot] (9) at (1, -3.5) {};
		\node [style=Z dot] (10) at (3, 1) {};
		\node [style=Z dot] (11) at (3, -1.5) {};
		\node [style=X phase dot] (12) at (2, -1.5) {$\frac{\pi}{2}$};
		\node [style=none] (13) at (-4, 2) {};
		\node [style=none] (14) at (-4, 1) {};
		\node [style=none] (15) at (-4, -1.5) {};
		\node [style=none] (16) at (-4, -2.5) {};
		\node [style=none] (17) at (-4, -3.5) {};
		\node [style=none] (18) at (4, 2) {};
		\node [style=none] (19) at (4, 1) {};
		\node [style=none] (20) at (4, -1.5) {};
		\node [style=none] (21) at (4, -2.5) {};
		\node [style=none] (22) at (4, -3.5) {};
		\node [style=Z dot] (23) at (-3, -0.75) {};
		\node [style=X dot] (24) at (-3, 0.25) {};
	\end{pgfonlayer}
	\begin{pgfonlayer}{edgelayer}
		\draw (13.center) to (18.center);
		\draw (19.center) to (14.center);
		\draw (15.center) to (20.center);
		\draw (21.center) to (16.center);
		\draw (17.center) to (22.center);
		\draw (0) to (2);
		\draw [style=star edge] (7) to (6);
		\draw [style=star edge] (8) to (9);
		\draw [style=star edge] (10) to (11);
		\draw (1) to (0);
		\draw (0) to (3);
		\draw (23) to (5);
		\draw (4) to (24);
	\end{pgfonlayer}
\end{tikzpicture}
 + \begin{tikzpicture}
	\begin{pgfonlayer}{nodelayer}
		\node [style=Z dot] (0) at (0, 2) {};
		\node [style=X dot] (1) at (-2, 1) {};
		\node [style=X dot] (2) at (0, 1) {};
		\node [style=X dot] (3) at (2, 1) {};
		\node [style=Z dot] (4) at (-3, 1) {};
		\node [style=Z dot] (5) at (-3, -1.5) {};
		\node [style=Z dot] (6) at (-1, -2.5) {};
		\node [style=Z dot] (7) at (-1, 1) {};
		\node [style=Z dot] (8) at (1, 1) {};
		\node [style=Z dot] (9) at (1, -3.5) {};
		\node [style=Z dot] (10) at (3, 1) {};
		\node [style=Z dot] (11) at (3, -1.5) {};
		\node [style=X phase dot] (12) at (2, -1.5) {$\frac{\pi}{2}$};
		\node [style=none] (13) at (-4, 2) {};
		\node [style=none] (14) at (-4, 1) {};
		\node [style=none] (15) at (-4, -1.5) {};
		\node [style=none] (16) at (-4, -2.5) {};
		\node [style=none] (17) at (-4, -3.5) {};
		\node [style=none] (18) at (4, 2) {};
		\node [style=none] (19) at (4, 1) {};
		\node [style=none] (20) at (4, -1.5) {};
		\node [style=none] (21) at (4, -2.5) {};
		\node [style=none] (22) at (4, -3.5) {};
		\node [style=X dot] (25) at (-3, -0.75) {};
		\node [style=X phase dot] (26) at (-3, 0.25) {$\pi$};
	\end{pgfonlayer}
	\begin{pgfonlayer}{edgelayer}
		\draw (13.center) to (18.center);
		\draw (19.center) to (14.center);
		\draw (15.center) to (20.center);
		\draw (21.center) to (16.center);
		\draw (17.center) to (22.center);
		\draw (0) to (2);
		\draw [style=star edge] (7) to (6);
		\draw [style=star edge] (8) to (9);
		\draw [style=star edge] (10) to (11);
		\draw (1) to (0);
		\draw (0) to (3);
		\draw (4) to (26);
		\draw (25) to (5);
	\end{pgfonlayer}
\end{tikzpicture}

\end{equation*}
Here, the left term can be further simplified to the following:
\begin{equation*}
    \begin{tikzpicture}
	\begin{pgfonlayer}{nodelayer}
		\node [style=Z dot] (0) at (-5, 9) {};
		\node [style=X dot] (1) at (-7, 8) {};
		\node [style=X dot] (2) at (-5, 8) {};
		\node [style=X dot] (3) at (-3, 8) {};
		\node [style=Z dot] (5) at (-8, 6.25) {};
		\node [style=Z dot] (6) at (-6, 5.25) {};
		\node [style=Z dot] (7) at (-6, 8) {};
		\node [style=Z dot] (8) at (-4, 8) {};
		\node [style=Z dot] (9) at (-4, 4.25) {};
		\node [style=Z dot] (10) at (-2, 8) {};
		\node [style=Z dot] (11) at (-2, 6.25) {};
		\node [style=X phase dot] (12) at (-3, 6.25) {$\frac{\pi}{2}$};
		\node [style=none] (13) at (-9, 9) {};
		\node [style=none] (14) at (-9, 8) {};
		\node [style=none] (15) at (-9, 6.25) {};
		\node [style=none] (16) at (-9, 5.25) {};
		\node [style=none] (17) at (-9, 4.25) {};
		\node [style=none] (18) at (-1, 9) {};
		\node [style=none] (19) at (-1, 8) {};
		\node [style=none] (20) at (-1, 6.25) {};
		\node [style=none] (21) at (-1, 5.25) {};
		\node [style=none] (22) at (-1, 4.25) {};
		\node [style=Z dot] (23) at (-8, 7) {};
		\node [style=X dot] (24) at (-8, 8) {};
		\node [style=Z dot] (25) at (5, 9) {};
		\node [style=X dot] (27) at (5, 8) {};
		\node [style=X dot] (28) at (7, 8) {};
		\node [style=Z dot] (29) at (2, 6.25) {};
		\node [style=Z dot] (30) at (4, 5.25) {};
		\node [style=Z dot] (31) at (4, 8) {};
		\node [style=Z dot] (32) at (6, 8) {};
		\node [style=Z dot] (33) at (6, 4.25) {};
		\node [style=Z dot] (34) at (8, 8) {};
		\node [style=Z dot] (35) at (8, 6.25) {};
		\node [style=X phase dot] (36) at (7, 6.25) {$\frac{\pi}{2}$};
		\node [style=none] (37) at (1, 9) {};
		\node [style=none] (38) at (1, 8) {};
		\node [style=none] (39) at (1, 6.25) {};
		\node [style=none] (40) at (1, 5.25) {};
		\node [style=none] (41) at (1, 4.25) {};
		\node [style=none] (42) at (9, 9) {};
		\node [style=none] (43) at (9, 8) {};
		\node [style=none] (44) at (9, 6.25) {};
		\node [style=none] (45) at (9, 5.25) {};
		\node [style=none] (46) at (9, 4.25) {};
		\node [style=Z dot] (47) at (2, 7) {};
		\node [style=X dot] (48) at (2, 8) {};
		\node [style=none] (49) at (0, 6.75) {$=$};
		\node [style=Z dot] (50) at (-5, 2.5) {};
		\node [style=X dot] (51) at (-4.75, 1.5) {};
		\node [style=X dot] (52) at (-3, 1.5) {};
		\node [style=Z dot] (53) at (-8, -0.25) {};
		\node [style=Z dot] (54) at (-6, -1.25) {};
		\node [style=Z dot] (56) at (-4, 1.5) {};
		\node [style=Z dot] (57) at (-4, -2.25) {};
		\node [style=Z dot] (58) at (-2, 1.5) {};
		\node [style=Z dot] (59) at (-2, -0.25) {};
		\node [style=X phase dot] (60) at (-3, -0.25) {$\frac{\pi}{2}$};
		\node [style=none] (61) at (-9, 2.5) {};
		\node [style=none] (62) at (-9, 1.5) {};
		\node [style=none] (63) at (-9, -0.25) {};
		\node [style=none] (64) at (-9, -1.25) {};
		\node [style=none] (65) at (-9, -2.25) {};
		\node [style=none] (66) at (-1, 2.5) {};
		\node [style=none] (67) at (-1, 1.5) {};
		\node [style=none] (68) at (-1, -0.25) {};
		\node [style=none] (69) at (-1, -1.25) {};
		\node [style=none] (70) at (-1, -2.25) {};
		\node [style=Z dot] (71) at (-8, 0.5) {};
		\node [style=X dot] (72) at (-8, 1.5) {};
		\node [style=none] (73) at (-10.75, 0.25) {$=$};
		\node [style=none] (74) at (-10.75, 1) {\Hopf, \SpiderFusion};
		\node [style=Z dot] (75) at (5, 2.5) {};
		\node [style=X dot] (77) at (7, 1.5) {};
		\node [style=Z dot] (78) at (2, -0.25) {};
		\node [style=Z dot] (79) at (4, -1.25) {};
		\node [style=Z dot] (81) at (6, -2.25) {};
		\node [style=Z dot] (82) at (8, 1.5) {};
		\node [style=Z dot] (83) at (8, -0.25) {};
		\node [style=X phase dot] (84) at (7, -0.25) {$\frac{\pi}{2}$};
		\node [style=none] (85) at (1, 2.5) {};
		\node [style=none] (86) at (1, 1.5) {};
		\node [style=none] (87) at (1, -0.25) {};
		\node [style=none] (88) at (1, -1.25) {};
		\node [style=none] (89) at (1, -2.25) {};
		\node [style=none] (90) at (9, 2.5) {};
		\node [style=none] (91) at (9, 1.5) {};
		\node [style=none] (92) at (9, -0.25) {};
		\node [style=none] (93) at (9, -1.25) {};
		\node [style=none] (94) at (9, -2.25) {};
		\node [style=Z dot] (95) at (2, 0.5) {};
		\node [style=X dot] (96) at (2, 1.5) {};
		\node [style=none] (97) at (0, 0.25) {$=$};
		\node [style=X dot] (98) at (6, 0.5) {};
		\node [style=Z dot] (99) at (-5, -4) {};
		\node [style=Z dot] (101) at (-8, -6.75) {};
		\node [style=Z dot] (102) at (-6, -7.75) {};
		\node [style=Z dot] (103) at (-4, -8.75) {};
		\node [style=Z dot] (105) at (-2, -6.75) {};
		\node [style=X phase dot] (106) at (-3, -6.75) {$\frac{\pi}{2}$};
		\node [style=none] (107) at (-9, -4) {};
		\node [style=none] (108) at (-9, -5) {};
		\node [style=none] (109) at (-9, -6.75) {};
		\node [style=none] (110) at (-9, -7.75) {};
		\node [style=none] (111) at (-9, -8.75) {};
		\node [style=none] (112) at (-1, -4) {};
		\node [style=none] (113) at (-1, -5) {};
		\node [style=none] (114) at (-1, -6.75) {};
		\node [style=none] (115) at (-1, -7.75) {};
		\node [style=none] (116) at (-1, -8.75) {};
		\node [style=Z dot] (117) at (-8, -6) {};
		\node [style=X dot] (118) at (-8, -5) {};
		\node [style=none] (119) at (-10.75, -6.25) {$=$};
		\node [style=none] (121) at (-9.75, -6.25) {$\sqrt{2}$};
		\node [style=Z dot] (122) at (-4, -6) {};
	\end{pgfonlayer}
	\begin{pgfonlayer}{edgelayer}
		\draw (13.center) to (18.center);
		\draw (15.center) to (20.center);
		\draw (21.center) to (16.center);
		\draw (17.center) to (22.center);
		\draw (0) to (2);
		\draw [style=star edge] (7) to (6);
		\draw [style=star edge] (8) to (9);
		\draw [style=star edge] (10) to (11);
		\draw (1) to (0);
		\draw (0) to (3);
		\draw (23) to (5);
		\draw (1) to (19.center);
		\draw (24) to (14.center);
		\draw (37.center) to (42.center);
		\draw (39.center) to (44.center);
		\draw (45.center) to (40.center);
		\draw (41.center) to (46.center);
		\draw (25) to (27);
		\draw [style=star edge] (31) to (30);
		\draw [style=star edge] (32) to (33);
		\draw [style=star edge] (34) to (35);
		\draw (25) to (28);
		\draw (47) to (29);
		\draw (48) to (38.center);
		\draw (31) to (43.center);
		\draw (31) to (25);
		\draw (61.center) to (66.center);
		\draw (63.center) to (68.center);
		\draw (69.center) to (64.center);
		\draw (65.center) to (70.center);
		\draw [style=star edge] (56) to (57);
		\draw [style=star edge] (58) to (59);
		\draw (50) to (52);
		\draw (71) to (53);
		\draw (72) to (62.center);
		\draw (51) to (67.center);
		\draw [style=star edge] (54) to (50);
		\draw (85.center) to (90.center);
		\draw (87.center) to (92.center);
		\draw (93.center) to (88.center);
		\draw (89.center) to (94.center);
		\draw [style=star edge] (82) to (83);
		\draw (75) to (77);
		\draw (95) to (78);
		\draw (96) to (86.center);
		\draw [style=star edge] (79) to (75);
		\draw [style=star edge] (98) to (81);
		\draw (77) to (91.center);
		\draw (107.center) to (112.center);
		\draw (109.center) to (114.center);
		\draw (115.center) to (110.center);
		\draw (111.center) to (116.center);
		\draw (117) to (101);
		\draw (118) to (108.center);
		\draw [style=star edge] (102) to (99);
		\draw (122) to (103);
		\draw [style=star edge] (105) to (99);
		\draw (113.center) to (99);
	\end{pgfonlayer}
\end{tikzpicture}

\end{equation*}

\newpage

Note that this simplification pattern could continue indefinitely, if more CNOTs and star edges were to be attached. It would always end up with one Z-spider
connecting all star edges, thus it is also possible to get two terms for the left term. The right term can be simplified analogously\sidenote{One
has to use the \PiCommutation-rule in order to handle the X-spider with $\pi$-phase.}, also resulting in two terms. We can conclude that using the
CNOT-grouping technique for star edges does not improve upon traditional methods.

It is worth reminding ourselves that this was only one part of the original paper. The other focus was on a weighting algorithm that takes the graph beyond
the immediate neighbors into account. As already mentioned, our own implementation will only use partial simplifications, and therefore we will
see that \reflemma{NOT_obstruction} is still of importance.

\section{Studied Approach}
\labsec{studied_approach}

Let us begin by defining the concept of \textit{multi-control Toffoli gate dense quantum circuits}. In general, these are simply quantum circuits,
where multi-control Toffoli gates have a high probability of occurring across the whole circuit. We will start off by considering the simplest
case\sidenote{In fact, since all the gates in this configuration apply a (invertible) NOT-gate on the same qubit, it is in fact trivial. We will still
show this example for illustrative purposes.}
\begin{equation}
    \input{5_Weighting_Algorithms/multiple_multi_ctrl_toffs_qcirc.tikz}
\end{equation}
for $m$ multi-control Toffoli gates with $n$ control bits each. According to \refeq{multi_ctrl_toff_zx}, we can give the alternative ZX-diagram representation
\begin{equation*}
    \input{5_Weighting_Algorithms/multiple_multi_ctrl_toffs_zx1.tikz}
\end{equation*}
which can be simplified to
\begin{equation}
    \input{5_Weighting_Algorithms/multiple_multi_ctrl_toffs_zx2.tikz} \ .
\end{equation}
Here, there exist two possibilities for what we can decompose\sidenote{In either case, the star edges at the bottom will be simplified automatically.},
highlighted using pink and green boxes. We will refer to these individual spiders as \textit{potential master nodes}. For such simple cases, one
can create a simple heuristic: If $n\leq m$, then we decompose all pink potential master nodes, else, if $n>m$, we decompose all green potential master nodes. 

Of course, such circuits are in practice of little interest due to their low complexity and because they would only occur in rare cases\sidenote{Additionally,
if they were to occur, such a heuristic is of little use as this would automatically be implemented by a greedy algorithm as described before.}.
In \refsec{application_multiloop_feynman_diagrams}, we will look at a concrete example of a class of quantum circuits with a high density of multi-control Toffoli
gates. One pleasant property of this class is that there are no other gates apart from NOT gates\sidenote{Technically, a \textit{diffusion operator} is part
of it, but, as we will see, it can be treated separately.}.

Taking these properties into account, a weighting algorithm was developed. Contrary to traditional algorithms, it takes neighbors of neighbors into account.
Additionally, it considers patterns that could be helpful for further simplifications (see "\textit{extra\_weight}"), something which has not been done
in \cite{sutcliffeProcedurallyOptimisedZXDiagram2024}. As mentioned already, we only utilize a partial simplification strategy, that is, the iterative
application of \SpiderFusion and \StateCopy. Furthermore, we make use of certain preliminary simplifications, such as applying the stack representation,
which we used for the proof of \reflemma{NOT_obstruction}. Therefore, this algorithm is specifically meant for quantum circuits only consisting of multi-control Toffoli
gates and NOT gates.

The main part of the weighting algorithm is the function \textit{get\_master\_v\_weight}, which is part of the PoC\sidenote{The code can be found
in \cite{Tix3DevFeynman_loop_diagram_qsim}.} implemented for the concrete application discussed in \refsec{application_multiloop_feynman_diagrams}:
\begin{lstlisting}[language=Python]
def get_master_v_weight(self, g, v):
    extra_weight = 2
    total_weight = 0
    for l1_nv in g.neighbors(v):
        if g.phase(l1_nv) != star_phase:
            continue
        total_weight += 1
        l2_nv = list(g.neighbors(l1_nv))
        l2_nv.remove(v)
        if len(l2_nv) > 1:
            continue
        l2_nv = l2_nv[0]
        if g.type(l2_nv) != Z:
            continue
        l3_nv = list(g.neighbors(l2_nv))
        l3_nv.remove(l1_nv)
        for cur_v in l3_nv:
            if g.phase(cur_v) == star_phase:
                total_weight += 1
            elif not g.type(cur_v) in [Z,X]:
                continue
            elif g.vertex_degree(cur_v) > 2:
                continue
            else:
                extra_weight_valid = False
                if g.type(cur_v) == X and g.phase(cur_v) == 1:
                    extra_weight_valid = True
                cur_v = list(g.neighbors(cur_v))
                cur_v.remove(l2_nv)
                if len(cur_v) > 1:
                    continue
                cur_v = cur_v[0]
                if g.phase(cur_v) != star_phase:
                    continue
                total_weight += (extra_weight if extra_weight_valid
                                else 1)
    return total_weight
\end{lstlisting}
The algorithm assigns a weight to each potential master node, and will then decompose the one with the highest weight. It performs a depth-first
search. In other words, it starts with a new branch only after having finished the previous branch. Consider the following visualization:

The potential master node $v$ has a branch we want to check\sidenote{The other branches/edges are not shown.}:
\begin{equation*}
    \begin{tikzpicture}
	\begin{pgfonlayer}{nodelayer}
		\node [style=Z dot] (0) at (0, 0) {};
		\node [style=none] (1) at (1, 0) {};
		\node [style=none] (2) at (1.5, 0) {$...$};
		\node [style=none] (3) at (-0.5, 0.5) {$v$};
	\end{pgfonlayer}
	\begin{pgfonlayer}{edgelayer}
		\draw (0) to (1.center);
	\end{pgfonlayer}
\end{tikzpicture}

\end{equation*}
We will only continue searching if the immediate neighbor is a star, in which case the weight is increased by one.
\begin{equation*}
    \begin{tikzpicture}
	\begin{pgfonlayer}{nodelayer}
		\node [style=Z dot] (0) at (0, 0) {};
		\node [style={h_star}] (2) at (1, 0) {};
		\node [style=none] (3) at (2, 0) {};
		\node [style=none] (4) at (2.5, 0) {$...$};
	\end{pgfonlayer}
	\begin{pgfonlayer}{edgelayer}
		\draw (0) to (3.center);
	\end{pgfonlayer}
\end{tikzpicture}

\end{equation*}
The searching will only continue if there is a Z-spider as its only neighbor, that can potentially have multiple neighbors. It only considers
Z-spiders as other patterns were not observed sufficiently enough.
\begin{equation*}
    \begin{tikzpicture}
	\begin{pgfonlayer}{nodelayer}
		\node [style=Z dot] (0) at (0, 0) {};
		\node [style={h_star}] (1) at (1, 0) {};
		\node [style=Z dot] (2) at (2, 0) {};
		\node [style=none] (4) at (3, -1) {};
		\node [style=none] (5) at (3, 1) {};
		\node [style=none, rotate=90] (6) at (2.75, 0) {$...$};
		\node [style=none] (7) at (3.5, 1) {$...$};
		\node [style=none] (8) at (3.5, -1) {$...$};
	\end{pgfonlayer}
	\begin{pgfonlayer}{edgelayer}
		\draw (2) to (0);
		\draw (5.center) to (2);
		\draw (2) to (4.center);
	\end{pgfonlayer}
\end{tikzpicture}

\end{equation*}
Here, each neighbor is checked. The weight is increased by one if it is a star. Only if the neighbor has a degree of two, the search on that
branch is not terminated. If it is an X-spider with phase $0$, a Z-spider with phase $0$ or $\pi$ followed by a star, the weight is increased
by one. If it is an X-spider with phase $\pi$, then the weight is increased by \textit{extra\_weight}, which was chosen to be two. The reason
for this can be seen in the following illustration, where $v$ is decomposed using \refthm{elementary_decomp}:
\begin{equation*}
    \begin{tikzpicture}
	\begin{pgfonlayer}{nodelayer}
		\node [style=Z dot] (0) at (-9.75, 1.75) {};
		\node [style={h_star}] (1) at (-8.75, 1.75) {};
		\node [style=Z dot] (2) at (-7.75, 1.75) {};
		\node [style=none] (3) at (-6.75, 0.75) {};
		\node [style=none] (4) at (-6.75, 2.75) {};
		\node [style=none, rotate=90] (5) at (-7, 1.75) {$...$};
		\node [style=none] (6) at (-6.25, 2.75) {$...$};
		\node [style=none] (7) at (-6.25, 0.75) {$...$};
		\node [style=none] (8) at (-5.25, 2) {};
		\node [style=none] (9) at (-4.5, 3.75) {};
		\node [style={h_star}] (11) at (-2.75, 3.75) {};
		\node [style=Z dot] (12) at (-1.75, 3.75) {};
		\node [style=none] (13) at (-0.75, 2.75) {};
		\node [style=none] (14) at (-0.75, 4.75) {};
		\node [style=none, rotate=90] (15) at (-1, 3.75) {$...$};
		\node [style=none] (16) at (-0.25, 4.75) {$...$};
		\node [style=none] (17) at (-0.25, 2.75) {$...$};
		\node [style=X dot] (18) at (-3.75, 3.75) {};
		\node [style={h_star}] (19) at (-2.75, -0.25) {};
		\node [style=Z dot] (20) at (-1.75, -0.25) {};
		\node [style=none] (21) at (-0.75, -1.25) {};
		\node [style=none] (22) at (-0.75, 0.75) {};
		\node [style=none, rotate=90] (23) at (-1, -0.25) {$...$};
		\node [style=none] (24) at (-0.25, 0.75) {$...$};
		\node [style=none] (25) at (-0.25, -1.25) {$...$};
		\node [style=none] (27) at (-5.25, 1.5) {};
		\node [style=none] (28) at (-4.5, -0.25) {};
		\node [style=X phase dot] (29) at (-3.75, -0.25) {$\pi$};
		\node [style=none] (30) at (0.25, 3.75) {};
		\node [style=none] (31) at (1.5, 3.75) {};
		\node [style=Z dot] (33) at (3.25, 3.75) {};
		\node [style=none] (34) at (4.25, 2.75) {};
		\node [style=none] (35) at (4.25, 4.75) {};
		\node [style=none, rotate=90] (36) at (4, 3.75) {$...$};
		\node [style=none] (37) at (4.75, 4.75) {$...$};
		\node [style=none] (38) at (4.75, 2.75) {$...$};
		\node [style=Z dot] (39) at (2.25, 3.75) {};
		\node [style=none] (40) at (0.25, -0.25) {};
		\node [style=none] (41) at (1.5, -0.25) {};
		\node [style=Z dot] (42) at (3.25, -0.25) {};
		\node [style=none] (43) at (4.25, -1.25) {};
		\node [style=none] (44) at (4.25, 0.75) {};
		\node [style=none, rotate=90] (45) at (4, -0.25) {$...$};
		\node [style=none] (46) at (4.75, 0.75) {$...$};
		\node [style=none] (47) at (4.75, -1.25) {$...$};
		\node [style=X dot] (48) at (2.25, -0.25) {};
		\node [style=none] (49) at (5.25, -0.25) {};
		\node [style=none] (50) at (6.5, -0.25) {};
		\node [style=none] (52) at (8.25, -1.25) {};
		\node [style=none] (53) at (8.25, 0.75) {};
		\node [style=none, rotate=90] (54) at (8, -0.25) {$...$};
		\node [style=none] (55) at (8.75, 0.75) {$...$};
		\node [style=none] (56) at (8.75, -1.25) {$...$};
		\node [style=X dot] (57) at (7.25, 0.75) {};
		\node [style=X dot] (58) at (7.25, -1.25) {};
		\node [style=none] (59) at (9.5, -1.25) {};
		\node [style=none] (60) at (10.5, -1.75) {};
		\node [style=none] (61) at (9.5, -2.25) {};
		\node [style=none] (62) at (-9.5, -2.25) {};
		\node [style=none] (63) at (-10.5, -2.75) {};
		\node [style=none] (64) at (-9.5, -3.25) {};
		\node [style=none] (65) at (-9.5, -4.25) {};
		\node [style=none] (66) at (-9.5, -5.25) {};
		\node [style=none] (67) at (-9.5, -6.25) {};
		\node [style=none] (68) at (-10.5, -3.75) {};
		\node [style=none] (69) at (-10.5, -4.75) {};
		\node [style=none] (70) at (-10.5, -5.75) {};
		\node [style=none] (71) at (-6.5, -3.25) {};
		\node [style=none] (72) at (-6, -3.25) {$...$};
		\node [style=X dot] (73) at (-8.5, -3.25) {};
		\node [style=none] (74) at (-6.5, -4.25) {};
		\node [style=none] (75) at (-6, -4.25) {$...$};
		\node [style=X dot] (76) at (-8.5, -4.25) {};
		\node [style=none] (77) at (-6.5, -5.25) {};
		\node [style=none] (78) at (-6, -5.25) {$...$};
		\node [style=X dot] (79) at (-8.5, -5.25) {};
		\node [style=none] (80) at (-6.5, -6.25) {};
		\node [style=none] (81) at (-6, -6.25) {$...$};
		\node [style=X dot] (82) at (-8.5, -6.25) {};
		\node [style=Z dot] (83) at (-7.5, -3.25) {};
		\node [style=X dot] (85) at (-7.5, -5.25) {};
		\node [style=X phase dot] (86) at (-7.5, -6.25) {$\pi$};
		\node [style=Z phase dot] (87) at (-7.5, -4.25) {$\pi$};
		\node [style=none] (88) at (-5.25, -3.25) {};
		\node [style=none] (89) at (-4.5, -3.25) {};
		\node [style=none] (90) at (-5.25, -6.25) {};
		\node [style=none] (91) at (-4.5, -6.25) {};
		\node [style=none] (92) at (-5.25, -4.25) {};
		\node [style=none] (93) at (-4.5, -4.25) {};
		\node [style=none] (94) at (-5.25, -5.25) {};
		\node [style=none] (95) at (-4.5, -5.25) {};
		\node [style=X dot] (96) at (-3.5, -3.25) {};
		\node [style=none] (97) at (-2.5, -3.25) {};
		\node [style=X dot] (98) at (-3.5, -4.25) {};
		\node [style=none] (99) at (-2.5, -4.25) {};
		\node [style=X dot] (100) at (-3.5, -5.25) {};
		\node [style=none] (101) at (-2.5, -5.25) {};
		\node [style=none] (103) at (-2.5, -6.25) {};
		\node [style=none] (105) at (0.25, -4.5) {$\text{If there}$};
		\node [style=none] (106) at (0.25, -5.25) {$\text{is a star}$};
		\node [style=none] (107) at (-1.25, -4) {};
		\node [style=none] (108) at (-1.25, -5.75) {};
		\node [style=none] (109) at (1.75, -5.75) {};
		\node [style=none] (110) at (1.75, -4) {};
		\node [style=X phase dot] (111) at (-3.5, -6.25) {$\pi$};
		\node [style=X dot] (112) at (2.5, -3.25) {};
		\node [style=none] (113) at (4.5, -3.25) {};
		\node [style=X dot] (114) at (2.5, -4.25) {};
		\node [style=none] (115) at (4.5, -4.25) {};
		\node [style=X dot] (116) at (2.5, -5.25) {};
		\node [style=none] (117) at (4.5, -5.25) {};
		\node [style=none] (118) at (4.5, -6.25) {};
		\node [style=X phase dot] (119) at (2.5, -6.25) {$\pi$};
		\node [style={h_star}] (120) at (3.5, -3.25) {};
		\node [style={h_star}] (121) at (3.5, -4.25) {};
		\node [style={h_star}] (122) at (3.5, -5.25) {};
		\node [style={h_star}] (123) at (3.5, -6.25) {};
		\node [style=none] (124) at (5.75, -3.25) {};
		\node [style=none] (125) at (6.5, -3.25) {};
		\node [style=none] (126) at (5.75, -6.25) {};
		\node [style=none] (127) at (6.5, -6.25) {};
		\node [style=none] (128) at (5.75, -4.25) {};
		\node [style=none] (129) at (6.5, -4.25) {};
		\node [style=none] (130) at (5.75, -5.25) {};
		\node [style=none] (131) at (6.5, -5.25) {};
		\node [style=none] (132) at (-2, -3.25) {$...$};
		\node [style=none] (133) at (-2, -4.25) {$...$};
		\node [style=none] (134) at (-2, -5.25) {$...$};
		\node [style=none] (135) at (-2, -6.25) {$...$};
		\node [style=none] (136) at (5, -3.25) {$...$};
		\node [style=none] (137) at (5, -4.25) {$...$};
		\node [style=none] (138) at (5, -5.25) {$...$};
		\node [style=none] (139) at (5, -6.25) {$...$};
		\node [style=none] (140) at (8.5, -3.25) {};
		\node [style=none] (141) at (8.5, -4.25) {};
		\node [style=none] (142) at (8.5, -5.25) {};
		\node [style=none] (143) at (8.5, -6.25) {};
		\node [style=none] (144) at (9, -3.25) {$...$};
		\node [style=none] (145) at (9, -4.25) {$...$};
		\node [style=none] (146) at (9, -5.25) {$...$};
		\node [style=none] (147) at (9, -6.25) {$...$};
		\node [style=Z dot] (148) at (7.5, -3.25) {};
		\node [style=Z dot] (149) at (7.5, -4.25) {};
		\node [style=Z dot] (150) at (7.5, -5.25) {};
		\node [style=X dot] (151) at (7.5, -6.25) {};
		\node [style=none] (152) at (9.5, -6.25) {};
		\node [style=none] (153) at (10.25, -6.25) {};
		\node [style=none] (158) at (7.25, -2.75) {};
		\node [style=none] (159) at (7.25, -3.75) {};
		\node [style=none] (160) at (8.25, -3.75) {};
		\node [style=none] (161) at (8.25, -2.75) {};
		\node [style=none] (162) at (7.25, -3.75) {};
		\node [style=none] (163) at (7.25, -4.75) {};
		\node [style=none] (164) at (8.25, -4.75) {};
		\node [style=none] (165) at (8.25, -3.75) {};
		\node [style=none] (166) at (7.25, -4.75) {};
		\node [style=none] (167) at (7.25, -5.75) {};
		\node [style=none] (168) at (8.25, -5.75) {};
		\node [style=none] (169) at (8.25, -4.75) {};
		\node [style=none] (174) at (7.25, -5.75) {};
		\node [style=none] (175) at (7.25, -6.75) {};
		\node [style=none] (176) at (8.25, -6.75) {};
		\node [style=none] (177) at (8.25, -5.75) {};
		\node [style=none] (178) at (-8, -5.75) {};
		\node [style=none] (179) at (-8, -6.75) {};
		\node [style=none] (180) at (-7, -6.75) {};
		\node [style=none] (181) at (-7, -5.75) {};
	\end{pgfonlayer}
	\begin{pgfonlayer}{edgelayer}
		\draw (2) to (0);
		\draw (4.center) to (2);
		\draw (2) to (3.center);
		\draw [style=diredge] (8.center) to (9.center);
		\draw (14.center) to (12);
		\draw (12) to (13.center);
		\draw (18) to (12);
		\draw (22.center) to (20);
		\draw (20) to (21.center);
		\draw [style=diredge] (27.center) to (28.center);
		\draw (29) to (20);
		\draw [style=diredge] (30.center) to (31.center);
		\draw (35.center) to (33);
		\draw (33) to (34.center);
		\draw [style=diredge] (40.center) to (41.center);
		\draw (44.center) to (42);
		\draw (42) to (43.center);
		\draw [style=diredge] (49.center) to (50.center);
		\draw (58) to (52.center);
		\draw (57) to (53.center);
		\draw [in=90, out=0] (59.center) to (60.center);
		\draw [in=0, out=-90] (60.center) to (61.center);
		\draw (61.center) to (62.center);
		\draw [in=90, out=-180] (62.center) to (63.center);
		\draw [style=diredge, in=-180, out=-90] (63.center) to (64.center);
		\draw [style=diredge, in=180, out=-90] (68.center) to (65.center);
		\draw [style=diredge, in=180, out=-90] (69.center) to (66.center);
		\draw [style=diredge, in=180, out=-90] (70.center) to (67.center);
		\draw (63.center) to (70.center);
		\draw (73) to (71.center);
		\draw (76) to (74.center);
		\draw (79) to (77.center);
		\draw [style=diredge] (88.center) to (89.center);
		\draw [style=diredge] (90.center) to (91.center);
		\draw [style=diredge] (92.center) to (93.center);
		\draw [style=diredge] (94.center) to (95.center);
		\draw (96) to (97.center);
		\draw (98) to (99.center);
		\draw (100) to (101.center);
		\draw (107.center) to (110.center);
		\draw (110.center) to (109.center);
		\draw (109.center) to (108.center);
		\draw (108.center) to (107.center);
		\draw (111) to (103.center);
		\draw (112) to (113.center);
		\draw (114) to (115.center);
		\draw (116) to (117.center);
		\draw (119) to (118.center);
		\draw [style=diredge] (124.center) to (125.center);
		\draw [style=diredge] (126.center) to (127.center);
		\draw [style=diredge] (128.center) to (129.center);
		\draw [style=diredge] (130.center) to (131.center);
		\draw [style=diredge] (152.center) to (153.center);
		\draw [style={yellow_highlight}] (160.center)
			 to (161.center)
			 to (158.center)
			 to (159.center)
			 to cycle;
		\draw [style={yellow_highlight}] (162.center)
			 to (163.center)
			 to (164.center)
			 to (165.center)
			 to cycle;
		\draw [style={yellow_highlight}] (166.center)
			 to (167.center)
			 to (168.center)
			 to (169.center)
			 to cycle;
		\draw [style={blue_highlight}] (176.center)
			 to (177.center)
			 to (174.center)
			 to (175.center)
			 to cycle;
		\draw (33) to (39);
		\draw (48) to (42);
		\draw (151) to (143.center);
		\draw (142.center) to (150);
		\draw (149) to (141.center);
		\draw (140.center) to (148);
		\draw [style={red_highlight}] (180.center)
			 to (181.center)
			 to (178.center)
			 to (179.center)
			 to cycle;
		\draw (82) to (80.center);
	\end{pgfonlayer}
\end{tikzpicture}

\end{equation*}
In the end, it is clear why an extra weight is useful: An X-spider could remove stars that might follow, whereas a Z-spider could not.

\newpage

\section{Application for Multiloop Feynman Diagrams}
\labsec{application_multiloop_feynman_diagrams}

In a quest to find examples of multi-control Toffoli gate dense quantum circuits with concrete applications, we came across a whole class of
such circuits in \cite{ramirez-uribeQuantumQueryingBased2024}. One such example can be seen in \reffig{4eloop_12edges_qcirc}.
\begin{figure*}[!ht]
	\includegraphics{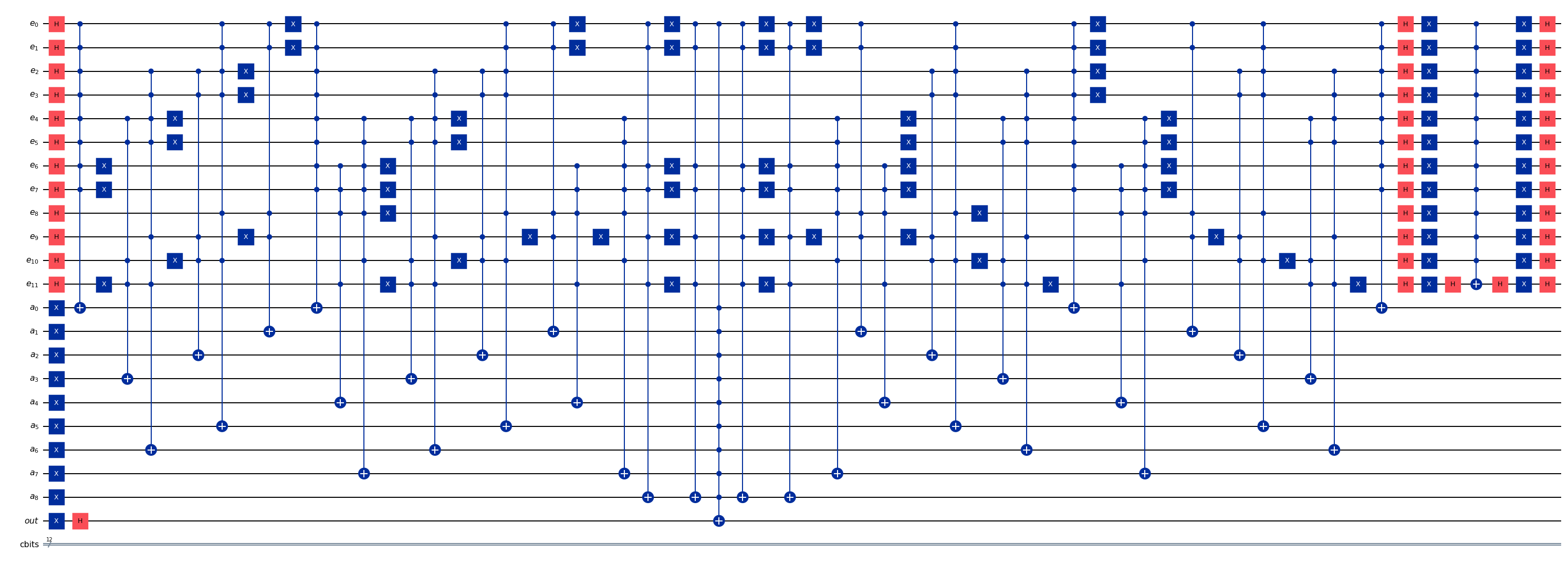}
	\caption[4 eloop 12 edges quantum circuit]{Multi-control Toffoli gate dense quantum circuit found in \cite{ramirez-uribeQuantumQueryingBased2024}.}
	\labfig{4eloop_12edges_qcirc}
\end{figure*}

Such quantum circuits implement a quantum algorithm based on \textit{Grover's algorithm}, a famous quantum search algorithm, which beats the classical
space-time complexity of $\mathcal{O}(N)$, only requiring $\mathcal{O}(\sqrt{N})$, for $N$ being the number of elements searched. More concretely,
it searches for \textit{causal singular configurations} of \textit{multiloop Feynman diagrams} \cite{ramirez-uribeQuantumQueryingBased2024}. The exact details of
the theory behind the algorithm lie beyond the scope of this thesis. However, the following will provide a brief overview.

In \refch{graphical_languages}, we very briefly introduced QFT and scattering processes. Calculating the scattering amplitude $\mathcal{A}$ of
an experiment conducted using a particle accelerator is one of the cornerstones of modern particle physics. One common approach is to use
perturbation theory, which expresses the amplitude as a power series, where each term corresponds to increasingly more complex interactions during
the scattering process. This allows physicists to approximate scattering processes by considering only the most significant contributions, typically
the first few terms in the series \cite{10.1093/acprof:oso/9780199699322.001.0001}. These terms can be represented and computed as Feynman diagrams.
However, such computations get cumbersome for more complex \textit{topologies}, in other words, when the number of \textit{external legs} in the Feynman diagrams
increases, or when the \textit{loop order} increases \cite{torresbobadillaLottyLooptreeDuality2021}.
\begin{example}
    The following illustrates the loop order, increasing from left to right:
    \begin{equation*}
        \input{5_Weighting_Algorithms/loop_order.tikz}
    \end{equation*}
    In this case, the number of external legs (colored red and blue) was kept constant. The number of couplings (colored green) did
    change, however \cite{bobadillaOffshellJacobiCurrents}.
\end{example}

\newpage

It is therefore an active area of research to find more efficient methods for such computations. One proposed alternative is based on the
\textit{loop-tree duality formalism}. The basic idea is to reinterpret complicated loop diagrams in terms of simpler tree-level
diagrams \cite{torresbobadillaLottyLooptreeDuality2021}. More mathematically\sidenote{Of course, all the following definitions are only listed for completeness, and
the reader is not expected to give them more thought. Interested readers should refer to \cite{10.1093/acprof:oso/9780199699322.001.0001},
\cite{ramirez-uribeQuantumQueryingBased2024} and \cite{ramirez-uribeQuantumAlgorithmFeynman2022}.}, one would start with \textit{generic loop integrals} and
\textit{scattering amplitudes}, which are written as
\begin{equation}
    \mathcal{A}_F^{(L)}=\int_{\ell_1 \ldots \ell_L} \mathcal{N}\left(\left\{\ell_s\right\}_L,\left\{p_j\right\}_P\right) \prod_{i=1}^n G_F\left(q_i\right)
\end{equation}
where $n$ is the number of propagators, $P$ is the number of external legs, $L$ is the number of loop momenta, $\ell_s$ with $s\in\{1,\dots,L\}$ are the primitive
loop momenta, $p_j$ with $j\in\{1,\dots,P\}$ are the external momenta, $q_i$ are the momenta flowing through the Feynman propagator $G_F$ and $\mathcal{N}$ is the
numerator given by the specific theory being used (\cf \cite{ListQuantumField2024}).

By applying \textit{Cauchy's residue theorem}, it is then possible to obtain
the \textit{loop-tree duality causal representation} \cite{ramirez-uribeQuantumQueryingBased2024}, which is written as
\begin{equation}\labeq{ltd_causal_representation}
    \mathcal{A}_D^{(L)}=\int_{\vec{\ell}_1 \ldots \vec{\ell}_L} \frac{1}{x_n} \sum_{\sigma \in \Sigma} \frac{\mathcal{N}_{\sigma\left(i_1, \ldots, i_{n-L}\right)}}{\lambda_{\sigma\left(i_1\right)}^{h_{\sigma\left(i_1\right)}} \cdots \lambda_{\sigma\left(i_{n-L}\right)}^{h_{\sigma\left(i_n\right)}}} +\left(\lambda_p^{+} \leftrightarrow \lambda_p^{-}\right),
\end{equation}
where, $x_n$ is a normalization factor, a given $\lambda$ is a causal propagator and a given $\sigma$ corresponds to a specific permutation, which is an element of
the set $\Sigma$, which in turn is determined by the causal configurations. This set is therefore crucial to ensure that only physical contributions are included in
the calculation.

Although the details are very advanced, the essence of Cauchy's residue theorem can be explained using a simple example, taken
from \cite{95CauchyResidue2017}.
\begin{theorem}\labthm{cauchy_residue}
    \cite{95CauchyResidue2017} Suppose $f(z)$ is analytic in the region $A$ except for a set of isolated singularities. Also suppose $C$ is a simple, closed curve in $A$ that
    does not go through any of the singularities of $f$ and is oriented counterclockwise. Then
    \begin{equation}
        \int_C f(z) d z=2 \pi i \sum_{k=1}^n \text { residue of } f \text { at } a_k,
    \end{equation}
    where the $a_k\in\mathbb{C}$ are the poles of the complex function $f$ lying inside $C$.
\end{theorem}
\begin{marginfigure}[*+7]
    \includegraphics{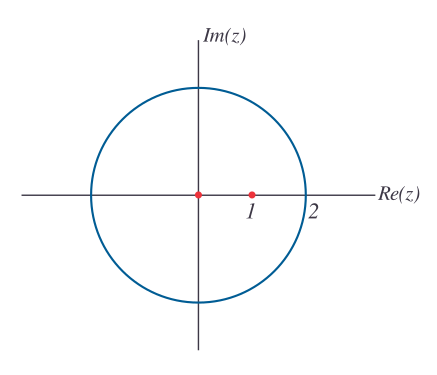}
    \caption[Contour for simple cauchys residue theorem example]{Contour of $f$ taken from \cite{95CauchyResidue2017}. Note that the contour has to be considered in counterclockwise direction.}
    \labfig{contour_for_cauchy_residue_example}
\end{marginfigure}
\begin{example}
    We want to compute the following integral:
    \begin{equation*}
        \int_{|z|=2} \frac{5 z-2}{z(z-1)} d z .
    \end{equation*}
    Let
    \begin{equation*}
        f(z)=\frac{5 z-2}{z(z-1)}.
    \end{equation*}
    In order to use \refthm{cauchy_residue}, we need to find the residue of $f$ at each of its poles. If $f$ has a \textit{simple pole}, then we can use
    a relatively simple formula for the residue \cite{10.1093/acprof:oso/9780199699322.001.0001}, given by
    \begin{equation}
        R(z_0)=\lim _{z \rightarrow z_0}(z-z_0) f(z).
    \end{equation}
    In our case we have two simple poles at $z=0$ and $z=1$, which, when plugged into $f$, would result in undefined behavior.
    \refthm{cauchy_residue} requires us only to use the residue at the poles lying inside the contour, which is true for both, as illustrated
    in \reffig{contour_for_cauchy_residue_example}.

    At $z=0$, we have
    \begin{equation*}
        R(0)=\lim_{z\rightarrow 0} (z-0)\frac{5 z-2}{z(z-1)}=\lim_{z\rightarrow 0} \frac{5 z-2}{z-1}=2.
    \end{equation*}
    And at $z=1$, we have
    \begin{equation*}
        R(0)=\lim_{z\rightarrow 1} (z-1)\frac{5 z-2}{z(z-1)}=\lim_{z\rightarrow 1} \frac{5 z-2}{z}=3.
    \end{equation*}

    We can now apply \refthm{cauchy_residue}:
    \begin{equation*}
        \int_C \frac{5 z-2}{z(z-1)} d z=2 \pi i[R(0)+R(1)]=10 \pi i .
    \end{equation*}
\end{example}

That being said, obtaining the causal representation in \refeq{ltd_causal_representation} requires finding causal configurations. This is exactly
what the algorithm searches for using the proposed quantum circuits.

The paper outlines how to create such a quantum circuit corresponding to a given topology using certain boolean functions. For example, the
two eloop topology with six edges corresponds to the following quantum circuit:
\begin{figure*}[!ht]
	\includegraphics{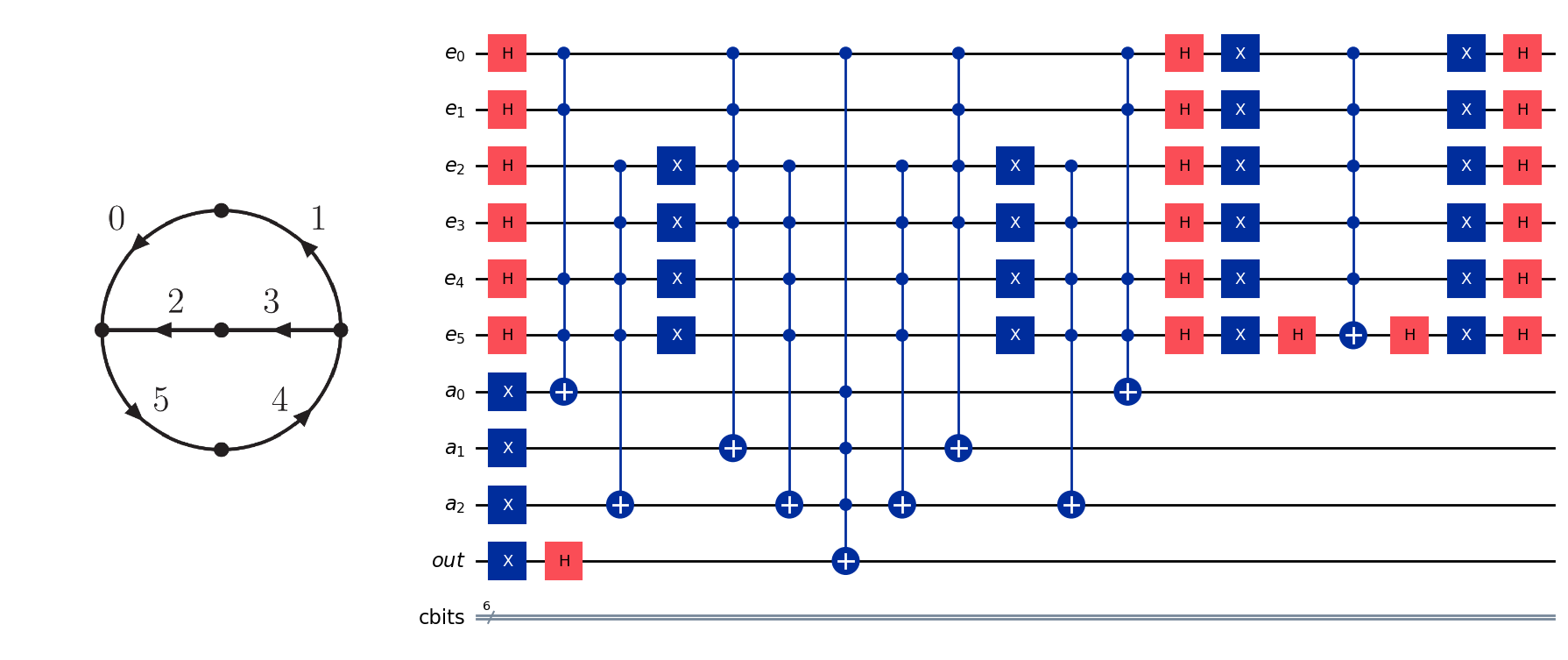}
	\caption[2 eloop 6 edges quantum circuit]{Two eloop topology with six edges (taken from \cite{ramirez-uribeQuantumQueryingBased2024})
    converted to a quantum circuit (generated in \cite{Tix3DevFeynman_loop_diagram_qsim}).}
	\labfig{2eloop_6edges_qcirc}
\end{figure*}

Note that the distinct right part of the quantum circuit\sidenote{We are referring to the right-most multi-control Toffoli gate surrounded by Hadamard
and NOT-gates.} is the \textit{diffusion operator}, a key aspect of Grover's algorithm \cite{GroversAlgorithm2024}.

In brief, Grover's algorithm is a \textit{quantum query algorithm} for unstructured search. This means that the search algorithm is oracle-based: Given
a function satisfying certain properties, we want to recover some information about the function\sidenote{But only with a certain probability.}. This should be done
using a minimal number of queries to the corresponding \textit{quantum oracle}\sidenote{This is similar to how the Oracle at Delphi would never reveal
everything she knows, but only what she is being asked.}.

\newpage

\begin{kaobox}[frametitle={Problem statement for Grover's algorithm}]
    Given a function $f:\{0,1\}^n\to\{0,1\}$, we can query it using the following quantum oracle given by the unitary operator $U_{\omega}$:
    \begin{equation*}
        \begin{cases}U_\omega|x\rangle=-|x\rangle & \text { for } x=\omega, \text { that is, } f(x)=1 \\ U_\omega|x\rangle=|x\rangle & \text { for } x \neq \omega, \text { that is, } f(x)=0\end{cases}
    \end{equation*}
    Alternatively, we can write
    \begin{equation}\labeq{quantum_oracle_def}
        U_\omega\ket{x}=(-1)^{f(x)}\ket{x} .
    \end{equation}
    We assume that only one $x$ satisfies $f(x)=1$.
    The goal is to identify $\omega$ with high probability \cite{GroversAlgorithm2024}.
\end{kaobox}

It can be seen that \refeq{quantum_oracle_def} is analogous to the oracle defined in \cite{ramirez-uribeQuantumQueryingBased2024}, which is given by
\begin{equation*}
    U_\omega\ket{e}\ket{a}\ket{out}=(-1)^{f(x)}\ket{e}\ket{a}\ket{out}
\end{equation*}

Grover's algorithm consists of the quantum oracle and the diffusion operator, which are applied $t$ times on a superposition of all $x$. With high probability,
we can say that the resulting state is $\omega$. The algorithm presented in \cite{ramirez-uribeQuantumQueryingBased2024} exhibits certain nuances. For example,
they demonstrated that the most effective results for most multiloop topologies are obtained when $t=1$. This means that we can simulate the quantum circuits
shown in \reffig{4eloop_12edges_qcirc} and \reffig{2eloop_6edges_qcirc} as they are, without the need to repeat the circuits.

The quantum circuit in \reffig{2eloop_6edges_qcirc} alongside with the quantum circuit in \reffig{4eloop_12edges_qcirc}, which corresponds to a four eloop topology with
twelve edges, were implemented using Qiskit \cite{qiskit2024} in order to verify the results from \cite{ramirez-uribeQuantumQueryingBased2024}.
Indeed, 23 causal configurations were found for the two eloop topology, and 1199 causal configurations for the four eloop topology, matching the
obtained values from the paper. The implementation can be found in \cite{Tix3DevFeynman_loop_diagram_qsim}.

The Proof-of-Concept for our algorithm described in \refsec{studied_approach} was implemented in \cite{Tix3DevFeynman_loop_diagram_qsim} as well.
The program uses a ZX-diagram representation of the two quantum circuits from before. It then applies the algorithm on the main part, that is, excluding
the diffusion operator. When it arrives at a fully simplified Clifford representation of all terms, highlighted in yellow, it will connect it
to the diffusion operator. This will then get simplified and decomposed if necessary. An example can be seen here:
\begin{equation*}
    \input{5_Weighting_Algorithms/example_diff_op_term_pre_simp.tikz} \ \ \leadsto \ \ \input{5_Weighting_Algorithms/example_diff_op_term_post_simp.tikz}
\end{equation*}
In this case, a decomposition would be necessary.

Once we are left with only Clifford terms, we have found our final decomposition. All terms will be converted to statevectors, which will then get
summed up, in order to obtain the final statevector. We can compute the probability for each state as described in \refch{quantum_computation}.
The resulting probability distribution can be seen in \reffig{prob_distr_output_comparison}.
\begin{figure*}[!ht]
	\includegraphics{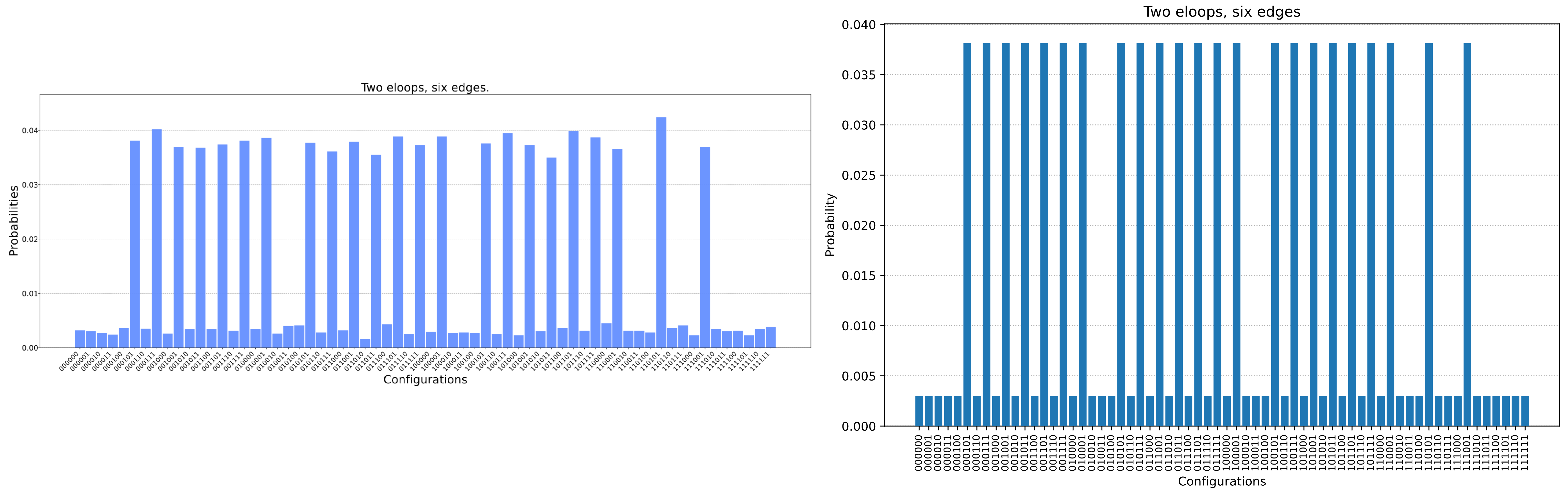}
	\caption[Probability distribution output comparison]{Left: Probability distribution from \cite{ramirez-uribeQuantumQueryingBased2024}.
    Right: Probability distribution generated using \cite{Tix3DevFeynman_loop_diagram_qsim}. Both are for the two eloop topology with six edges.}
	\labfig{prob_distr_output_comparison}
\end{figure*}

By counting the number of peaks\sidenote{This was done computationally by calculating a threshold given by the average of the maximum and the minimum
probability, and then checking which states were above this threshold.}, one can determine the number of causal configurations. Indeed, the plot
matches the probability distribution from the original paper.

Readers interested in the steps taken before applying the weighting algorithm should refer to the Jupyter notebook in \cite{Tix3DevFeynman_loop_diagram_qsim},
where by setting the first parameter of \textit{qalgo.run(show\_diagrams=True, check\_decomposition=True)} to \textit{True}, one can see each major
step of the program graphically. In brief, all the NOT gates in the quantum circuit are pushed to the sides by applying the \StateCopy-rule. Since
this can alter the form of the multi-control Toffoli gates\sidenote{It will create a structure similar to the one found in \refeq{stacks_repr_use_case}.},
we apply the stack representation described in \refeq{stacks_repr}, used for the proof of \reflemma{NOT_obstruction}. This makes the weighting process
easier, as all the edges of interest connected to the potential master nodes are of the same form\sidenote{This will result in two stacks per row.
Since a decomposition applied to one of the two stacks causes the same effect on the other, the two weights are treated
as one weight. Readers interested in more details should refer to the code implementation.}. Finally, the \StateCopy-rule is applied on all
present states as a preliminary simpliflication routine.

It should be noted, however, that the implementation utilizes W nodes (\cf \refsec{related_calculi}) to represent the stars, and placeholder numbers as
phases, which is a programmatic trick to mark potential master nodes. This is deprecated, and was therefore replaced in the subsequent implementation.
The better alternative is to use zero-labelled H-boxes (\cf \refeq{star_H_box_equiv}) and \textit{vdata}. Both features can be easily implemented using
PyZX \cite{kissingerPyZXLargeScale2020}, the main library used for both implementations.

Now, in order to compare our weighting algorithm to the algorithm developed in \cite{kochSpeedyContractionZX2023a}, the quizx version by Mark Koch
was modified to also work with non-scalar diagrams. This made it possible to decompose the quantum circuits for the two and four eloop topologies
using their algorithm. The implementation can also be found in \cite{Tix3DevFeynman_loop_diagram_qsim}.

The measured runtimes compare as follows:
\begin{equation*}
    \begin{array}{|l|l|l|}
        \hline & \text{Modified quizx} & \text{Our implementation} \\
        \hline \text{Two eloop topology} & \SI{0.0111416}{\second} & \approx\SI{1.2}{\second} \text{ (} \SI{0.035086}{\second} \text{ for} \\
        \text{with six edges} & & \text{weighting algorithm)} \\
        \hline \text{Four eloop topology} & \SI{2.0793388}{\second} & \approx\SI{81.9}{\second} \text{ (} \SI{1.807702}{\second} \text{ for} \\
        \text{with twelve edges} & & \text{weighting algorithm)} \\
        \hline
    \end{array}
\end{equation*}
It is clear that in this domain, no improvements could be made. This outcome is anticipated, given that our implementation is written in
Python, utilizing PyZX as the main library, whereas the modified quizx implementation is written in Rust, utilizing quizx. Both Rust and
quizx are known to achieve much faster performance than their Python counterpart \cite{PythonVSRust} \cite{ZxcalcQuizx2024}. For instance,
certain graph manipulations\sidenote{The part that took the longest to execute in the four eloop topology case was the translation of ZX-diagrams
into statevectors and subsequent vector computations.} in quizx might be $5700$ times faster than the corresponding
implementation in PyZX \cite{ZxcalcQuizx2024}. Additionally, our implementation is a very unoptimized Proof-of-Concept.

As mentioned in \cite{sutcliffeProcedurallyOptimisedZXDiagram2024}, it is difficult to give an exact formula for the space-time complexity
of such weighting algorithms, but the rough upper-bound will be similar to the polynomial one from the original paper, which should
be insignificant in comparison to the exponential process of producing stabilizer terms. This also matches the measurements of the time taken
only for the weighting algorithm. We conclude that the obtained results were expected and are thus not unsatisfactory.

We are primarily interested in the final count of stabilizer terms, as it is the only valid metric for concrete comparisons, being independent of
implementation details and thus ensuring an objective evaluation. The obtained results are as follows:
\begin{equation*}
    \begin{array}{|l|l|l|}
        \hline & \text{Modified quizx} & \text{Our implementation} \\
        \hline \text{Two eloop topology} & 52 \text{ terms} & 48 \text{ terms} \\
        \hline \text{Four eloop topology} & 1810 \text{ terms} & 1948 \text{ terms} \\
        \hline
    \end{array}
\end{equation*}
Using our algorithm, it was indeed possible to improve the final number of terms for the simpler topology, although only by a rather
small margin. For the more complex topology, the modified quizx version outperformed our algorithm by a similar relative margin. We suspect that
the main reason for not getting better results is our partial simplification strategy. The modified quizx version utilizes a full simplification strategy.
Since the smaller topology corresponds to a smaller quantum circuit, simplifications will not have a great influence. For bigger topologies and thus also
bigger circuits, simplifications can outweigh the expected improvement of a weighting algorithm. In \refsec{importance_of_simp_strat}, we will see compelling
evidence for this hypothesis by considering randomly generated quantum circuits with similar properties to those studied in this section.

\newpage

\section{Importance of Simplification Strategies}
\labsec{importance_of_simp_strat}

At the beginning of \refsec{basic_idea}, we outlined the reasoning behind our decision to utilize a partial simplification strategy instead of full
simplification approach. The goal was to mimic certain aspects of the original algorithm introduced in \cite{sutcliffeProcedurallyOptimisedZXDiagram2024}.
It was thought that their reasoning would also apply to our case, namely that full simplification strategies could destroy potentially useful patterns.
Our partial decomposition strategy was therefore able to use the idea of stacks, as described in the last section.
In \cite{sutcliffeProcedurallyOptimisedZXDiagram2024}, the authors speculated that more advanced simplification routines could have further contributed
to reducing the final number of terms, but due to a lack of evidence for this speculation, we decided to start with partial simplifications.

Despite this, in the last section we have seen preliminary evidence that our partial simplification strategy hinders us from getting a reliable improvement
over the modified quizx version. To strengthen the case for our hypothesis, we will discuss a benchmark in the subsequent discussion.

The benchmark will apply our algorithm and the modified quizx version on randomly generated multi-control Toffoli gate dense quantum circuits with
variable amounts of NOT gates, CNOT gates and multi-control Toffoli gates\sidenote{The target bits of the multi-control Toffoli gates are distributed over the
bottom qubits, in order to create a similar pattern to those found in the quantum circuits from \refsec{application_multiloop_feynman_diagrams}.}. Additionally,
the number of qubits will be varied. Each configuration will be sampled 50 times with different random seeds. If the execution of either algorithm takes more than
three minutes, it will be interrupted. Although the benchmark used parallelization, the computation took more than 20 hours.
\begin{remark}
    For the two considered topologies from \refsec{application_multiloop_feynman_diagrams} we checked that no star decompositions had affected the results.
    This benchmark, on the other hand, occasionally resulted in star decompositions being used. Since this gives a more realistic comparison between our
    algorithm and the state-of-the-art approach, we still accepted it. Additionally, it should be noted that the benchmark uses scalar diagrams, in order to compare
    our algorithm to the original quizx version by Mark Koch, and not to the modified version.
\end{remark}
Due to the three minute timeout, certain configurations contain less than 50 samples. The statistical relevance can therefore be distinguished according
to the marker used in the following plots: a cube for (relatively) high relevance, a big circle for medium relevance, a small circle for low relevance and a cross for no sample obtained.

Two plots were generated. The first one shows the average performance ratio per datapoint. This represents the ratio of the final number of terms given
by quizx to the final number of terms given by our algorithm. Therefore, a high performance ratio is favorable. The second plot shows the percentage
of improvements relative to the total number of relevant samples collected.

The source code for the benchmark can be found in \cite{Tix3DevRand_multi_ctrl_toff_dense_qcirc_sim}.

\newpage

The results can be seen in \reffig{benchmark_average_performance_ratio} and in \reffig{benchmark_percentage_of_improvements}.

\begin{figure*}[!ht]
	\includegraphics{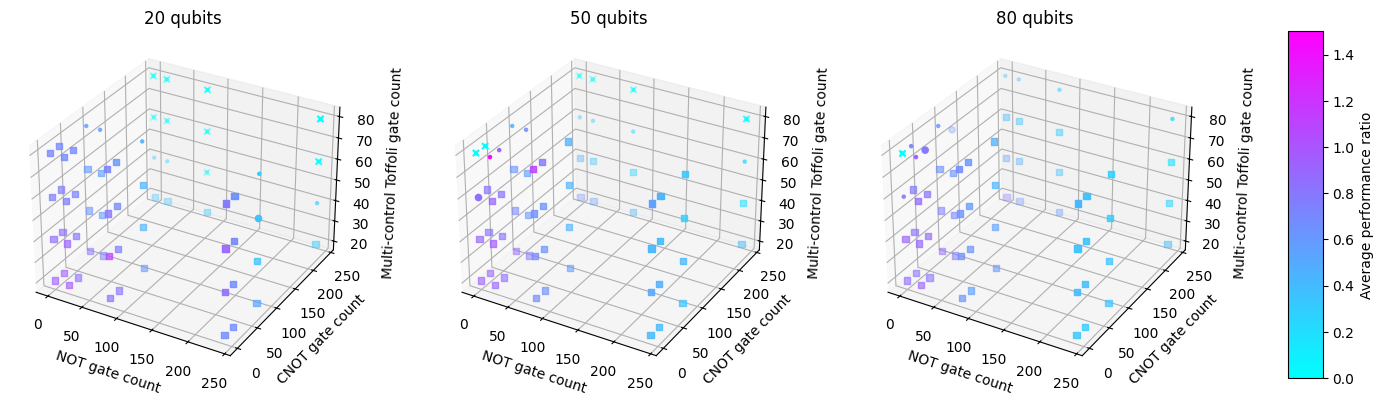}
	\caption[Average performance ratio]{Benchmark results regarding average performance ratio, generated using \cite{Tix3DevRand_multi_ctrl_toff_dense_qcirc_sim}.}
	\labfig{benchmark_average_performance_ratio}
\end{figure*}

\begin{figure*}[!ht]
	\includegraphics{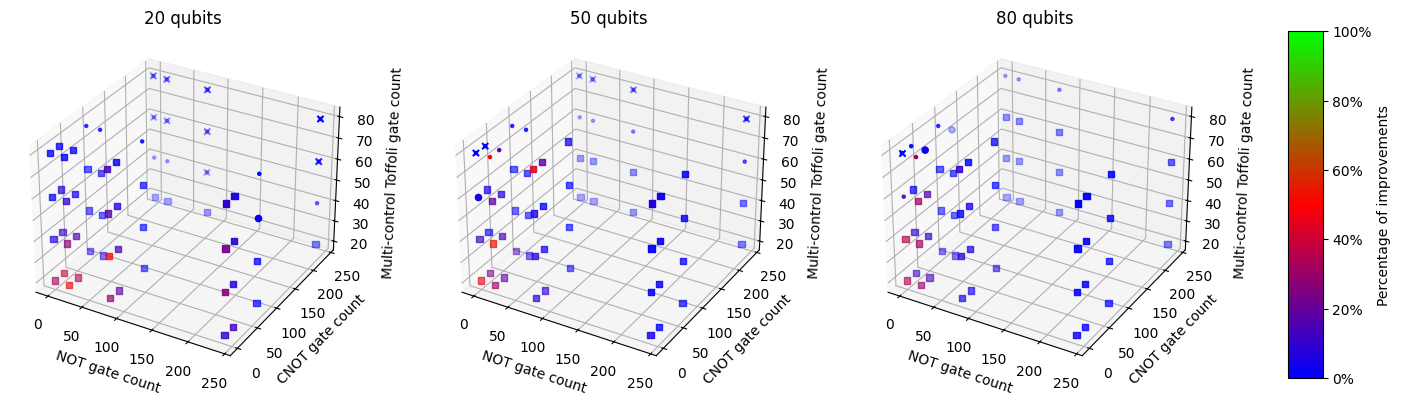}
	\caption[Percentage of improvements]{Benchmark results regarding percentage of improvements, generated using \cite{Tix3DevRand_multi_ctrl_toff_dense_qcirc_sim}.}
	\labfig{benchmark_percentage_of_improvements}
\end{figure*}

Examining \reffig{benchmark_average_performance_ratio}, we can see that overall, there are only a few occurences where the average performance ratio is greater
than one. Nevertheless, we can see that according to \reffig{benchmark_percentage_of_improvements}, our algorithm is \textit{capable} of delivering an advantage
across various configurations, at times yielding improvements for up to $50\%$ of the samples.

By looking at the vertical edges of the graphs\sidenote{We are referring to the configurations that have either $0$ or $240$ NOT gates, and $240$ or $0$ CNOT gates, respectively.},
one can see that a high NOT count combined with a low CNOT count yields on average a better performance than a high CNOT count combined with a low NOT count.
This aligns with our hypothesis, as partial simplifications can handle NOT gates relatively well, whereas only a full simplifcation strategy could reduce the added complexity
of CNOT gates effectively.

Overall, low counts of NOT and CNOT gates tend to yield the best results. Furthermore, increasing the number of multi-control Toffoli gates does not strongly affect
the average performance ratio. This is expected, since to our algorithm, increasing this number simply means increasing the degree of the potential master nodes, assuming
obstructions are handled sufficiently well by the simplification strategy.

Our algorithm was able to yield improvements for particularly many configurations in the 20 qubit case. This can be seen as further support for our hypothesis, as
more qubits allow for more simplifications to apply.

Finally, \reffig{benchmark_percentage_of_improvements} shows a very similar distribution to \reffig{benchmark_average_performance_ratio}, which demonstrates that there
were no significant outliers affecting the average performance ratio.

\newpage

To see what it would mean to apply a full simplification strategy, we first have to clarify what we mean by a full simplification. PyZX offers two main simplification routines,
the first one being \textit{clifford\_simp} and the second one being \textit{full\_reduce}. The latter applies \textit{clifford\_simp} alongside certain routines that make use
of \textit{gadgetization} \cite{kissingerPyZXLargeScale2020}. As the quizx implementation we are comparing against utilizes \textit{clifford\_simp} and furthermore
\textit{full\_reduce} does not work well with H-boxes that are used to represent the star edges, we will use \textit{clifford\_simp} for the following demonstration.

Consider the following randomly generated 17-qubit scalar diagram in \reffig{cliff_simp_demo1}, consisting of 25 NOT and CNOT gates, and six multi-control Toffoli gates.
\begin{figure*}[!ht]
	\includegraphics{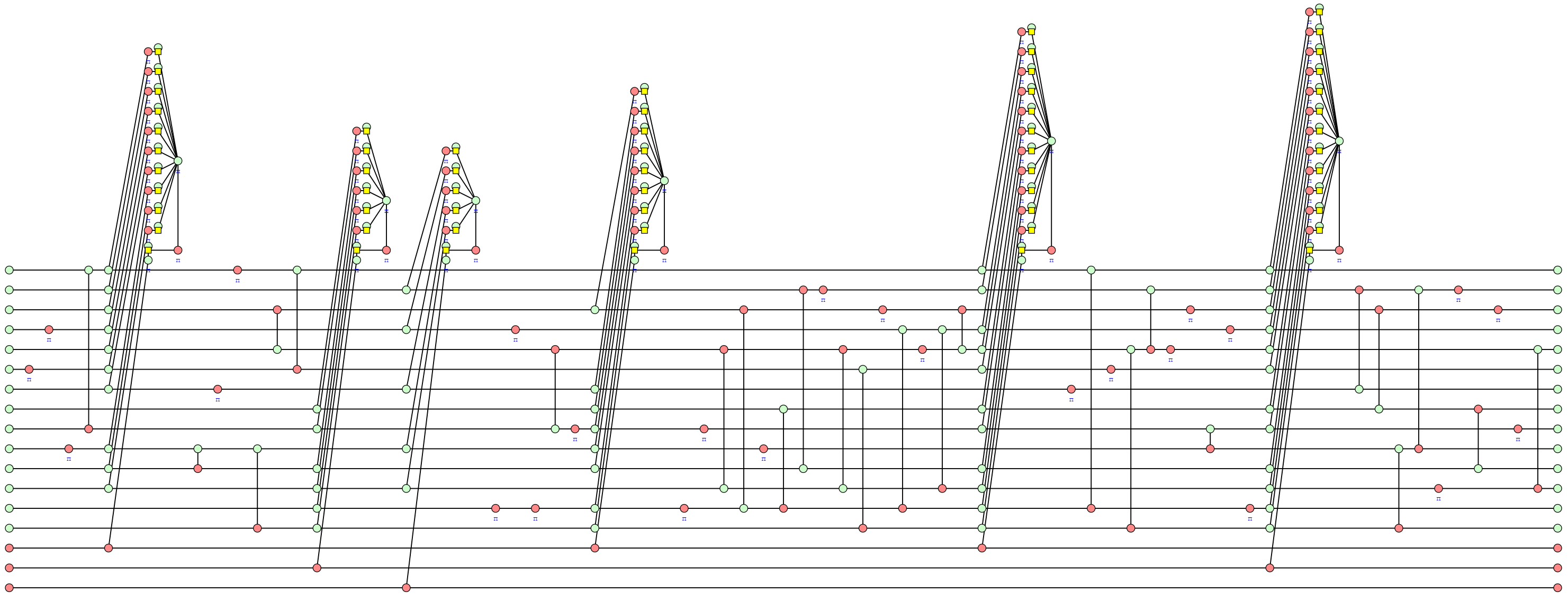}
	\caption[Initial 17-qubit scalar diagram]{The initial 17-qubit scalar diagram before applying any simplification strategy, generated
    in \cite{Tix3DevRand_multi_ctrl_toff_dense_qcirc_sim}.}
	\labfig{cliff_simp_demo1}
\end{figure*}

After applying our partial simplification strategy employed thus far, which includes an unsuccessful attempt at creating two stacks, we get the diagram
in \reffig{cliff_simp_demo2}.
\begin{figure*}[!ht]
	\includegraphics{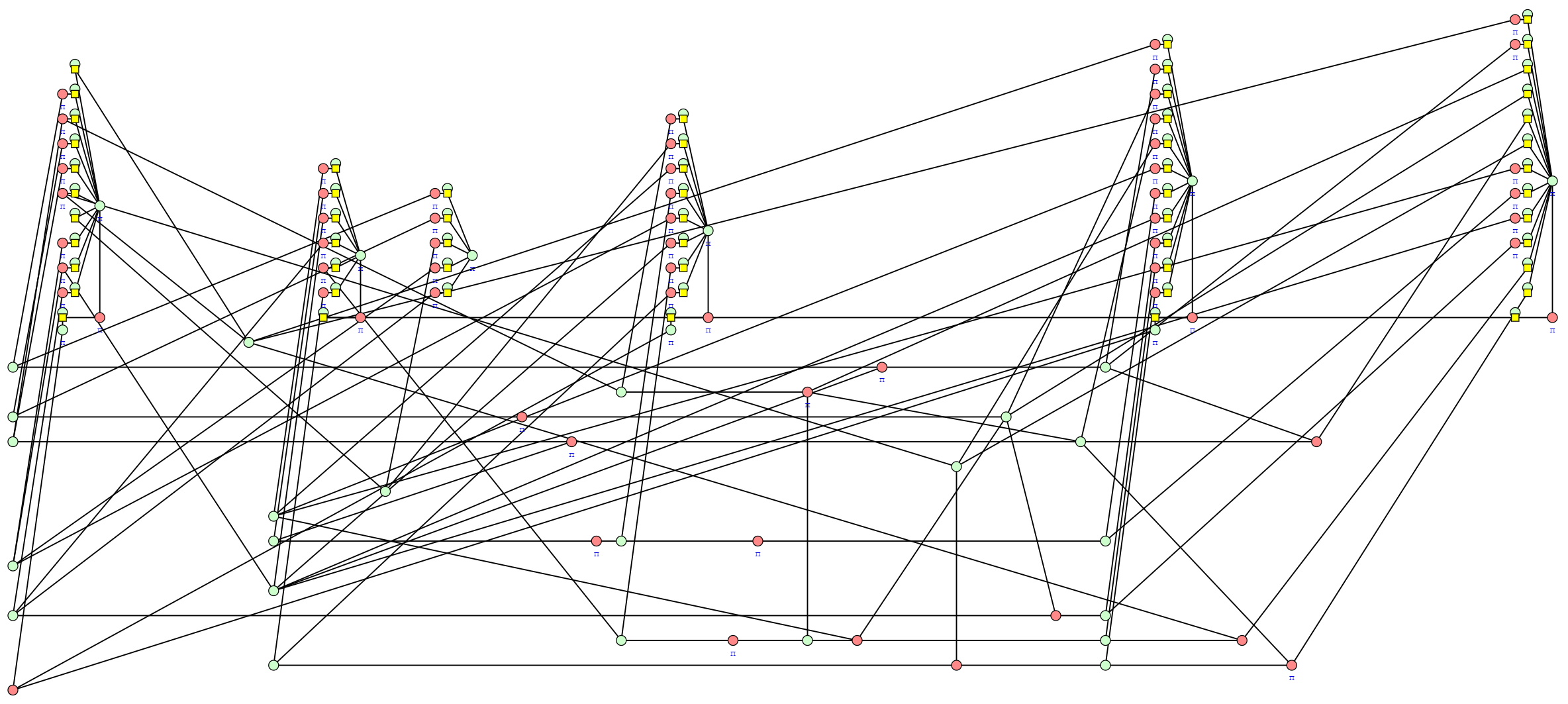}
	\caption[17-qubit scalar diagram after partial simplification]{The 17-qubit scalar diagram after applying our partial simplification strategy, generated
    in \cite{Tix3DevRand_multi_ctrl_toff_dense_qcirc_sim}.}
	\labfig{cliff_simp_demo2}
\end{figure*}

\newpage

Comparing this to the diagram obtained in \reffig{cliff_simp_demo3} after applying \textit{clifford\_simp}, we can observe significant differences.
\begin{figure*}[!ht]
	\includegraphics{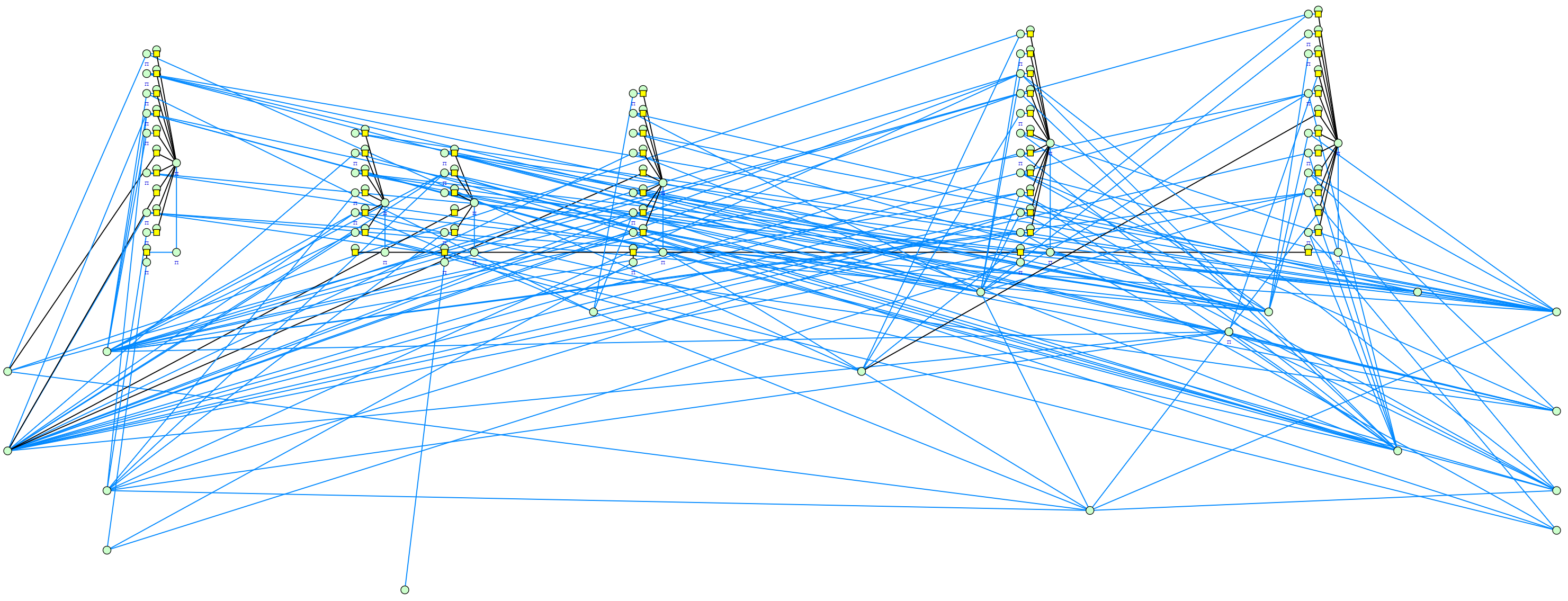}
	\caption[17-qubit scalar diagram after full simplification]{The 17-qubit scalar diagram after applying the full simplification strategy, generated
    in \cite{Tix3DevRand_multi_ctrl_toff_dense_qcirc_sim}.}
	\labfig{cliff_simp_demo3}
\end{figure*}

Most notably, the new diagram contains a lot more Hadamard edges. Additionally, there is a lot more clustering. In other words, this means that there are less spiders,
but higher vertex degrees. This is very important, as it allows a weighting algorithm to make a better prediction due to a bigger coverage of the diagram,
whilst not having to change the search depth. Nevertheless, this might require clarification, as it is easy to think that our partial simplification performs
very poorly. Concretely, our simplification strategy was made for a class of quantum circuits capable of searching for causal configurations, for which we know that
CNOT gates can practically never occur. This is why the diagram in \reffig{cliff_simp_demo2} is only poorly simplified. We added a lot of CNOT gates in order to
artificially increase the complexity of the circuit\sidenote{Instead of seeing the addition of CNOT gates as a tool to change the complexity without having to
change the number of qubits, we could see this as changing the class we initially set out to study, in order for it to now include different gates.}.

Given this evidence, we conclude that changing our algorithm such that it is capable of obtaining reliable improvements in the final number of terms would
require the use of \textit{clifford\_simp}.
This would then require changes to the weighting algorithm, as it was not intended to be used with Hadamard edges. A similar analysis to the one conducted
in \refsec{studied_approach} would need to reveal which patterns (that may include Hadamard edges) can be attributed to which weight. Ideally, this would be
implemented in our modified quizx version, in order to make the obtained speed-ups more practical and easier to compare. Additionally, this would also allow
the embedding of star decompositions and potentially even state decompositions into the weighting algorithm\sidenote{This would again require a detailed analysis of
post-simplification patterns.}. Finally, it is worth mentioning that it is not clear how big the actual improvement would be if we were to use such a refined algorithm.
As of right now, the improvements are significantly less than the ones observed in \cite{sutcliffeProcedurallyOptimisedZXDiagram2024}. This could be a general property
of our algorithm, or it might be because the benchmarks were conducted in a certain regime that is not suitable to demonstrate bigger improvements, similar to how in
\reffig{procedurally_results1} and \reffig{procedurally_results2} the obtained improvements for low T-counts is relatively small. In either case, it would
be interesting to investigate the effects of increasing the search depth. Since in practice the cost of searching is not negligible, it might be an option
to start with a big search depth and then only use more shallow searches afterwards. This might prevent going down a suboptimal branch in the \textit{tree}
of possible vertex decompositions.

\pagelayout{wide} 
\addpart{Part IV: Conclusion}\labpart{part_iv}
\pagelayout{margin} 

\setchapterpreamble[u]{\margintoc}
\chapter{Conclusion}
\labch{conclusion}

\section{Summary}

In this thesis, we have explored two primary research directions. The first direction concerns the state decompositions of non-stabilizer states made up
of star edges, a topic previously studied in the literature \cite{kochSpeedyContractionZX2023a} \cite{laakkonenGraphicalStabilizerDecompositions2022}.
Using simulated annealing, we were able to find five novel decompositions: three 5-to-6 decompositions and two 4-to-5 decompositions, with phases
constrained to $0$, $\frac{\pi}{2}$ or $-\frac{\pi}{2}$. We tested the software \cite{laakkonenTuomas56Cliffs2024} on an AWS EC2 C6a instance \cite{AmazonEC2C6a}
and observed that results were either found relatively quickly, or not at all. Therefore, we could not take advantage of the high-performance computer.

The second direction focused on dynamic decompositions and their application to weighting algorithms, drawing inspiration from
\cite{sutcliffeProcedurallyOptimisedZXDiagram2024}.
We demonstrated that their proposed idea of CNOT-grouping for the reduction of T-gates cannot be transferred to the reduction of star edges. Nevertheless, we
described a new weighting algorithm designed for multi-control Toffoli gate dense quantum circuits. This algorithm was tailored to a specific class of quantum
circuits, which is used to find causal configurations of multiloop Feynman diagrams. Thus, the algorithm only needs to take special care of NOT-gates, which
is accomplished by using a "stack" representation, that we introduced whilst studying the CNOT-grouping technique.

Our implementation successfully simulated the two tested topologies. A comparison to the implementation from \cite{kochSpeedyContractionZX2023a} revealed that
our algorithm could only provide an improvement in the final number of terms for the simpler topology. We hypothesized that this is due to our partial simplification
strategy, whereas they used a full simplification strategy.

In order to test our hypothesis, we generated 192 different types of random quantum circuits and sampled each type 50 times. The resulting data seems to
support our hypothesis.

\section{Outlook}

The results from working with a high-performance computer seem to suggest that future work on methods to find novel state decompositions
would require the use of a different approach than simulated annealing.

Similarly, the benchmarks of our weighting algorithm provide convincing evidence that future work should primarily consider using full
simplification strategies. Additionally, it would be interesting to see the effects of increasing the search depth, at least in the beginning of the procedure.

Whilst finishing this project, we made a final discovery: It turns out that the quizx implementation from \cite{kochSpeedyContractionZX2023a} can produce
reliable improvements when not using star decompositions. In other words, using star decompositions makes sense for certain quantum circuits such as
those used for barren plateau detection, but might be disadvantageous for the tested multi-control Toffoli gate dense quantum circuits. However, it might be
possible that a more clever algorithm could embed star decompositions such that no disadvantages are created, and potentially even create advantages by
incorporating it into a weighting algorithm. More information can be found in \cite{Tix3DevUnexpected_improv_of_speedy_algo}.


\backmatter 
\setchapterstyle{plain} 



\defbibnote{bibnote}{Here are the references in citation order.\par\bigskip} 
\printbibliography[heading=bibintoc, title=Bibliography, prenote=bibnote] 

@book{10.1093/acprof:oso/9780199699322.001.0001,
  title = {Quantum Field Theory for the Gifted Amateur},
  author = {Lancaster, Tom and Blundell, Stephen J.},
  date = {2014-04},
  publisher = {Oxford University Press},
  doi = {10.1093/acprof:oso/9780199699322.001.0001},
  isbn = {978-0-19-969932-2}
}

@online{182BriefSummary2017,
  title = {18.2: {{Brief}} Summary of the Origins of Quantum Theory},
  shorttitle = {18.2},
  date = {2017-11-11T17:04:58Z},
  url = {https://phys.libretexts.org/Bookshelves/Classical_Mechanics/Variational_Principles_in_Classical_Mechanics_(Cline)/18%3A_The_Transition_to_Quantum_Physics/18.02%3A_Brief_summary_of_the_origins_of_quantum_theory},
  urldate = {2024-09-15},
  langid = {english},
  organization = {Physics LibreTexts}
}

@online{95CauchyResidue2017,
  title = {9.5: {{Cauchy Residue Theorem}}},
  shorttitle = {9.5},
  date = {2017-09-05T19:33:31Z},
  url = {https://math.libretexts.org/Bookshelves/Analysis/Complex_Variables_with_Applications_(Orloff)/09%3A_Residue_Theorem/9.05%3A_Cauchy_Residue_Theorem},
  urldate = {2024-12-01},
  langid = {english},
  organization = {Mathematics LibreTexts}
}

@online{AmazonEC2C6a,
  title = {Amazon {{EC2 C6a Instances}} - {{Amazon Web Services}}},
  url = {https://aws.amazon.com/ec2/instance-types/c6a/},
  urldate = {2024-11-17},
  langid = {american},
  organization = {Amazon Web Services, Inc.}
}

@article{backensZHCompleteGraphical2019,
  title = {{{ZH}}: {{A Complete Graphical Calculus}} for {{Quantum Computations Involving Classical Non-linearity}}},
  shorttitle = {{{ZH}}},
  author = {Backens, Miriam and Kissinger, Aleks},
  date = {2019-01-31},
  journaltitle = {Electronic Proceedings in Theoretical Computer Science},
  shortjournal = {Electron. Proc. Theor. Comput. Sci.},
  volume = {287},
  eprint = {1805.02175},
  eprinttype = {arXiv},
  eprintclass = {quant-ph},
  pages = {23--42},
  issn = {2075-2180},
  doi = {10.4204/EPTCS.287.2}
}

@article{backensZXcalculusCompleteStabilizer2014,
  title = {The {{ZX-calculus}} Is Complete for Stabilizer Quantum Mechanics},
  author = {Backens, Miriam},
  date = {2014-09},
  journaltitle = {New Journal of Physics},
  shortjournal = {New J. Phys.},
  volume = {16},
  number = {9},
  pages = {093021},
  publisher = {IOP Publishing},
  issn = {1367-2630},
  doi = {10.1088/1367-2630/16/9/093021},
  langid = {english}
}

@online{beaudrapPauliFusionComputational2019,
  title = {Pauli {{Fusion}}: A Computational Model to Realise Quantum Transformations from {{ZX}} Terms},
  shorttitle = {Pauli {{Fusion}}},
  author = {family=Beaudrap, given=Niel, prefix=de, useprefix=false and Duncan, Ross and Horsman, Dominic and Perdrix, Simon},
  date = {2019-04-29},
  eprint = {1904.12817},
  eprinttype = {arXiv},
  eprintclass = {quant-ph},
  doi = {10.48550/arXiv.1904.12817},
  pubstate = {prepublished},
  version = {1}
}

@online{beaudrapZXCalculusLanguage2017,
  title = {The {{ZX}} Calculus Is a Language for Surface Code Lattice Surgery},
  author = {family=Beaudrap, given=Niel, prefix=de, useprefix=false and Horsman, Dominic},
  date = {2017-04-27},
  eprint = {1704.08670},
  eprinttype = {arXiv},
  eprintclass = {quant-ph},
  doi = {10.48550/arXiv.1704.08670},
  pubstate = {prepublished},
  version = {1}
}

@inreference{BlochSphere2024,
  title = {Bloch Sphere},
  booktitle = {Wikipedia},
  date = {2024-07-16T22:40:36Z},
  url = {https://en.wikipedia.org/w/index.php?title=Bloch_sphere&oldid=1234939005},
  urldate = {2024-09-22},
  langid = {english}
}

@online{bobadillaOffshellJacobiCurrents,
  title = {Off-Shell {{Jacobi}} Currents within the Loop-Tree Duality},
  author = {Bobadilla, William J Torres},
  url = {https://indico.cern.ch/event/577856/contributions/3420312/attachments/1878617/3094415/EPS-2019_ckd.pdf},
  langid = {english}
}

@article{bravyiTradingClassicalQuantum2016,
  title = {Trading Classical and Quantum Computational Resources},
  author = {Bravyi, Sergey and Smith, Graeme and Smolin, John},
  date = {2016-06-29},
  journaltitle = {Physical Review X},
  shortjournal = {Phys. Rev. X},
  volume = {6},
  number = {2},
  eprint = {1506.01396},
  eprinttype = {arXiv},
  eprintclass = {quant-ph},
  pages = {021043},
  issn = {2160-3308},
  doi = {10.1103/PhysRevX.6.021043}
}

@article{caretteQuantumAlgorithmsOracles2021a,
  title = {Quantum {{Algorithms}} and {{Oracles}} with the {{Scalable ZX-calculus}}},
  author = {Carette, Titouan and D'Anello, Yohann and Perdrix, Simon},
  date = {2021-09-18},
  journaltitle = {Electronic Proceedings in Theoretical Computer Science},
  shortjournal = {Electron. Proc. Theor. Comput. Sci.},
  volume = {343},
  pages = {193--209},
  issn = {2075-2180},
  doi = {10.4204/EPTCS.343.10},
  langid = {english}
}

@inproceedings{caretteRecipeQuantumGraphical2020,
  title = {A {{Recipe}} for {{Quantum Graphical Languages}}},
  booktitle = {47th {{International Colloquium}} on {{Automata}}, {{Languages}}, and {{Programming}} ({{ICALP}} 2020)},
  author = {Carette, Titouan and Jeandel, Emmanuel},
  editor = {Czumaj, Artur and Dawar, Anuj and Merelli, Emanuela},
  date = {2020},
  series = {Leibniz {{International Proceedings}} in {{Informatics}} ({{LIPIcs}})},
  volume = {168},
  pages = {118:1--118:17},
  publisher = {Schloss Dagstuhl – Leibniz-Zentrum für Informatik},
  location = {Dagstuhl, Germany},
  issn = {1868-8969},
  doi = {10.4230/LIPIcs.ICALP.2020.118},
  isbn = {978-3-95977-138-2}
}

@online{chancellorGraphicalStructuresDesign2018,
  title = {Graphical {{Structures}} for {{Design}} and {{Verification}} of {{Quantum Error Correction}}},
  author = {Chancellor, Nicholas and Kissinger, Aleks and Roffe, Joschka and Zohren, Stefan and Horsman, Dominic},
  date = {2018-01-12},
  eprint = {1611.08012},
  eprinttype = {arXiv},
  eprintclass = {quant-ph},
  doi = {10.48550/arXiv.1611.08012},
  pubstate = {prepublished},
  version = {3}
}

@online{codsiCuttingEdgeGraphicalStabiliser,
  title = {Cutting-{{Edge Graphical Stabiliser Decompositions}} for {{Classical Simulation}} of {{Quantum Circuits}}},
  author = {Codsi, Julien},
  url = {https://www.maths.ox.ac.uk/system/files/inline-files/J%20Codsi%2021-22.pdf},
  langid = {english}
}

@online{coeckeFoundationsNearTermQuantum2020,
  title = {Foundations for {{Near-Term Quantum Natural Language Processing}}},
  author = {Coecke, Bob and family=Felice, given=Giovanni, prefix=de, useprefix=false and Meichanetzidis, Konstantinos and Toumi, Alexis},
  date = {2020-12-07},
  eprint = {2012.03755},
  eprinttype = {arXiv},
  eprintclass = {quant-ph},
  doi = {10.48550/arXiv.2012.03755},
  langid = {english},
  pubstate = {prepublished}
}

@incollection{coeckeInteractingQuantumObservables2008,
  title = {Interacting {{Quantum Observables}}},
  booktitle = {Automata, {{Languages}} and {{Programming}}},
  author = {Coecke, Bob and Duncan, Ross},
  editor = {Aceto, Luca and Damgård, Ivan and Goldberg, Leslie Ann and Halldórsson, Magnús M. and Ingólfsdóttir, Anna and Walukiewicz, Igor},
  date = {2008},
  volume = {5126},
  pages = {298--310},
  publisher = {Springer Berlin Heidelberg},
  location = {Berlin, Heidelberg},
  issn = {0302-9743, 1611-3349},
  doi = {10.1007/978-3-540-70583-3_25},
  isbn = {978-3-540-70582-6 978-3-540-70583-3},
  langid = {english}
}

@article{cowtanPhaseGadgetSynthesis2020,
  title = {Phase {{Gadget Synthesis}} for {{Shallow Circuits}}},
  author = {Cowtan, Alexander and Dilkes, Silas and Duncan, Ross and Simmons, Will and Sivarajah, Seyon},
  date = {2020-05-01},
  journaltitle = {Electronic Proceedings in Theoretical Computer Science},
  shortjournal = {Electron. Proc. Theor. Comput. Sci.},
  volume = {318},
  eprint = {1906.01734},
  eprinttype = {arXiv},
  eprintclass = {quant-ph},
  pages = {213--228},
  issn = {2075-2180},
  doi = {10.4204/EPTCS.318.13}
}

@article{debeaudrapTechniquesReducePi2020,
  title = {Techniques to {{Reduce Pi}}/4-{{Parity-Phase Circuits}}, {{Motivated}} by the {{ZX Calculus}}},
  author = {De Beaudrap, Niel and Bian, Xiaoning and Wang, Quanlong},
  date = {2020-05-01},
  journaltitle = {Electronic Proceedings in Theoretical Computer Science},
  shortjournal = {Electron. Proc. Theor. Comput. Sci.},
  volume = {318},
  pages = {131--149},
  issn = {2075-2180},
  doi = {10.4204/EPTCS.318.9},
  langid = {english}
}

@article{debeaudrapTensorNetworkRewriting2021,
  title = {Tensor {{Network Rewriting Strategies}} for {{Satisfiability}} and {{Counting}}},
  author = {De Beaudrap, Niel and Kissinger, Aleks and Meichanetzidis, Konstantinos},
  date = {2021-09-06},
  journaltitle = {Electronic Proceedings in Theoretical Computer Science},
  shortjournal = {Electron. Proc. Theor. Comput. Sci.},
  volume = {340},
  pages = {46--59},
  issn = {2075-2180},
  doi = {10.4204/EPTCS.340.3},
  langid = {english}
}

@inreference{DiracDeltaFunction2024,
  title = {Dirac Delta Function},
  booktitle = {Wikipedia},
  date = {2024-10-13T11:47:19Z},
  url = {https://en.wikipedia.org/w/index.php?title=Dirac_delta_function&oldid=1250931522#Definitions},
  urldate = {2024-10-16},
  langid = {english}
}

@book{duncanGraphsStatesNecessity2009,
  title = {Graphs {{States}} and the Necessity of {{Euler Decomposition}}},
  author = {Duncan, Ross and Perdrix, Simon},
  date = {2009},
  volume = {5635},
  eprint = {0902.0500},
  eprinttype = {arXiv},
  eprintclass = {quant-ph},
  doi = {10.1007/978-3-642-03073-4}
}

@article{duncanGraphtheoreticSimplificationQuantum2020,
  title = {Graph-Theoretic {{Simplification}} of {{Quantum Circuits}} with the {{ZX-calculus}}},
  author = {Duncan, Ross and Kissinger, Aleks and Perdrix, Simon and family=Wetering, given=John, prefix=van de, useprefix=false},
  date = {2020-06-04},
  journaltitle = {Quantum},
  volume = {4},
  pages = {279},
  publisher = {Verein zur Förderung des Open Access Publizierens in den Quantenwissenschaften},
  doi = {10.22331/q-2020-06-04-279},
  langid = {british}
}

@inproceedings{duncanRewritingMeasurementBasedQuantum2010,
  title = {Rewriting {{Measurement-Based Quantum Computations}} with {{Generalised Flow}}},
  booktitle = {Automata, {{Languages}} and {{Programming}}},
  author = {Duncan, Ross and Perdrix, Simon},
  editor = {Abramsky, Samson and Gavoille, Cyril and Kirchner, Claude and Meyer auf der Heide, Friedhelm and Spirakis, Paul G.},
  date = {2010},
  pages = {285--296},
  publisher = {Springer},
  location = {Berlin, Heidelberg},
  doi = {10.1007/978-3-642-14162-1_24},
  isbn = {978-3-642-14162-1},
  langid = {english}
}

@article{duncanVerifyingSteaneCode2014,
  title = {Verifying the {{Steane}} Code with {{Quantomatic}}},
  author = {Duncan, Ross and Lucas, Maxime},
  date = {2014-12-27},
  journaltitle = {Electronic Proceedings in Theoretical Computer Science},
  shortjournal = {Electron. Proc. Theor. Comput. Sci.},
  volume = {171},
  eprint = {1306.4532},
  eprinttype = {arXiv},
  eprintclass = {quant-ph},
  pages = {33--49},
  issn = {2075-2180},
  doi = {10.4204/EPTCS.171.4}
}

@article{eastAKLTStatesZXDiagramsDiagrammatic2022a,
  title = {{{AKLT-States}} as {{ZX-Diagrams}}: {{Diagrammatic Reasoning}} for {{Quantum States}}},
  shorttitle = {{{AKLT-States}} as {{ZX-Diagrams}}},
  author = {East, Richard D.P. and family=Wetering, given=John, prefix=van de, useprefix=true and Chancellor, Nicholas and Grushin, Adolfo G.},
  date = {2022-01-04},
  journaltitle = {PRX Quantum},
  shortjournal = {PRX Quantum},
  volume = {3},
  number = {1},
  pages = {010302},
  publisher = {American Physical Society},
  doi = {10.1103/PRXQuantum.3.010302}
}

@online{eastSpinnetworksZXcalculus2022a,
  title = {Spin-Networks in the {{ZX-calculus}}},
  author = {East, Richard D. P. and Martin-Dussaud, Pierre and family=Wetering, given=John Van, prefix=de, useprefix=false},
  date = {2022-11-18},
  eprint = {2111.03114},
  eprinttype = {arXiv},
  eprintclass = {quant-ph},
  url = {http://arxiv.org/abs/2111.03114},
  urldate = {2024-10-18},
  langid = {english},
  pubstate = {prepublished}
}

@online{FeynmanDsonVanjpgJPEGImage,
  title = {{{FeynmanDsonVan}}.Jpg ({{JPEG Image}}, 800~×~529 Pixels)},
  url = {http://woodahl.physics.iupui.edu/Astro105/FeynmanDsonVan.jpg},
  urldate = {2024-10-16}
}

@article{feynmanSimulatingPhysicsComputers1982,
  title = {Simulating Physics with Computers},
  author = {Feynman, Richard P.},
  date = {1982-06-01},
  journaltitle = {International Journal of Theoretical Physics},
  shortjournal = {Int J Theor Phys},
  volume = {21},
  number = {6},
  pages = {467--488},
  issn = {1572-9575},
  doi = {10.1007/BF02650179},
  langid = {english}
}

@article{feynmanSpaceTimeApproachQuantum1949,
  title = {Space-{{Time Approach}} to {{Quantum Electrodynamics}}},
  author = {Feynman, R. P.},
  date = {1949-09-15},
  journaltitle = {Physical Review},
  shortjournal = {Phys. Rev.},
  volume = {76},
  number = {6},
  pages = {769--789},
  issn = {0031-899X},
  doi = {10.1103/PhysRev.76.769},
  langid = {english}
}

@online{FigureFourDifferent,
  title = {Figure 2. {{The}} Four Different Regions of the {{Penrose}} Diagram.},
  url = {https://www.researchgate.net/figure/The-four-different-regions-of-the-Penrose-diagram_fig2_346519045},
  urldate = {2024-10-16},
  langid = {english},
  organization = {ResearchGate}
}

@online{garciaEfficientInnerproductAlgorithm2013,
  title = {Efficient {{Inner-product Algorithm}} for {{Stabilizer States}}},
  author = {Garcia, Hector J. and Markov, Igor L. and Cross, Andrew W.},
  date = {2013-08-07},
  eprint = {1210.6646},
  eprinttype = {arXiv},
  eprintclass = {cs},
  doi = {10.48550/arXiv.1210.6646},
  pubstate = {prepublished}
}

@article{garvieVerifyingSmallestInteresting2018,
  title = {Verifying the {{Smallest Interesting Colour Code}} with {{Quantomatic}}},
  author = {Garvie, Liam and Duncan, Ross},
  date = {2018-02-27},
  journaltitle = {Electronic Proceedings in Theoretical Computer Science},
  shortjournal = {Electron. Proc. Theor. Comput. Sci.},
  volume = {266},
  eprint = {1706.02717},
  eprinttype = {arXiv},
  eprintclass = {quant-ph},
  pages = {147--163},
  issn = {2075-2180},
  doi = {10.4204/EPTCS.266.10}
}

@online{gorardZXCalculusExtendedHypergraph2020,
  title = {{{ZX-Calculus}} and {{Extended Hypergraph Rewriting Systems I}}: {{A Multiway Approach}} to {{Categorical Quantum Information Theory}}},
  shorttitle = {{{ZX-Calculus}} and {{Extended Hypergraph Rewriting Systems I}}},
  author = {Gorard, Jonathan and Namuduri, Manojna and Arsiwalla, Xerxes D.},
  date = {2020-10-05},
  eprint = {2010.02752},
  eprinttype = {arXiv},
  eprintclass = {cs},
  doi = {10.48550/arXiv.2010.02752},
  pubstate = {prepublished}
}

@online{gottesmanHeisenbergRepresentationQuantum1998,
  title = {The {{Heisenberg Representation}} of {{Quantum Computers}}},
  author = {Gottesman, Daniel},
  date = {1998-07-01},
  eprint = {quant-ph/9807006},
  eprinttype = {arXiv},
  url = {http://arxiv.org/abs/quant-ph/9807006},
  urldate = {2024-05-04},
  langid = {english},
  pubstate = {prepublished}
}

@online{GraphicalRepresentationParticle,
  title = {A Graphical Representation of the ‘Particle in a Box’ as Solution Of...},
  url = {https://www.researchgate.net/figure/A-graphical-representation-of-the-particle-in-a-box-as-solution-of-the-time-independent_fig5_337259777},
  urldate = {2024-09-21},
  langid = {english},
  organization = {ResearchGate}
}

@inreference{GroversAlgorithm2024,
  title = {Grover's Algorithm},
  booktitle = {Wikipedia},
  date = {2024-09-19T15:53:02Z},
  url = {https://en.wikipedia.org/w/index.php?title=Grover%27s_algorithm&oldid=1246543089},
  urldate = {2024-09-29},
  langid = {english}
}

@online{hadzihasanovicDiagrammaticAxiomatisationQubit2015,
  title = {A {{Diagrammatic Axiomatisation}} for {{Qubit Entanglement}}},
  author = {Hadzihasanovic, Amar},
  date = {2015-01-28},
  eprint = {1501.07082},
  eprinttype = {arXiv},
  eprintclass = {cs},
  doi = {10.48550/arXiv.1501.07082},
  pubstate = {prepublished}
}

@inproceedings{hillebrandQuantumProtocolsInvolving2011,
  title = {Quantum {{Protocols}} Involving {{Multiparticle Entanglement}} and Their {{Representations}} in the Zx-Calculus.},
  author = {Hillebrand, A.},
  date = {2011},
  url = {https://www.semanticscholar.org/paper/Quantum-Protocols-involving-Multiparticle-and-their-Hillebrand/c5690bce73e51dbd21f12f869f7372b264ba622d},
  urldate = {2024-12-17}
}

@article{jeandelCompletenessZXCalculus2020a,
  title = {Completeness of the {{ZX-Calculus}}},
  author = {Jeandel, Emmanuel and Perdrix, Simon and Vilmart, Renaud},
  date = {2020-06-04},
  journaltitle = {Logical Methods in Computer Science},
  volume = {Volume 16, Issue 2},
  url = {https://lmcs.episciences.org/6532/pdf},
  langid = {english}
}

@article{kaiserPhysicsFeynmansDiagrams2005,
  title = {Physics and {{Feynman}}'s {{Diagrams}}},
  author = {Kaiser, David},
  date = {2005},
  langid = {english}
}

@article{kissingerClassicalSimulationQuantum2022,
  title = {Classical {{Simulation}} of {{Quantum Circuits}} with {{Partial}} and {{Graphical Stabiliser Decompositions}}},
  author = {Kissinger, Aleks and family=Wetering, given=John, prefix=van de, useprefix=true and Vilmart, Renaud},
  namea = {Le Gall, François and Morimae, Tomoyuki},
  nameatype = {collaborator},
  date = {2022},
  pages = {13 pages, 867977 bytes},
  publisher = {Schloss Dagstuhl – Leibniz-Zentrum für Informatik},
  issn = {1868-8969},
  doi = {10.4230/LIPICS.TQC.2022.5},
  isbn = {9783959772372},
  langid = {english}
}

@article{kissingerPyZXLargeScale2020,
  title = {{{PyZX}}: {{Large Scale Automated Diagrammatic Reasoning}}},
  shorttitle = {{{PyZX}}},
  author = {Kissinger, Aleks and family=Wetering, given=John, prefix=van de, useprefix=false},
  date = {2020-05-01},
  journaltitle = {Electronic Proceedings in Theoretical Computer Science},
  shortjournal = {Electron. Proc. Theor. Comput. Sci.},
  volume = {318},
  eprint = {1904.04735},
  eprinttype = {arXiv},
  eprintclass = {quant-ph},
  pages = {229--241},
  issn = {2075-2180},
  doi = {10.4204/EPTCS.318.14}
}

@article{kissingerReducingNumberNonClifford2020,
  title = {Reducing the Number of Non-{{Clifford}} Gates in Quantum Circuits},
  author = {Kissinger, Aleks and family=Wetering, given=John, prefix=van de, useprefix=true},
  date = {2020-08-11},
  journaltitle = {Physical Review A},
  shortjournal = {Phys. Rev. A},
  volume = {102},
  number = {2},
  pages = {022406},
  publisher = {American Physical Society},
  doi = {10.1103/PhysRevA.102.022406}
}

@article{kissingerSimulatingQuantumCircuits2022,
  title = {Simulating Quantum Circuits with {{ZX-calculus}} Reduced Stabiliser Decompositions},
  author = {Kissinger, Aleks and family=Wetering, given=John, prefix=van de, useprefix=true},
  date = {2022-10-01},
  journaltitle = {Quantum Science and Technology},
  shortjournal = {Quantum Sci. Technol.},
  volume = {7},
  number = {4},
  eprint = {2109.01076},
  eprinttype = {arXiv},
  eprintclass = {quant-ph},
  pages = {044001},
  issn = {2058-9565},
  doi = {10.1088/2058-9565/ac5d20}
}

@book{KissingerWetering2024Book,
  title = {Picturing Quantum Software: {{An}} Introduction to the {{ZX-calculus}} and Quantum Compilation},
  author = {Kissinger, Aleks and family=Wetering, given=John, prefix=van de, useprefix=true},
  date = {2024},
  publisher = {Preprint}
}

@online{kochSpeedyContractionZX2023a,
  title = {Speedy {{Contraction}} of {{ZX Diagrams}} with {{Triangles}} via {{Stabiliser Decompositions}}},
  author = {Koch, Mark and Yeung, Richie and Wang, Quanlong},
  date = {2023-07-04},
  eprint = {2307.01803},
  eprinttype = {arXiv},
  eprintclass = {quant-ph},
  url = {http://arxiv.org/abs/2307.01803},
  urldate = {2024-02-26},
  langid = {english},
  pubstate = {prepublished}
}

@thesis{laakkonenGraphicalStabilizerDecompositions2022,
  type = {phdthesis},
  title = {Graphical {{Stabilizer Decompositions For Counting Problems}}},
  author = {Laakkonen, Tuomas},
  date = {2022},
  institution = {University of Oxford},
  url = {https://www.cs.ox.ac.uk/people/aleks.kissinger/theses/laakkonen-thesis.pdf},
  urldate = {2024-04-08}
}

@software{laakkonenTuomas56Cliffs2024,
  title = {Tuomas56/Cliffs},
  author = {Laakkonen, Tuomas},
  date = {2024-04-29T13:44:42Z},
  origdate = {2022-05-12T16:37:40Z},
  url = {https://github.com/tuomas56/cliffs},
  urldate = {2024-11-17}
}

@inreference{ListQuantumField2024,
  title = {List of Quantum Field Theories},
  booktitle = {Wikipedia},
  date = {2024-08-11T11:26:47Z},
  url = {https://en.wikipedia.org/w/index.php?title=List_of_quantum_field_theories&oldid=1239763410},
  urldate = {2024-12-28},
  langid = {english}
}

@inreference{LogicGate2024,
  title = {Logic Gate},
  booktitle = {Wikipedia},
  date = {2024-10-15T18:05:23Z},
  url = {https://en.wikipedia.org/w/index.php?title=Logic_gate&oldid=1251345389},
  urldate = {2024-10-16},
  langid = {english}
}

@misc{Mathematica,
  title = {Mathematica, {{Version}} 14.1},
  author = {Inc., Wolfram Research},
  url = {https://www.wolfram.com/mathematica}
}

@article{mccleanBarrenPlateausQuantum2018a,
  title = {Barren Plateaus in Quantum Neural Network Training Landscapes},
  author = {McClean, Jarrod R. and Boixo, Sergio and Smelyanskiy, Vadim N. and Babbush, Ryan and Neven, Hartmut},
  date = {2018-11-16},
  journaltitle = {Nature Communications},
  shortjournal = {Nat Commun},
  volume = {9},
  number = {1},
  eprint = {1803.11173},
  eprinttype = {arXiv},
  eprintclass = {quant-ph},
  pages = {4812},
  issn = {2041-1723},
  doi = {10.1038/s41467-018-07090-4},
  langid = {english}
}

@software{mjsutcliffe99Mjsutcliffe99ProcOptCut2024,
  title = {Mjsutcliffe99/{{ProcOptCut}}},
  author = {{mjsutcliffe99}},
  date = {2024-09-24T15:21:51Z},
  origdate = {2024-02-22T11:30:46Z},
  url = {https://github.com/mjsutcliffe99/ProcOptCut},
  urldate = {2024-11-26}
}

@article{moyardIntroductionZXcalculus2023a,
  title = {Introduction to the {{ZX-calculus}}},
  author = {Moyard, Romain},
  date = {2023-06-06},
  journaltitle = {PennyLane Demos},
  publisher = {Xanadu},
  url = {https://pennylane.ai/qml/demos/tutorial_zx_calculus/},
  urldate = {2024-10-19},
  langid = {english}
}

@article{muussLinearCombinationsZXdiagramsa,
  title = {Linear Combinations of {{ZX-diagrams}} for Parameterized Quantum Circuits},
  author = {Muuss, Gina},
  langid = {english}
}

@book{nielsenQuantumComputationQuantum2012,
  title = {Quantum {{Computation}} and {{Quantum Information}}: 10th {{Anniversary Edition}}},
  shorttitle = {Quantum {{Computation}} and {{Quantum Information}}},
  author = {Nielsen, Michael A. and Chuang, Isaac L.},
  date = {2012-06-05},
  edition = {1},
  publisher = {Cambridge University Press},
  doi = {10.1017/CBO9780511976667},
  isbn = {978-1-107-00217-3 978-0-511-97666-7}
}

@online{Novel_star_state_decompositionsZxbarrenprivateResulttxt,
  title = {Novel\_star\_state\_decompositions/Zxbarren-Private/Result.Txt at Main · {{Tix3Dev}}/Novel\_star\_state\_decompositions},
  url = {https://github.com/Tix3Dev/novel_star_state_decompositions/blob/main/zxbarren-private/result.txt},
  urldate = {2024-12-10}
}

@inreference{OperatorPhysics2024,
  title = {Operator (Physics)},
  booktitle = {Wikipedia},
  date = {2024-07-21T12:55:39Z},
  url = {https://en.wikipedia.org/w/index.php?title=Operator_(physics)&oldid=1235832028},
  urldate = {2024-10-15},
  langid = {english}
}

@inreference{ParticleBox2024,
  title = {Particle in a Box},
  booktitle = {Wikipedia},
  date = {2024-07-03T09:01:49Z},
  url = {https://en.wikipedia.org/w/index.php?title=Particle_in_a_box&oldid=1232351657},
  urldate = {2024-09-19},
  langid = {english}
}

@article{pehamEquivalenceCheckingQuantum2022,
  title = {Equivalence {{Checking}} of {{Quantum Circuits With}} the {{ZX-Calculus}}},
  author = {Peham, Tom and Burgholzer, Lukas and Wille, Robert},
  date = {2022-09},
  journaltitle = {IEEE Journal on Emerging and Selected Topics in Circuits and Systems},
  shortjournal = {IEEE J. Emerg. Sel. Topics Circuits Syst.},
  volume = {12},
  number = {3},
  pages = {662--675},
  issn = {2156-3357, 2156-3365},
  doi = {10.1109/JETCAS.2022.3202204},
  langid = {english}
}

@inreference{PenroseDiagram2024,
  title = {Penrose Diagram},
  booktitle = {Wikipedia},
  date = {2024-08-30T14:20:07Z},
  url = {https://en.wikipedia.org/w/index.php?title=Penrose_diagram&oldid=1243101446},
  urldate = {2024-10-16},
  langid = {english}
}

@inreference{PenroseGraphicalNotation2024,
  title = {Penrose Graphical Notation},
  booktitle = {Wikipedia},
  date = {2024-09-09T06:39:28Z},
  url = {https://en.wikipedia.org/w/index.php?title=Penrose_graphical_notation&oldid=1244793027},
  urldate = {2024-10-16},
  langid = {english}
}

@article{pfaendlerAdvancementsQuantumComputing2024,
  title = {Advancements in {{Quantum Computing}}—{{Viewpoint}}: {{Building Adoption}} and {{Competency}} in {{Industry}}},
  shorttitle = {Advancements in {{Quantum Computing}}—{{Viewpoint}}},
  author = {Pfaendler, Sieglinde M. -L. and Konson, Konstantin and Greinert, Franziska},
  date = {2024-03-01},
  journaltitle = {Datenbank-Spektrum},
  shortjournal = {Datenbank Spektrum},
  volume = {24},
  number = {1},
  pages = {5--20},
  issn = {1610-1995},
  doi = {10.1007/s13222-024-00467-4},
  langid = {english}
}

@online{physicsRichardFeynmanSession2015,
  title = {Richard {{Feynman}} - {{Session IV}}},
  author = {family=Physics, given=American Institute, prefix=of, useprefix=false},
  date = {2015-01-27T09:50:26-05:00},
  url = {https://www.aip.org/history-programs/niels-bohr-library/oral-histories/5020-4},
  urldate = {2024-12-18},
  langid = {english}
}

@inproceedings{poorCompletenessArbitraryFinite2023b,
  title = {Completeness for Arbitrary Finite Dimensions of {{ZXW-calculus}}, a Unifying Calculus},
  booktitle = {2023 38th {{Annual ACM}}/{{IEEE Symposium}} on {{Logic}} in {{Computer Science}} ({{LICS}})},
  author = {Poór, Boldizsár and Wang, Quanlong and Shaikh, Razin A. and Yeh, Lia and Yeung, Richie and Coecke, Bob},
  date = {2023-06-26},
  eprint = {2302.12135},
  eprinttype = {arXiv},
  eprintclass = {quant-ph},
  pages = {1--14},
  doi = {10.1109/LICS56636.2023.10175672}
}

@online{poorUniqueNormalForm,
  title = {A Unique Normal Form for Prime-Dimensional Qudit {{Clifford ZX-calculus}}},
  author = {Poor, Boldizsar},
  url = {https://www.cs.ox.ac.uk/michael.benedikt/msctheses/2022/1059275_97382463_1.pdf},
  langid = {english}
}

@online{poorZXcalculusCompleteFiniteDimensional2024,
  title = {{{ZX-calculus}} Is {{Complete}} for {{Finite-Dimensional Hilbert Spaces}}},
  author = {Poór, Boldizsár and Shaikh, Razin A. and Wang, Quanlong},
  date = {2024-05-17},
  eprint = {2405.10896},
  eprinttype = {arXiv},
  eprintclass = {quant-ph},
  url = {http://arxiv.org/abs/2405.10896},
  urldate = {2024-10-23},
  langid = {english},
  pubstate = {prepublished}
}

@inreference{Postulate2024,
  title = {Postulate},
  booktitle = {Simple {{English Wikipedia}}, the Free Encyclopedia},
  date = {2024-01-06T18:31:55Z},
  url = {https://simple.wikipedia.org/w/index.php?title=Postulate&oldid=9289219},
  urldate = {2024-09-16},
  langid = {english}
}

@online{PythonVSRust,
  title = {Python {{VS Rust}} Benchmarks, {{Which}} Programming Language or Compiler Is Faster},
  url = {https://programming-language-benchmarks.vercel.app/python-vs-rust},
  urldate = {2024-12-06}
}

@misc{qiskit2024,
  title = {Quantum Computing with {{Qiskit}}},
  author = {Javadi-Abhari, Ali and Treinish, Matthew and Krsulich, Kevin and Wood, Christopher J. and Lishman, Jake and Gacon, Julien and Martiel, Simon and Nation, Paul D. and Bishop, Lev S. and Cross, Andrew W. and Johnson, Blake R. and Gambetta, Jay M.},
  date = {2024},
  eprint = {2405.08810},
  eprinttype = {arXiv},
  eprintclass = {quant-ph},
  doi = {10.48550/arXiv.2405.08810}
}

@inreference{QuantumDecoherence2024,
  title = {Quantum Decoherence},
  booktitle = {Wikipedia},
  date = {2024-07-09T00:37:04Z},
  url = {https://en.wikipedia.org/w/index.php?title=Quantum_decoherence&oldid=1233425853},
  urldate = {2024-09-22},
  langid = {english}
}

@inreference{QuantumFieldTheory2024,
  title = {Quantum Field Theory},
  booktitle = {Wikipedia},
  date = {2024-10-08T02:02:04Z},
  url = {https://en.wikipedia.org/w/index.php?title=Quantum_field_theory&oldid=1250024595},
  urldate = {2024-10-16},
  langid = {english}
}

@inreference{QuantumLogicGate2024,
  title = {Quantum Logic Gate},
  booktitle = {Wikipedia},
  date = {2024-08-16T17:26:42Z},
  url = {https://en.wikipedia.org/w/index.php?title=Quantum_logic_gate&oldid=1240670066},
  urldate = {2024-10-22},
  langid = {english}
}

@article{ramirez-uribeQuantumAlgorithmFeynman2022,
  title = {Quantum Algorithm for {{Feynman}} Loop Integrals},
  author = {Ramírez-Uribe, Selomit and Rentería-Olivo, Andrés E. and Rodrigo, Germán and Sborlini, German F. R. and Silva, Luiz Vale},
  date = {2022-05-16},
  journaltitle = {Journal of High Energy Physics},
  shortjournal = {J. High Energ. Phys.},
  volume = {2022},
  number = {5},
  eprint = {2105.08703},
  eprinttype = {arXiv},
  eprintclass = {hep-ph, physics:hep-th, physics:quant-ph},
  pages = {100},
  issn = {1029-8479},
  doi = {10.1007/JHEP05(2022)100}
}

@online{ramirez-uribeQuantumQueryingBased2024,
  title = {Quantum Querying Based on Multicontrolled {{Toffoli}} Gates for Causal {{Feynman}} Loop Configurations and Directed Acyclic Graphs},
  author = {Ramírez-Uribe, Selomit and Rentería-Olivo, Andrés E. and Rodrigo, Germán},
  date = {2024-04-04},
  eprint = {2404.03544},
  eprinttype = {arXiv},
  eprintclass = {hep-ph, physics:hep-th, physics:quant-ph},
  doi = {10.48550/arXiv.2404.03544},
  pubstate = {prepublished}
}

@article{rojoQuantumMechanics2LectureNotes2021,
  title = {{{QuantumMechanics2-LectureNotes}}},
  author = {Rojo, Juan},
  date = {2021},
  journaltitle = {Quantum Mechanics},
  langid = {english}
}

@inreference{SchrodingerEquation2024,
  title = {Schrödinger Equation},
  booktitle = {Wikipedia},
  date = {2024-08-10T09:07:08Z},
  url = {https://en.wikipedia.org/w/index.php?title=Schr%C3%B6dinger_equation&oldid=1239596294},
  urldate = {2024-09-17},
  langid = {english}
}

@online{seymourMeaningFeynmanDiagrams,
  title = {The {{Meaning}} of {{Feynman Diagrams}}},
  author = {Seymour, Michael H},
  url = {https://indico.cern.ch/event/421552/sessions/170249/attachments/884224/1242865/feynman.pdf},
  urldate = {2024-10-16},
  langid = {english}
}

@online{shaikhCategoricalSemanticsFeynman2022,
  title = {Categorical {{Semantics}} for {{Feynman Diagrams}}},
  author = {Shaikh, Razin A. and Gogioso, Stefano},
  date = {2022-05-01},
  eprint = {2205.00466},
  eprinttype = {arXiv},
  eprintclass = {quant-ph},
  doi = {10.48550/arXiv.2205.00466},
  langid = {english},
  pubstate = {prepublished}
}

@online{shaikhFockedupZXCalculus2024,
  title = {The {{Focked-up ZX Calculus}}: {{Picturing Continuous-Variable Quantum Computation}}},
  shorttitle = {The {{Focked-up ZX Calculus}}},
  author = {Shaikh, Razin A. and Yeh, Lia and Gogioso, Stefano},
  date = {2024-06-05},
  eprint = {2406.02905},
  eprinttype = {arXiv},
  eprintclass = {quant-ph},
  doi = {10.48550/arXiv.2406.02905},
  langid = {english},
  pubstate = {prepublished}
}

@article{shaikhHowSumExponentiate2023,
  title = {How to {{Sum}} and {{Exponentiate Hamiltonians}} in {{ZXW Calculus}}},
  author = {Shaikh, Razin A. and Wang, Quanlong and Yeung, Richie},
  date = {2023-11-16},
  journaltitle = {Electronic Proceedings in Theoretical Computer Science},
  shortjournal = {Electron. Proc. Theor. Comput. Sci.},
  volume = {394},
  eprint = {2212.04462},
  eprinttype = {arXiv},
  eprintclass = {quant-ph},
  pages = {236--261},
  issn = {2075-2180},
  doi = {10.4204/EPTCS.394.14}
}

@inreference{ShorsAlgorithm2024,
  title = {Shor's Algorithm},
  booktitle = {Wikipedia},
  date = {2024-09-12T23:04:49Z},
  url = {https://en.wikipedia.org/w/index.php?title=Shor%27s_algorithm&oldid=1245427578},
  urldate = {2024-09-29},
  langid = {english}
}

@inreference{SimulatedAnnealing2024,
  title = {Simulated Annealing},
  booktitle = {Wikipedia},
  date = {2024-09-14T14:47:05Z},
  url = {https://en.wikipedia.org/w/index.php?title=Simulated_annealing&oldid=1245689836},
  urldate = {2024-11-17},
  langid = {english}
}

@online{Solvay_conference_1927,
  title = {Solvay\_conference\_1927},
  url = {https://upload.wikimedia.org/wikipedia/commons/6/6e/Solvay_conference_1927.jpg},
  urldate = {2024-09-15}
}

@inreference{SolvayConference2024,
  title = {Solvay {{Conference}}},
  booktitle = {Wikipedia},
  date = {2024-08-08T19:57:09Z},
  url = {https://en.wikipedia.org/w/index.php?title=Solvay_Conference&oldid=1239352124},
  urldate = {2024-09-15},
  langid = {english}
}

@online{sutcliffeFastClassicalSimulation2024,
  title = {Fast Classical Simulation of Quantum Circuits via Parametric Rewriting in the {{ZX-calculus}}},
  author = {Sutcliffe, Matthew and Kissinger, Aleks},
  date = {2024-03-11},
  eprint = {2403.06777},
  eprinttype = {arXiv},
  eprintclass = {quant-ph},
  doi = {10.48550/arXiv.2403.06777},
  pubstate = {prepublished}
}

@online{sutcliffeProcedurallyOptimisedZXDiagram2024,
  title = {Procedurally {{Optimised ZX-Diagram Cutting}} for {{Efficient T-Decomposition}} in {{Classical Simulation}}},
  author = {Sutcliffe, Matthew and Kissinger, Aleks},
  date = {2024-03-16},
  eprint = {2403.10964},
  eprinttype = {arXiv},
  eprintclass = {quant-ph},
  url = {http://arxiv.org/abs/2403.10964},
  urldate = {2024-05-27},
  langid = {english},
  pubstate = {prepublished}
}

@inreference{TimeComplexity2024,
  title = {Time Complexity},
  booktitle = {Wikipedia},
  date = {2024-08-11T19:17:05Z},
  url = {https://en.wikipedia.org/w/index.php?title=Time_complexity&oldid=1239822001},
  urldate = {2024-09-29},
  langid = {english}
}

@online{Tix3DevFeynman_loop_diagram_qsim,
  title = {{{Tix3Dev}}/Feynman\_loop\_diagram\_qsim},
  url = {https://github.com/Tix3Dev/feynman_loop_diagram_qsim},
  urldate = {2024-11-30},
  langid = {english},
  organization = {GitHub}
}

@online{Tix3DevNovel_star_state_decompositions,
  title = {{{Tix3Dev}}/Novel\_star\_state\_decompositions},
  url = {https://github.com/Tix3Dev/novel_star_state_decompositions},
  urldate = {2024-11-17}
}

@online{Tix3DevRand_multi_ctrl_toff_dense_qcirc_sim,
  title = {{{Tix3Dev}}/Rand\_multi\_ctrl\_toff\_dense\_qcirc\_sim},
  url = {https://github.com/Tix3Dev/rand_multi_ctrl_toff_dense_qcirc_sim},
  urldate = {2024-12-10}
}

@online{Tix3DevUnexpected_improv_of_speedy_algo,
  title = {{{Tix3Dev}}/Unexpected\_improv\_of\_speedy\_algo},
  url = {https://github.com/Tix3Dev/unexpected_improv_of_speedy_algo},
  urldate = {2024-12-29}
}

@article{torresbobadillaLottyLooptreeDuality2021,
  title = {Lotty – {{The}} Loop-Tree Duality Automation},
  author = {Torres~Bobadilla, William J.},
  date = {2021-06-14},
  journaltitle = {The European Physical Journal C},
  shortjournal = {Eur. Phys. J. C},
  volume = {81},
  number = {6},
  pages = {514},
  issn = {1434-6052},
  doi = {10.1140/epjc/s10052-021-09235-0},
  langid = {english}
}

@inreference{UltravioletCatastrophe2024,
  title = {Ultraviolet Catastrophe},
  booktitle = {Wikipedia},
  date = {2024-05-18T21:42:22Z},
  url = {https://en.wikipedia.org/w/index.php?title=Ultraviolet_catastrophe&oldid=1224518327},
  urldate = {2024-09-15},
  langid = {english}
}

@online{vandeweteringZXcalculusWorkingQuantum2020,
  title = {{{ZX-calculus}} for the Working Quantum Computer Scientist},
  author = {family=Wetering, given=John, prefix=van de, useprefix=true},
  date = {2020-12-27},
  eprint = {2012.13966},
  eprinttype = {arXiv},
  eprintclass = {quant-ph},
  doi = {10.48550/arXiv.2012.13966},
  pubstate = {prepublished}
}

@unpublished{wangDifferentiatingIntegratingZX2022,
  title = {Differentiating and {{Integrating ZX Diagrams}} with {{Applications}} to {{Quantum Machine Learning}}},
  author = {Wang, Quanlong and Yeung, Richie and Koch, Mark},
  date = {2022},
  eprint = {2201.13250v4},
  eprinttype = {arXiv}
}

@inreference{WaveFunction2024b,
  title = {Wave Function},
  booktitle = {Wikipedia},
  date = {2024-08-18T12:29:49Z},
  url = {https://en.wikipedia.org/w/index.php?title=Wave_function&oldid=1240945609},
  urldate = {2024-10-16},
  langid = {english}
}

@online{WolframPhysicsProject,
  title = {The {{Wolfram Physics Project}}: {{Finding}} the {{Fundamental Theory}} of {{Physics}}},
  shorttitle = {The {{Wolfram Physics Project}}},
  url = {https://www.wolframphysics.org/},
  urldate = {2024-12-17},
  langid = {english}
}

@online{xuHerculeanTaskClassical2023,
  title = {A {{Herculean}} Task: {{Classical}} Simulation of Quantum Computers},
  shorttitle = {A {{Herculean}} Task},
  author = {Xu, Xiaosi and Benjamin, Simon and Sun, Jinzhao and Yuan, Xiao and Zhang, Pan},
  date = {2023-02-17},
  eprint = {2302.08880},
  eprinttype = {arXiv},
  doi = {10.48550/arXiv.2302.08880},
  pubstate = {prepublished}
}

@article{zhaoAnalyzingBarrenPlateau2021,
  title = {Analyzing the Barren Plateau Phenomenon in Training Quantum Neural Networks with the {{ZX-calculus}}},
  author = {Zhao, Chen and Gao, Xiao-Shan},
  date = {2021-06-04},
  journaltitle = {Quantum},
  shortjournal = {Quantum},
  volume = {5},
  eprint = {2102.01828},
  eprinttype = {arXiv},
  eprintclass = {quant-ph},
  pages = {466},
  issn = {2521-327X},
  doi = {10.22331/q-2021-06-04-466},
  langid = {english}
}

@software{ZxcalcQuizx2024,
  title = {Zxcalc/Quizx},
  date = {2024-11-01T06:57:00Z},
  origdate = {2021-01-15T12:08:46Z},
  url = {https://github.com/zxcalc/quizx},
  urldate = {2024-11-27},
  organization = {ZX Calculus Projects}
}

\end{document}